\documentclass[11pt,tikz]{article}

\usepackage[T1]{fontenc}
\usepackage{textcomp}
\usepackage{palatino}
\usepackage{mathpazo}
\usepackage{stmaryrd}
\usepackage{multirow}
\usepackage{wrapfig}

\usepackage{hyperref}
\hypersetup{pdfpagemode=UseNone}

\usepackage{amsfonts}
\usepackage{amssymb}
\usepackage{amsmath}
\usepackage{latexsym}
\usepackage{amsthm}
\usepackage{eepic}
\usepackage{color}

\usepackage[shortlabels]{enumitem}

\usepackage{mathdots}
\usepackage{tikz,caption}

\newtheorem{theorem}{Theorem}
\newtheorem{proposition}[theorem]{Proposition}
\newtheorem{lemma}[theorem]{Lemma}
\newtheorem{definition}[theorem]{Definition}
\newtheorem{fact}[theorem]{Fact}
\newtheorem{corollary}[theorem]{Corollary}

\newtheorem{theorem*}{Theorem}
\newtheorem{corollary*}{Corollary}

\newtheorem{remarkx}[theorem]{Remark}
\newenvironment{remark}{\begin{remarkx}}{\qed \end{remarkx}}

\def\Tr{\textnormal{Tr}}

\def\Im{\textnormal{Im}}
\def\Re{\textnormal{Re}}

\def\Supp{\textnormal{Supp }}

\def\wg{{\mathbf w}_G}
\def\<{\langle}
\def\>{\rangle}

\def\Prob{\mathbb{P}}

\def\D{\mathcal{D}}

\def\H{\mathcal{H}}

\def\P{\mathcal{P}}

\def\S{\mathcal{S}}

\def\L{\mathcal{L}}

\newcommand{\commentout}[1]{}

\numberwithin{theorem}{section}
\numberwithin{equation}{section}

\newenvironment{proofof}[1]{\begin{trivlist}%
\item[]{\flushleft\textsc{Proof of #1.} }}
{\qed\end{trivlist}}

\newcommand{\talkingPoint}[1]{}
\newcommand{\YtalkingPoint}[1]{}

\title{Robust protocols for securely expanding randomness and distributing keys using untrusted quantum devices}

\author{%
   Carl~A.~Miller and Yaoyun~Shi  \\
 \\
  Department of Electrical Engineering and Computer Science\\
  University of Michigan, Ann Arbor, MI 48109, USA\\
  \texttt{{carlmi,shiyy@umich.edu}}
}

\date{\today}

\begin{document}
\maketitle
\thispagestyle{empty}

\commentout{
\begin{center}
\fbox{\color{red}\parbox{6in}{\begin{center} \huge UNPUBLISHED - NOT FOR DISTRIBUTION \end{center}}}
\end{center}
}

\begin{abstract}
Randomness is a vital resource for modern day information processing,
especially for cryptography. A wide range of
applications critically rely on abundant, high quality random numbers
generated securely. Here we show how to expand a  random seed
at an exponential rate without trusting the underlying quantum devices.
Our approach is secure against the most general adversaries, and has the following
new features: cryptographic level of security, tolerating a constant level 
of imprecision in the devices, requiring only a unit
size quantum memory per device component for the honest implementation, and allowing a large natural class of constructions for the protocol. 
In conjunct with a recent work by Chung, Shi and Wu,
it also leads to robust unbounded expansion using just 2 multi-part devices.
When adapted for distributing cryptographic keys, our method achieves, for the 
first time, exponential expansion combined with cryptographic security
and noise tolerance.
The proof proceeds by showing that the R{\'e}nyi divergence of the outputs
of the protocol (for a specific bounding operator) decreases linearly as
the protocol iterates.
At the heart of the proof are 
a new uncertainty principle on quantum measurements,
and a method for simulating trusted measurements
with untrusted devices.
\end{abstract}

\newpage 

\tableofcontents

\section{Background and Summary of Results}
\label{backgroundsec}

\subsection{The Problem and Its Motivations}
Randomness is an indispensable resource for  modern day information processing. 
Without randomness, there would be no fast randomized algorithms, accurate statistical scientific simulations, fair gaming, or secure cryptography. A wide range of
applications rely on methods for generating randomness with high quality and in a large quantity. Consider, for example,
all the computers and handheld devices that connect to the Internet using public key cryptography such as RSA and DSA
for authentication and encryption, and that use secret key cryptography for secure connections. It is probably conservative to estimate
that the number of random bits used each day for cryptography is in the order of trillions. 

While randomness seems to be abundant in everyday life, its efficient and secure generation is a difficult
problem. A typical random number generator such as the \texttt{/dev/random/} generator in Linux  kernel,
would start with random ``seeds'', including the thermal noise of the hardware (e.g. from Intel's Ivy Bridge processors),
system boot time in nanoseconds, user inputs, etc., and apply a deterministic function to produce required random bits. Those methods suffer from at least three
fundamental vulnerabilities.  

The first is due to the fact that no deterministic procedure can increase randomness.
Thus when there is not enough randomness to start with, the output randomness is not sufficient to guarantee security. 
In particular, if the internal state of the pseudorandom generator is correctly guessed or is exposed for other reasons,
the output would become completely predictable to the adversary. The peril of the lack of entropy has been demonstrated repeatedly~\cite{GuttermanPR06,Ristenpart:2010,Lenstra+}. \cite{Heninger:2012} were able to break the DSA secret keys of over 1\% of the SSH hosts that they scanned on the Internet, by exploiting the insufficient randomness used to generate the keys.

The second vulnerability is that the security of current pseudorandom generators are  not only
based on unproven assumptions, such as the hardness of factoring the product of two large primes, but also
assume that their adversaries have limited computational capability.
Therefore, they will fail necessarily if the hardness assumptions turn out to be completely false,
or the adversaries gain dramatic increase in computational power, such as through developing quantum computers. 

Finally, all those methods rely on  trusting the correctness and truthfulness of the generator. The dynamics of market
economy leads to a small number of vendors supplying the hardware for random number generation.
The demand for platform compatibility results in a small number of generating methods.
Thus the risk of the generators containing exploitable vulnerabilities or secret backdoors is necessarily significant.
Recent evidence suggest that this is in fact the reality~\cite{NYT:Snowden}. Thus for users demanding the
highest level of security with the minimum amount of trust, no current solution is satisfactory.

Quantum mechanics postulates true randomness, thus provides a promising approach for mitigating those drawbacks.
Applying a sequence of quantum operations
can increase entropy even when the operations are applied deterministically, as
some quantum operations are inherently unpredictable.
Indeed, commercial random number generators based on quantum technology have started to emerge
(e.g. by ID Quantique SA). Furthermore, the randomness produced can be unconditionally secure, i.e.
without assumptions on the computational power of the adversary.

However, as classical beings, users of
quantum random number generators cannot be certain that the quantum device --- the quantum state inside
and the quantum operations applied --- is running
according to the specification. How can  a classical user 
ensure that a possibly malicious quantum device is working properly?

Non-local games --- games with multiple non-communicating
players --- provide such a possibility. Consider, for example,
the celebrated GHZ game~\cite{GHZ} illustrated in Fig.~(\ref{fig:GHZ}).  It is now known (\cite{mck:2010}, see also \cite{MillerS:self-testing:2013}) that any quantum strategy achieving close to the optimal quantum winning
probability must be very close to the optimal strategy itself. Consequently, the output of
each component is near perfectly random. Intuitively, one needs only to run the game multiple times (using some initial randomness to choose the input string for each round) and if the observed winning average is close to the optimal quantum winning probability, 
then the output should be sufficiently random.
Therefore, the trust on the quantum device can now be replaced by the
condition of non-communication between the different components. This
condition can be verified through classical means, e.g., by separating
the components at a distance so that they do not have time to communicate.
\commentout{
Note that for any cryptographic protocol to
be secure, it is necessary to assume that the components involved 
cannot communicate with the adversary. Thus throughout the paper we
make this assumption on the untrusted device(s) in use.
Since intra-device communications involve different pairs of components,
we treat its restriction a separate item for enforcement or trust.
Special relativity provides a classically verifiable method for enforcing such a restriction.
It may not be necessary as there are potentially other alternatives for achieving this goal,
or may not be feasible if the duration of no-communication is too long for the distance
to be practical.
}

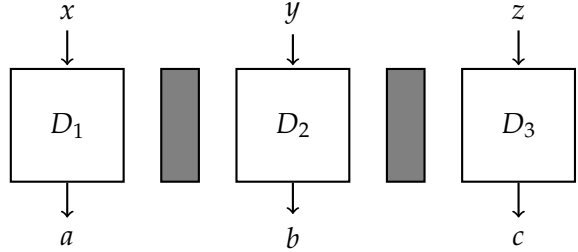
\begin{figure*}
\begin{center}
\begin{tikzpicture}[shorten >=1pt,node distance=1.5cm,thick,scale=1, every node/.style={scale=1}]

\node (x) {$x$};
\node (v1) [right of = x] {};
\node (y) [right of = v1] {$y$};
\node (v2) [right of = y] {};
\node (z) [right of = v2] {$z$};
\node [rectangle, minimum size=1.5cm,draw] (D1) [below of = x] {$D_1$};
\node [rectangle, minimum width=.5cm, minimum height = 1.5cm, draw,fill=gray] (brick1) [below of = v1] {};
\node [rectangle, minimum size=1.5cm,draw] (D2) [below of = y] {$D_2$};
\node [rectangle, minimum size=.5cm,minimum height = 1.5cm,draw,fill=gray] (brick1) [below of = v2] {};
\node [rectangle, minimum size=1.5cm,draw] (D3) [below of = z] {$D_3$};
\node (a) [below of = D1] {$a$};
\node (b) [below of = D2] {$b$};
\node (c) [below of = D3] {$c$};

\path[->]
(x) edge node {}(D1)
(y) edge node {}(D2)
(z) edge node {}(D3)
(D1) edge node {}(a)
(D2) edge node {}(b)
(D3) edge node {}(c);
\end{tikzpicture}
\end{center}
\caption{A three-part device playing the GHZ game. Each part $D_1$, $D_2$, and $D_3$, receives a single bit 
and outputs a single bit. The input $(x,y,z)$ is drawn uniformly from $\{(0,0,0), (0,1,1), (1,0,1), (1,1,0)\}$.
The device wins if $a\oplus b\oplus c = x\vee y\vee z$. No communication among the parts is allowed when the game starts.
An optimal classical strategy is for each part
to output $1$, winning with $3/4$ probability. An optimal quantum strategy is for the three parts
to share the GHZ state $\frac{1}{\sqrt{2}}(|000\rangle+|111\rangle)$, and for each part to measure
$\sigma_x$ on input $0$, and measure $\sigma_y$ on input $1$. This strategy wins with certainty.
}\label{fig:GHZ}
\end{figure*}

\cite{col:2006,ColbeckK:2011} proposed using
nonlocal games as the basis for untrusted-device randomness
expansion. Turning the intuition above into rigorous proofs
turns out to be rather challenging. Classical security was proved in~\cite{pam:2010}, \cite{fehr13}, \cite{pironio13}, and in the later work \cite{CoudronVY:2013},
which allowed a very broad class of nonlocal games.
While useful, classical security does not guard against quantum adversaries, thus is inadequate as quantum
computation is becoming a reality.
Furthermore,  an expansion protocol without quantum security cannot be safely composed with other quantum protocols.
\cite{Vazirani:dice} were the first to prove quantum security,
using a protocol that expands the initial seed exponentially.

\subsection{Related Problems}

The randomness expansion problem is closely related
to the problem of quantum key distribution (QKD), where
two parties at a distance wish to establish a common (random) secret 
using a public quantum channel.  Key distribution is a fundamental cryptographic primitive, and also one of the oldest
problems in quantum information~\cite{BB84,ekert91,Mayers01,LC99,Biham:2006,SP00a}.  

Also, untrusted-device randomness expansion is part of the broader area of untrusted-device, or ``device-independent,'' quantum cryptography.
This area of quantum cryptography was pioneered by \cite{MY98a}. It was also developed in parallel by
other researchers, such as \cite{BHK05},
from the perspective of non-locality with inspirations from \cite{ekert91}.
It has now become an important and intensively studied paradigm
for understanding the power and limitations of quantum cryptography.
\commentout{ when the underlying
quantum devices are imperfect or even malicious.}

An important related problem in
untrusted-device cryptography is {\em randomness amplification} \cite{cr:2012}, 
where one wants to obtain near-perfect randomness
from a \textit{weak} random source using untrusted quantum devices (and without any additional randomness).  The paper \cite{CSW14}, which is a companion paper to the present one (with a common author) studies the amplification problem.

\subsection{Overview of Our Results}\label{ourresultssubsec}
In this work, we analyze a simple exponentially expanding untrusted-device randomness expansion protocol (referred to as the {\em one-shot protocol}). We
give a proof of security against the most general
quantum adversaries. More importantly, we accomplish
all of the following additional features, none of which has
been accomplished by previous works.

The first is {\em cryptographic security} in the output.\footnote{We thank Kai-Min Chung and Xiaodi Wu
for pointing out this feature of our result.}  The error parameters are not only exponentially small in the input length, but are
also negligible (i.e. smaller than any inverse polynomial function) in the running time of the protocol
(which is asymptotically the number of uses of the device.) This is the conventional theoretical requirement 
for cryptographic level of security ---   the chance that
an adversary can distinguish the protocol output from an ideal uniform distribution is negligible, as measured against
the amount of resource used for running the protocol.

Secondly, the protocol is {\em robust}, i.e. tolerating a constant level of ``noise'', or implementation imprecision. 
Thus any honest implementation that performs below the optimal level by a small constant amount will still pass our test with overwhelming probability.
For example, we show that any device which wins the GHZ game with probability
at least $0.985$ will achieve exponential randomness expansion
with probability approaching $1$.

\commentout{
We model noise as the average deviation of a device interaction
from an ideal implementation. This ``black-box'' noise model does not refer to the inner-working
of the devices thus is fairly general.
This robustness feature enables simplest --- e.g. without adding fault-tolerant overhead --- practical implementation,
thus is desirable especially at this moment of time when even small scale accurate quantum information
processing is still being developed. When  fault-tolerant quantum computing techniques become mature,
it is conceivable that non-robust protocols will become practical. However, when the number of uses
of a device is not pre-determined, such as in an unbounded expansion to be discussed below,
one cannot set the target error threshold for the fault-tolerant 
implementation. A robust protocol then remains indispensable.
}

Third, our protocol requires only a {\em constant size quantum memory} for an honest implementation.
In between two rounds of interactions, the different components of the device are allowed to
interact arbitrarily. Thus an honest device could 
establish its entanglement on the fly, and
needs only to maintain the entanglement (with a constant level of fidelity) for the duration of a single game.
Given the challenge of maintaining coherent quantum states, this
feature greatly reduces implementation complexity.\footnote{An alternative
for achieving the small quantum memory requirement
is to introduce an additional device component that is required
to function as an entanglement creation and distribution component and cannot receive information
from other device components. This model would require a communication restriction throughout the protocol.
}

Fourth, relying on a powerful observation of Chung, Shi and Wu~\cite{CSW14} --- what they call the Equivalence Lemma ---
we show that one can sequentially compose instances of our one-shot protocol,
alternating between two untrusted devices and achieve {\em unbounded} randomness
expansion
starting with a fixed length seed. The additively accumulating error parameters remain almost identical to
the one-shot errors, since they decrease geometrically. 

Finally, our protocol allows a large natural class of games to be used. The class consists
of all binary XOR games ---
games whose inputs and outputs are binary and whose scoring function depends on the inputs and the XOR of the outputs ---
that are {\em strongly self-testing}. The latter property says that any quantum strategy that is
$\epsilon$-close to optimal in its winning probability must be $O(\sqrt{\epsilon} )$ close to a
unique
optimal strategy in its both its state and its measurements.  (We call this
``strongly self-testing'' because this error
relationship is the best possible.)  We explored
this class previously in \cite{MillerS:self-testing:2013}.
The class of strong self-tests includes the CHSH game and the GHZ game,
two commonly used games in quantum information.
Broadening the class of usable games has the benefit
of enabling greater design space, as different implementation
technologies may favor different games. (For example, the highly accurate topological quantum computing
approach using Majorana fermions is known not to be quantum universal~\cite{Nayak:2008}. In particular,
Deng and Duan~\cite{Deng:Majorana:2013} showed that for randomness expansion using Majorana fermions, three qubits
are required. Our proof allows the use of Majorana fermions for randomness expansion
through the GHZ game.)

We include two applications of our expansion protocols.  Our protocol can
be used in combination with the randomness amplification
results of \cite{CSW14} to create a robust, untrusted-device quantum protocol that converts an arbitrary
weak random source into near-perfect output randomness of an arbitrary large length.  
This opens the possibility for
unconditionally secure cryptography with the minimum trust on the randomness source and the implementing device.
The second application is to adapt our protocol for untrusted-device quantum key distribution, resulting
in a robust and secure protocol that requires only a small (polylogarithmic)
initial seed. 

\subsection{Related Works}
Prior to our paper, the groundbreaking work
of~\cite{Vazirani:dice} was the first and only
work achieving simultaneous exponential expansion and quantum security. 
As far as we know from
their analysis, their security proof achieves only inverse polynomial security, and thus is not cryptographically secure; it is not noise-tolerant (as it requires perfect behavior
on some rounds); and it also does not have the feature of constant-sized
quantum memory.

Robust DI-QKD was already achieved with full security in \cite{Vazirani:fully}.
(There
were also previous non-robust proofs~\cite{BCK,ruv:2013}
and proofs that require that the number of devices
increases with the length of the key, e.g., \cite{HRW10}, \cite{mpa11}.)
The new feature offered by our QKD result is that our seed is polylogarithmic, while that
of \cite{Vazirani:fully} is linear.  

The paper~\cite{Vazirani:fully} on untrusted-device QKD
can be considered as a robust
randomness expansion protocol with a {\em linear} rate of  expansion
(without the constant memory feature).
A natural way to develop \cite{Vazirani:fully} further as an expansion
result would be to change the input distribution to
one that is non-uniform (so as to require less than a linear seed)
and to apply the proof to a more general class of games
(such as those of \cite{CoudronVY:2013}).  
To
our knowledge a formal analysis
of these generalizations has not yet been published, and they are
a topic
for further research.
\commentout{
Similarly, as presented and analyzed the protocol
does not allow inter-component communication in-between rounds. 
It is not clear (when used as a linear expansion protocol) whether their analysis can be strengthened
to allow some in-between-round communications. }

\commentout{We note that while the Vazirani-Vidick protocol for key distribution protocol uses a specific game, the analysis remains valid for a much broader class of games, such as the "randomness generating" games defined in~\cite{CoudronVY:2013}, as pointed out to us by Vidick~\cite{Vidick:PC}.
Since all strong self-testing binary XOR games are randomness generating, the class of games allowed for untrusted-device key distribution is broader in~\cite{Vazirani:fully} than ours.
The class of games allowed in~\cite{Vazirani:dice}
require more special properties, and the class of games allowed there appear to be incomparable with ours.
Note that it is not clear if this latter comparison is meaningful as our protocol achieves several additional features not in~\cite{Vazirani:dice}.}

\cite{CY} did contemporaneous work on the problem of
unbounded randomness expansion.  Their paper was the first to prove that
(non-robust) unbounded expansion is possible with
a \textit{constant} number of devices.   We
independently proved that robust expansion
is possible with $\log^*  ( N )$ devices.  After we learned
of their work we observed that a result with both
features --- robustness and a constant number of devices ---
follows by combining results from
our work and \cite{CSW14}.  We discuss this more in the
next subsection (see the remarks
that follow Corollary~\ref{co:unboundedNaive}).

\subsection{Technical Statements}
\label{statementressubsec}

Our main protocol (Figure~\ref{protocolrtable}) is based on \cite{pam:2010} and
\cite{CoudronVY:2013}.  (Indeed, it is only a slight modification of a protocol
from the classical security results of \cite{CoudronVY:2013} --- the main differences are the class
of games that we use, and most importantly,  that
we explicitly allow in-between-rounds quantum communication.)
We use the idea from \cite{pam:2010} to conserve seed by
giving a fixed input to the device on most rounds.

The games we use involve $n$ parties, with $n \geq 2$.  Such a game
is played by a single \textit{device}, which consists of $n$ \textit{components}, where
each component has a classical input/output interface.\footnote{We note that
the literature
on this subject has some differences in terminology.  Some authors
would use the word ``device'' in the way that we have used the
word ``component.''}  For any game that we use, we let $\mathbf{w}_G$ denote the highest winning probability which can be achieved by a quantum strategy,
and let $\mathbf{f}_G = 1 - \mathbf{w}_G$ denote the smallest possible
failure probability that can be achieved by a quantum strategy.


\begin{figure}
\begin{center}
\setlength{\fboxsep}{10pt}
\setlength{\fboxrule}{1pt}
\fbox{\parbox{\linewidth}{

\begin{minipage}{.45\linewidth}
\textit{Arguments:}
\begin{enumerate}[topsep=3pt]
\item[$N:$] a positive integer (the {\bf output length}.)
\item[$\eta:$] A real $\in(0,\frac{1}{2})$. (The {\bf error tolerance}.)
\item[$q:$] A real $\in(0,1)$.  (The {\bf test probability}.)
\item[$G:$] An $n$-player nonlocal game that is a {\bf strong self-test} \cite{MillerS:self-testing:2013}.
\item[$D:$] An untrusted device (with $n$ components) that can play $G$ repeatedly and cannot receive any additional information. In a single use the different components cannot communicate; in between
uses, there is no restriction.
\end{enumerate} 
\end{minipage}
\hskip0.2in
\begin{minipage}{.5\linewidth}
\begin{center}
\scalebox{0.8}{
\begin{tikzpicture}[shorten >=1pt,node distance=1.5cm,thick,scale=1, every node/.style={scale=1}]

\node [rectangle, minimum width=.5cm, minimum height = 1.5cm, draw,fill=gray] (brick) {};
\node [rectangle, minimum size=1.5cm,draw] (D1) [left of = brick] {$D_1$};
\node [rectangle, minimum size=1.5cm,draw] (D2) [right of = brick] {$D_2$};
\node [rectangle, minimum width=.5cm, minimum height = 1.5cm, draw,fill=gray] (brick2) [right of = D2] {};
\node [rectangle, minimum size=1.5cm,draw] (Dn) [right of = brick2] {$D_n$};
\node [circle, draw] (i) [above of = D2,yshift=2] {};
\node (g) [left of = i, xshift=-2.5cm,yshift=.5cm] {$g$};
\node [circle, draw] (o) [below of = D2,yshift=-2] {};
\node (r) [left of = o, yshift=-1cm, xshift=-2.5cm] {};

\path[->]
(i) edge node {} (D1.north)
(i) edge node {} (D2.north)
(i) edge node {} (Dn.north)
(D1.south) edge node {} (o)
(Dn.south) edge node {} (o)
(D2.south) edge node {} (o);

\draw [<->, densely dotted, red, thick, minimum width=.25cm] ([yshift=-.25cm] D1.south) -- ([yshift=-.25cm] Dn.south);
\draw [<->, densely dotted,red,thick] (D1.east) -- ([xshift=.25cm] D1.east) -- ([xshift=.25cm, yshift=-1cm] D1.east);
\draw [<->, densely dotted,red,thick] (D2.east) -- ([xshift=.25cm] D2.east) -- ([xshift=.25cm, yshift=-1cm] D2.east);
\draw [<->, densely dotted,red,thick] (Dn.west) -- ([xshift=-.25cm] Dn.west) -- ([xshift=-.25cm, yshift=-1cm] Dn.west);

\draw[->] ([yshift=.5cm,xshift=-3.75cm] i.north) -- ([yshift=.5cm] i.north) -> (i);
\draw[->] (g) -- ([xshift=-3.75cm] o.west) -> (o);
\draw[->] (o.south) -- ([yshift=-.5cm] o.south) ->  node [below] {result} ([yshift=-.5cm, xshift=-4cm] o.south); 

\end{tikzpicture}
}
\captionof*{figure}{A diagram of Protocol R.  The dotted red lines denote in-between-round communications.}
\label{fig:R}
\end{center}
\end{minipage}
\medskip

\textit{\bf Protocol $R$:}

\begin{enumerate}[topsep=3pt]
\item A bit $g \in \{ 0, 1 \}$ is chosen
according to a biased $(1 - q, q )$ distribution.


\item If $g = 1$ (``game round''), then an input string is
chosen at random from $\{ 0, 1 \}^n$ (according
a probability distribution specified by $G$) and
given to $D$.
Depending on the outputs,
a``P'' (pass)  or an ``F'' (fail) is recorded
according to the rules of the game $G$.


\item If $g = 0$ (``generation round''), then the input string
$00 \ldots 0$ is given to the device, and the output
of the first component $D_1$ is determined.
If the output of the first component
is $0$, the event $H$ (``heads'') is recorded;
otherwise the event $T$ (``tails'') is recorded.


\item Steps $1-3$ are repeated $N-1$ (more) times.


\item If the total number of failures exceeds
 $\left( 1 - \mathbf{w}_G + \eta \right)q N$, the protocol
\textbf{aborts}.    
Otherwise, the protocol \textbf{succeeds}.
If the protocol succeeds, the output consists
of an $N$-length sequence from the alphabet
$\{ P, F, H, T \}$ representing the outcomes of each round.
\end{enumerate}
}}
\caption{The Central Protocol R}
\label{protocolrtable}
\end{center}
\end{figure}

We discuss some the concepts necessary to evaluate the security of Protocol R.
We measure the amount of randomness produced by the
 quantum min-entropy $H_{min} (
X | E )$, where $X$ denotes the output of the protocol and $E$ denotes the (possibly quantum) information possessed 
by an adversary.  This quantity is appropriate because it measures
the amount of uniformly random bits that can be extracted from $X$
by a randomess extractor (see Chapter 5 of \cite{R05}).

\commentout{The protocol is meant to be used in conjunction with a randomness
extractor --- that is, a deterministic function $\mathrm{Ext} ( X, S )$
on two arguments, $X$ being the
{\em source}, which in this case
is the output of the protocol, and $S$ being a short perfectly random {\em seed}.
It is known that there are quantum-proof randomness extractors $\mathrm{Ext}(X, S)$ that will convert any $N$ bits $X$ that have min-entropy $\Omega(N)$ to near perfect output randomness
of length $\Theta(N)$. }

Let $y\ge0$.  We will say that a subnormalized classical-quantum state
$\rho$ is
{\em $y$-ideal} if its normalization $\rho / \Tr ( \rho )$ has conditional min-entropy greater than or equal to $y$.  (For
convenience, we will say that the zero state is $y$-ideal for all $y$.)
Let $\epsilon_s, \epsilon_c, \lambda$ be reals in $[0, 1]$. 
A randomness expansion protocol is said to have a
{\em yield of $y$ extractable bits} with  a {\em soundness error} $\epsilon_s$ if for any device $D$, and any purifying system $E$ for $D$, the state of 
$(X, E)$ corresponding to the ``success'' event
is always within trace distance $\epsilon_s$ of a $y$-ideal state. 
It is said to have a  {\em completeness error} $\epsilon_c$ with 
noise level $\lambda$ if
there exists an implementation, referred to as the ``correct'' implementation,
so that for any implementation which
deviates by no more than $\lambda$ from the correct implementation, the probability of aborting
is always $\le \epsilon_c$.
If both the soundness and the completeness errors are $\le\epsilon$,
we simply say the protocol has an error $\epsilon$. 

\commentout{
\begin{remark} Our definition of completeness error is quite strong and some weaker notions
may be more desirable in some settings. For example, changing the input state
considered from ``any'' input state to those from a set $\S$ of ideal input states, we
have a notion of completeness with respect to $\S$. Yet another alternative is to
measure error by distance to a normalized ideal output state.  Those quantities are closely related thus our results
can be modified correspondingly.
For example, if an output state is within
$\epsilon$ trace distance to an ideal output state, the chance of aborting is $\le\epsilon$.
Conversely, if the probability of aborting
is $\le \epsilon$, the closest ideal output state is within $2(\epsilon_s+\epsilon)$ trace
distance. 
\end{remark}
}

\vbox{
Our main result is the following. 

\begin{theorem}[Main Theorem]   For any $n$-player strong self-test $G$, and any $\delta > 0$,
\label{mainthmprelim}
there exist positive constants $q_0, \eta_0, K,  b$, such that the following hold
when Protocol R is executed with parameters
$q \leq q_0$, $\eta \leq \eta_0$.
\begin{enumerate}
\item \textnormal{(Soundness.)} The yield is at least 
$(1 - \delta ) N$ extractable bits with a soundness error $\epsilon_s = K \exp(-b qN)$.

\item \textnormal{(Completeness.)} For any constant $\eta'$, $0<\eta'<\eta$, 
the protocol tolerates $\eta'$ noise level with a completeness
error $\epsilon_c=\exp(-(\eta-\eta')^2 qN/3)$.
\end{enumerate}
\end{theorem}
}

The difficult part of this result is the soundness claim, which
follows from the results of section~\ref{REUDsec} 
(see Corollary~\ref{highratecor}).
The completeness claim follows from the Azuma-Hoeffding
inequality, and is proved in Proposition~\ref{completenessprop}.

Note that the bits $g_1, \ldots, g_N$ can be generated
by $O ( N h ( q ) )$ uniformly random bits with an error $\exp(-\Omega(qN))$, where $h$ denotes the Shannon entropy function.
Therefore, when $q$ is chosen to be small, the protocol needs only $\omega ( \log N )$ initial bits and one device
to achieve $\Omega ( N )$ extractable bits with negligible error. 

\begin{corollary}[One-shot Min-entropy Expansion]\label{co:oneshot} For any real $\omega\in(0,1)$, 
setting $q=\Theta(k^{\omega}/2^{k^{1-\omega}})$
in Theorem~\ref{mainthmprelim}, Protocol R converts any $k$ uniform bits to $2^{k^{1-\omega}}$ extractable bits with $\exp(-\Omega(k^{\omega}))$ soundness and completeness errors.
\end{corollary}

To obtain near perfect random bits, we apply a quantum-proof strong randomness extractor,
in particular one that extracts a source of a linear amount of conditional quantum min-entropy. 
The parameters of our protocols depend critically on the seed length of such extractors, thus we introduce
the following definition.
\begin{definition}[Seed Length Index] We call a real $\nu$ a {\em seed length index} if
there exists  a quantum-proof strong extractor
extracting $\Theta(N)$ bits from a $(N, \Theta(N))$ source with error parameter $\epsilon$ using
$\log^{1/\nu}(N/\epsilon)$ bits of seed. Denote by $\mu$ the supremum of all seed length indices.
\end{definition}
Such extractors exist with $\nu\ge1/2$, e.g., Trevisan's extractors~\cite{Trevisan:2001}
shown to be quantum-proof by De {\em et al.}~\cite{De:2012}. Thus $\mu\ge 1/2$.
The definition of soundness error for producing $y$ bits of perfect randomness
is the same for producing extractable random bits, except that the ideal C-Q state conditioned on Success
is the product state of $y$ perfectly random bits and a quantum state. The following
corollary follows directly by composing protocol R and an extractor with $\nu$ close to $\mu$.

\begin{corollary}[One-shot Randomness Expansion]\label{co:oneshotRand} For any $\omega\in(0,\mu)$, 
setting $q=\Theta(k^{\omega}/2^{k^{\mu-\omega}})$
in Theorem~\ref{mainthmprelim}, Protocol R composed with an appropriate quantum-proof strong extractor
converts $k$ bits to $2^{k^{\mu-\omega}}$ uniform bits with soundness and completeness
errors $\exp(-\Omega(k^{\omega}))$.
\end{corollary}

The next corollary addresses cryptographic security.  
(Note: In measuring running time, one round of interaction with the device is considered a unit time.)

\begin{corollary}[Cryptographic Security] With the parameters in Corollary~\ref{co:oneshotRand},
the running time of the protocol is $T:=\Theta(2^{k^{\mu-\omega}})$. Thus for any $\lambda>1$, setting
$\omega=\frac{\lambda}{1+\lambda}\mu$, the errors are $\exp(-\Omega(\log^{\lambda} T))$,
which are negligible in $T$. That is, the protocol with those parameters achieves
cryptographic quality of security (while still exponentially expanding.)
\end{corollary}

Once we have near perfect randomness as output,
we can use it as the input to another instance of the protocol, thus expanding further
with an accumulating error parameter. As the error parameters
decrease at an exponential rate, they are dominated by the first set of errors.

\begin{corollary}[Robust Unbounded Randomness Expansion]\label{co:unboundedNaive} 
For all integers $N$ and $k$, and any real $\omega\in (0, \mu)$, $k$ uniformly random bits can be expanded to $N$ output bits with $\exp(-\Omega(k^{\omega}))$ error under a constant level of noise. The procedure uses 
$O( \log^* N)$ iterations of Protocol R using $O(\log^*N)$ devices.
\end{corollary} 

To decrease the number of devices used in unbounded expansion, a possibility
(used, e.g., in \cite{fehr13}) is to cross-feed the outputs of two devices
(i.e., give the output of one device as the input to another, and then vice
versa).  But there is an apparent obstacle for proving security for such an approach: once a device produces output, this output is now correlated
with the device itself.  When this output is fed to a second device to produce
new output, one needs to show
that the correlation with the first device does not persist.
(If it did, then at the third iteration one would be feeding the first device a seed that was
correlated with the first device itself, thus causing an insecurity.)

\cite{CY} call this
the \textit{input security problem}, and solve it 
by an improved analysis of the Reichardt-Unger-Vazirani protocol~\cite{ruv:2013}.
Under this new analysis, the RUV protocol 
turns a uniform-to-device input into a globally-uniform  output with sufficiently strong 
parameters. By interleaving the RUV protocol with
the exponentially expanding protocol of \cite{Vazirani:dice},
they prove non-robust unbounded expansion with 4 two-part devices.

An independent result of \cite{CSW14} can be used to address this problem in a
different way.  The Equivalence Lemma of~\cite{CSW14} states that if a randomness
expansion protocol is secure with a globally random input,
then it is also automatically secure with any uniform-to-device input. 
This means that the correlation of each device with its own output does
not cause a problem.  Consequently, unbounded expansion with $2$ devices
can be achieved by cross-feeding {\it any} secure randomness expansion protocol.

We therefore have the following corollary (which is subsequent to \cite{CY},
though based on independent techniques).  See section~\ref{sec:unbounded}.

\begin{corollary}[Robust Unbounded Randomness Expansion with 2 Devices]\label{co:unbounded} 
The number of (multi-part) devices used in Corollary~\ref{co:unboundedNaive} can be reduced to $2$.
\end{corollary} 

To apply our protocol to randomness amplification, we can use
the results of \cite{CSW14}.  The amplification protocol in \cite{CSW14}
requires having a robust randomness certification procedure
to call as a subroutine; for this, we can use Protocol $R$ with
$q = \Theta (1)$.  
The  amplification protocol converts a $n$-bit, min-entropy $\ge k$
weak source to a near perfectly random output of $\Theta(k)$ bits. Then we can concatenate with the protocol of Corollary~\ref{co:unbounded}
to expand to an arbitrarily long near perfect randomness.  (Here the improvement
from Corollary~\ref{co:unboundedNaive} to Corollary~\ref{co:unbounded} 
implies that the number of devices need not depend on the output length.)

\commentout{While the QKD protocol of Vazirani and Vidick~\cite{Vazirani:fully}
can also serve as a building block for the Chung-Shi-Wu protocol, the use of our protocol has two advantages
for in-between-rounds communications and flexibility of the non-local game used.}

\begin{corollary}[with \cite{CSW14} --- Randomness Amplification] \label{co:randamp}
Let $\nu\in[1/2, \mu]$ be a seed length index.
For all sufficiently large integer $k$, any integer $n = \exp(O(k^{\nu^2}))$,
any real $\epsilon=\exp(-O(k^{\nu^2}))$,
any $(n, k)$ source can be converted to an arbitrarily long near perfect randomness
with $\epsilon$ soundness and completeness errors under a (universal) constant level
of noise. The number of devices used is $2^{O(\log^{1/\nu} (n/\epsilon))}$, which in particular does not depend on the output length.
\end{corollary}

We point out that the number of devices $T=T(n,1/\epsilon)$ used as a function of the weak source length $n$ and 
the error parameter $\epsilon$ grows super-polynomially (if $\mu<1$) or polynomially (if $\mu=1$).
It remains a major open problem if $T(n,1/\epsilon)$ can be substantially reduced or even be made a universal constant.
We stress, however, that the limitation imposed by this function is better interpreted as limiting the achievable error,
instead of computational efficiency. This is because, $T$ could still scale efficiently as a function of the {\em output} length.
For example, to output $N$ bits, as long as $\epsilon=\exp(-O(\log^{\nu} N))$, the number of devices is still polynomial in
$N$.  Therefore, the question of improving $T$ is the question of broadening the application of the combined amplification-expansion
protocol to settings requiring inverse-polynomial or even cryptographic quality of error.

Lastly, we state our result on quantum key distribution.
Suppose that Alice and Bob would like to establish a shared secret string in an environment where trusted randomness is a scarce resource,
and consequently their initial randomness is much shorter than the desired output length.
(As with other studies on quantum key distribution (QKD), we will sidestep the authentication issue, assuming that the man-in-the-middle attack is already dealt with.)
One way to adapt our randomness expansion protocol for untrusted-device QKD scenario is for Alice to expand her initial randomness,
then use the expanded, secure randomness to execute the untrusted device QKD protocol of Vazirani and Vidick~\cite{Vazirani:fully}.
The end result is an exponentially expanding key distribution protocol.
An alternative approach, which is the focus of our new contribution, is to directly adapt our expansion protocol
to achieve simultaneously randomness expansion and key distribution (see Protocol $R_{\textrm{kd}}$
in Fig.~\ref{protocolrqkd}). The benefits of doing so is the reduction of 
the number of untrusted devices from 2 to 1.

We present the details in section~\ref{sec:qkd} and state our main result on key distribution below.
The notion of soundness and completeness errors are similarly defined: the soundness error
is the distance of the output distribution to a mixture of aborting and an output randomness of a desired smooth
min-entropy,
and the completeness error is the probability of aborting for an honest (possibly noisy) implementation.

\begin{corollary}[Robust Untrusted-Device QKD with Short Seed]\label{co:qkd} 
For any strong self-test $G$, there exist positive constants $r, \lambda, \eta, q_0$ such that for infinitely many positive
integers $N$ and
any $q\le q_0$, Protocol  $R_{\textrm{kd}}$~(Fig.~\ref{protocolrqkd}) satisfies the following.

\begin{enumerate}
\item \textnormal{(Soundness.)} The protocol obtains a key of $rN$ extractable bits
with a soundness error
\begin{eqnarray*}
\epsilon_s = \exp(-\Omega(qN)+O(1)).
\end{eqnarray*}
\item \textnormal{(Completeness.)} For any constant $\eta'$, $0<\eta'<\eta$, 
the protocol tolerates $\eta'$ noise level with a completeness
error $\epsilon_c=\exp(-\Omega((\eta-\eta')^2 qN))$.
\end{enumerate}

The number of initial random bits is $O(Nh(q) + \log N)$, and the time complexity is polynomial in $N$.

\commentout{
In particular, for any $\omega\in(0,1)$ and all sufficiently large $k$,
$N':=2^{k^{1-\omega}}$ extractable bits can be distributed using one (multi-part) untrusted quantum device
starting with a seed of $k$ bits, with soundness and completeness errors
$\exp(-\Omega(k^{\omega}))$ and tolerating a constant level of noise. The corresponding $N=\Theta(N')$,
thus the time complexity is $\textrm{poly}(N')$.
}
\end{corollary}

Thus, for example, if we set $q = (\log^2 N ) / N$, we can distribute $\Omega ( N )$
extractable bits using a seed of size $O ( \log^3 N )$, with error terms achieving cryptographic security.
Composing this protocol with a quantum-proof randomness extractor that uses a polylogarithmic seed
\cite{De:2012} yields untrusted-device QKD from a polylogarithmic seed.

\subsection{Numerical Results}
The proof methods in the paper are sufficient
to give actual numerical bounds for the amount of
randomness generated by Protocol R.
In subsection~\ref{ghzsubsec} we offer an example
showing how this is done.
If $G$ is a strong self-test, then there is an associated
quantity $\mathbf{v}_G > 0$ (called the trust coefficient).  Let
\begin{eqnarray}
\pi ( y ) = 1 - 2 y \log \left( \frac{1}{y} \right) -
2 (1 - y ) \log \left( \frac{1}{1-y} \right).
\end{eqnarray}
We show that if $\eta < \mathbf{v}_G/2$, Protocol $R$ produces
$\pi ( \eta/\mathbf{v}_G) N$ extractable bits
per round, modulo error terms (see
Corollary~\ref{rateextractablecor}).
In particular, a positive rate is achieved
provided that $\pi ( \eta / \mathbf{v}_G ) > 0$,
which occurs when $\eta < 0.11 \cdot \mathbf{v}_G$.
Subsection~\ref{ghzsubsec}, shows
that $\mathbf{v}_{GHZ} \geq 0.14$.
Therefore, the GHZ game achieves 
a positive linear rate provided
that $\eta < 0.11 \cdot 0.14 = 0.0154$.

\section{Overview of Proofs}
\label{pfoverviewsec}

While proving classical security of randomness expansion protocols is mainly
appropriate applications of Azuma-Hoeffiding inequality, proving quantum security
is much more challenging. The proof for the Vazirani-Vidick protocol~\cite{Vazirani:dice}
relies on a characterization of quantum smooth min-entropy
based on the quantum-security of Trevisan's extractors~\cite{De:2012}.
We take a completely different approach, without any reference to extractors 
in the main security proof.  Below we summarize some of the tools used in our proof, which
we are hopeful will find applications elsewhere.


\subsection{Quantum R{\'e}nyi Entropies}

We follow previous work \cite{TomamichelCR:2009}, \cite{DupuisFS:2013}
and use the Renyi entropy function $H_\alpha ( \rho )$ and Renyi divergence function
$D_\alpha ( \rho \| \sigma )$ to lower bound the number of extractable bits in a classical
register with quantum side information.  (See subsection~\ref{renyisection}.) 
Crucially, we use the newer definition of the quantum Renyi divergence function (the ``sandwiched'' definition)
which was introduced in \cite{JaksicOPP:2011} and developed in
\cite{MullerDSFT:2013}, \cite{WildeWY:2013}.   

In \cite{TomamichelCR:2009}, the authors prove a lower bound on the conditional smooth min-entropy
of $n$ identical copies of a bipartite system $\rho_{AB}$ in terms of its relative entropy
$H( A \mid B )_\rho$.  They accomplish this by using the Renyi entropy $H_\alpha$ as an intermediate quantity,
exploiting inequalities that relate it to both $H_{min}^\epsilon$ and $H$, and
then using the additive property of $H_\alpha$.  Our proof incorporates a similar line of reasoning: we prove inductively
an upper bound on the Renyi divergence of the outputs of Protocol $R$ (conditioned on the adversary),
and then use this to compute a lower bound for the same outputs expressed in  terms of smooth min-entropy.

An challenge in our proofs
is choosing the right parameter $\alpha$.  If $\alpha$
is too close to $1$, the penalty term in the inequality that
relates $H_\alpha$ to $H^\epsilon_{min}$ will be large enough to make the lower bound
on smooth min-entropy useless; but if $\alpha$ is too far
from $1$, the Renyi entropy is not sensitive enough
to detect the effect of rare events, such as the game rounds
in Protocol $R$.  The parameter $\alpha$ is therefore
adjusted according to parameters in the Protocol R --- roughly
speaking, it is set so that $\alpha - 1$ is proportional
to the parameter $q$.  

Our first original result (Theorem~\ref{uncertaintythm})
is an Renyi entropy uncertainty principle for measurements on an entangled qubit.  If $QE$ is
a bipartite system where $\dim Q = 2$, let $\{ \rho_0, \rho_1 \}$
and $\{ \rho_+ , \rho_- \}$ denote the subnormalized states
of $E$ that arise from measuring the computational basis
and the Hadamard basis on $Q$, respectively.  Theorem~\ref{uncertaintythm}
expresses uniform constraints (independent of the dimension of $E$)
on the quantities $\Tr [ \rho_x^{1+\epsilon} ]$.  This parallels 
other known uncertainty relations \cite{WehnerW:2010}.
The  proof is based on a known
matrix inequality for the $(2+2\epsilon)$-Schatten norm.

\subsection{Partially Trusted Measurement Simulation}

A key insight which enables our proof is that untrusted devices
can be used to simulate \textit{partially} trusted measurements.
Let us say that a \textit{device with trusted measurements} $F$ is a 
single-part input-output device which receives a single bit as an input,
and, depending on the value of the bit, performs one of two perfectly
anti-commutative binary measurements on a quantum system.  
The measurements of the device are trusted, but
the state is unknown.  Now consider another 
single-part binary device $F'$
which performs as follows (for some real parameters
$v, h$):
\begin{enumerate}
\item On input $0$, $F'$ performs the same measurement as $F$.

\item On input $1$, one of the following occurs at random:
\begin{enumerate}
\item $F'$ performs the same measurement
as $D$ (probability = $v$);
\item $F'$ outputs a perfectly coin flip (probability = $h$);
\item $F'$ performs an unknown measurement (probability = $1 - v - h$).
\end{enumerate}
\end{enumerate}
The device $F'$ is what we will call a \textit{partially trusted} device
(see Definition~\ref{partiallytrusteddef} for a formal definition).

Consider the state of the device $D$
after steps 1--3 in Protocol R.  Let $G_1$ be a classical
register containing the bit $g$, and let $O_1$ be a classical register
which we set to be $0$ if the output is $P$ or $H$, and $1$
if the output is $F$ or $T$.  We show (sections \ref{gamessec}--\ref{REpartialsec}) that the joint state of $G_1 O_1$ can be simulated by a partially trusted
device $D'$ which accepts $G_1$ as its input and produces $O_1$ as
its input.  (Here, ``simulation'' means that if either device
is prepared with an intial purifying system $E$, the joint state
$EG_1O_1$ will be the same up to isomorphism regardless
of which device was used.)

We define a new protocol (Protocol A', Figure~\ref{protaprimefig}) which
is essentially Protocol R with its device replaced by a single-part
partially trusted device.  Proving the security of Protocol R reduces to 
proving the security of Protocol A'.

\subsection{An Induction Proof With a Weighted Measure of Randomness}

The next step is to prove the security of Protocol $A'$.
Let $G = (G_1, \ldots, G_N)$ and $O = (O_1, \ldots, O_N)$ denote
registers containing the input bits and output bits, respectively,
from Protocol $A'$, and let $E$ denote a purifying system for the
device in Protocol $A'$.  Let $\Gamma_{EGO}^s$ denote the subnormalized
state of these three systems corresponding to the ``success'' event ($s$).
Our approach is to prove an upper bound on the (negative) quantity
\begin{eqnarray}
\label{quantitytobound}
D_\alpha \left( \Gamma^s_{EGO} \|  \Gamma_{EG} \otimes \mathbb{I}_O \right).
\end{eqnarray}

Another central insight for our proof is the idea of using a weighted measure
of randomness.  Consider the first-round registers $G_1$ and $O_1$,
and $E$.  The bounding
operator $\Gamma_{E G_1} \otimes \mathbb{I}_{O_1}$ on
$E G_1 O_1$ is equal to
\begin{eqnarray*}
(1-q) \Gamma_E \otimes \left| 00 \right> \left< 00 \right| + 
(1-q) \Gamma_E \otimes \left| 01 \right> \left< 01 \right| + (q) \Gamma_E \otimes \left| 10 \right> \left< 10 \right| + 
(q) \Gamma_E \otimes \left| 11 \right> \left< 11 \right|
\end{eqnarray*}
Let $\lambda > 0$ be a real parameter, and consider the following
alternative operator, where we have inserted the
factor $2^\lambda$ in the fourth summand:
\begin{eqnarray*}
 \Sigma & := 
& (1-q) \Gamma_E \otimes \left| 00 \right> \left< 00 \right| + 
(1-q) \Gamma_E \otimes \left| 01 \right> \left< 01 \right| 
+ (q) \Gamma_E \otimes \left| 10 \right> \left< 10 \right| + 
(q) 2^\lambda \Gamma_E \otimes \left| 11 \right> \left< 11 \right|.
\end{eqnarray*}
The factor $2^\lambda$ artificially adds randomness when
the event $(g, o) = (1, 1)$ (which corresponds to a game-loss
in Protocol $R$) occurs.  Effectively, we lower our expectation
for randomness according to how well the device is performing.

Our uncertainty principle for Renyi entropy implies that, for appropriate
$\lambda, \alpha$, the quantity $D_\alpha ( \Gamma_{EG_1O_1} \|
\Sigma )$ has a uniform upper bound less than zero.  This enables an induction proof
which shows an upper bound on (\ref{quantitytobound}).

A version of this argument is carried out in section~\ref{partialtrustapp}.
We deduce a lower bound on the number extractable bits output by Protocol
$A'$.  By the reduction discussed above, this implies a lower bound
on the number of extractable bits output by Protocol $R$ (see
section~\ref{REUDsec},
Corollary~\ref{noisyratecor}).

\subsection{Quantum Key Distribution}

Proving quantum distribution requires first showing that
when the noise tolerance in Protocol $R$ is set sufficiently
low, and the protocol succeeds, then the device must
score well not only during game rounds but
also during generation rounds.  This is accomplished
using Azuma's inequality.  A  consequence is that if two parties
possess different subsets of the components of the device $D$,
they can use these devices to construct strings of length $N$ which
differ in at most $(1/2 - \lambda)N$ places, where 
$\lambda > 0$.  We then perform
efficient information reconciliation on these strings,
adapting previous work \cite{Guruswami:side}, \cite{Smith:optimal}. 

\section{Preliminaries}

\subsection{Notation}

\label{prelimsec}

When a sequence is defined, we will use Roman font to refer to individual terms
(e.g., $h_1, \ldots, h_n$) and boldface font to refer to the sequence as a whole (e.g., $\mathbf{h}$).
For any bit $b$, let $\overline{b} = 1 - b$.  For any sequence of bits
$\mathbf{b} = ( b_1, \ldots, b_n )$, let $\overline{\mathbf{b}} = (\overline{b_1} , 
\ldots , \overline{b_n} )$.

We write the expression $f ( x )^y$ (where $f$ is a function) to mean
$(f(x))^y$.  Thus, for example, in the expression
\begin{eqnarray}
\Tr [ Z]^{1/q}
\end{eqnarray}
the $(1/q)$th power map is applied after the trace function, not before it.

We write $(\log x)$ to denote the logarithm with base $2$,
and we write $(\ln x)$ to denote the logarithm with base $e$.  We use
$h \colon [0, 1 ] \to \mathbb{R}$ to denote the Shannon entropy function:
\begin{eqnarray}
h ( x ) = - x \log x - (1-x) \log (1-x).
\end{eqnarray}

We will use capital letters (e.g., $Q$) to denote quantum systems.
We use the same letter to denote both the system itself
and the complex Hilbert space which represents it.
For any finite-dimensional complex Hilbert space $Q$, let
$\L ( Q )$ denote the set of linear maps from $Q$ to itself, and let
\begin{eqnarray}
\P ( Q  ) & =  & \{ \sigma \in \L ( Q ) \mid \sigma \geq 0 \} \\ 
\S ( Q  ) & =  & \{ \sigma \in \L ( Q ) \mid \sigma \geq 0, \Tr ( \sigma ) \leq 1 \} \\ 
\D ( Q  ) & =  & \{ \sigma \in \L ( Q ) \mid \sigma \geq 0, \Tr ( \sigma ) = 1 \}.
\end{eqnarray}
These are, respectively, the set of positive semidefinite operators,
the set of subnormalized positive semidefinite operators, and the
set of density operators.

If $\rho_1 \colon {X}_1 \to {Y}_1$ and
$\rho_2 \colon {X}_2 \to {Y}_2$ are two linear
operators, then we denote by $\rho_1 \oplus \rho_2$ the
operator from ${X}_1 \oplus {X}_2$
to ${Y}_1 \oplus {Y}_2$ which maps
$(x_1, x_2)$ to $(\rho_1 ( x_1 ) , \rho_2 ( x_2 ))$.

If $(B, E)$ is a bipartite system, and
$\rho$ is a density operator on ${B} \otimes {E}$
representing a classical-quantum state, then we may
express $\rho$ as a diagonal-block operator
\begin{eqnarray}
\rho = \left[ \begin{array}{ccccc} 
\rho_1 \\
& \rho_2 \\
&& \rho_3 \\
&&& \ddots \\
&&&& \rho_m 
\end{array} \right],
\end{eqnarray}
where $\rho_1 , \ldots, \rho_m$ denote
the subnormalized operators
on ${E}$ corresponding to the basis
states of the classical register $B$.  Alternatively,
we may express $\rho$ as $\rho =
\rho_1 \oplus \rho_2 \oplus \ldots \oplus
\rho_m$.

For any $\alpha > 0$, and any linear operator $X$,
let $\left\| X \right|_\alpha$ denote the Schatten norm:
\begin{eqnarray}
\left\| X \right\|_{\alpha} & = & \Tr [ (X^* X )^{\alpha/2} ]^{1/\alpha}.
\end{eqnarray}
Note that if $X$ is positive semidefinite, this may be written more simply as
\begin{eqnarray}
\left\| X \right\|_\alpha & = & \Tr [ X^\alpha ]^{1/\alpha}.
\end{eqnarray}

We will often be concerned with the function $Z \mapsto Z^x$,
where $x \in [0, 2]$.  We note the following mathematical properties.

\begin{proposition}
\label{powmatrixprop}
Let $\gamma \in [0,1]$, and let $Z, W$ denote positive
semidefinite operators on $\mathbb{C}^n$. 
\begin{enumerate}
\item[(a)] If $Z \leq W$, then $Z^\gamma \leq W^\gamma$.

\item[(b)] If $Z \leq W$ and $X = W - Z$, then
\begin{eqnarray}
\Tr (X^{1+\gamma}  ) + \Tr ( Z^{1+\gamma} ) \leq \Tr ( W^{1+\gamma}).
\end{eqnarray}
\end{enumerate}
\end{proposition}

\begin{proof}
Part (a) is given by Theorem 2.6 in \cite{Carlen:traceinequalities}.  
Part (b) follows from part (a) by the following reasoning:
\begin{eqnarray}
\Tr ( W^{1+\gamma} ) & = & \Tr ( W \cdot W^\gamma ) \\
& = & \Tr ( X \cdot W^\gamma ) + \Tr ( Z \cdot W^\gamma ) \\
& \geq & \Tr ( X \cdot X^\gamma ) + \Tr ( Z \cdot Z^\gamma ) \\
& = & \Tr ( X^{1+ \gamma} ) + \Tr ( Z^{1+\gamma} ) .
\end{eqnarray}
This completes the proof.
\end{proof}

\subsection{Quantum R{\'e}nyi Divergence}

\label{renyisection}

In this subsection we state the definitions of the two primary
measures of randomness used in this paper (Renyi divergence
and smooth min-entropy) and establish their relationship.
We quote the definition of quantum R{\'e}nyi divergence from
\cite{JaksicOPP:2011}, \cite{MullerDSFT:2013}, \cite{WildeWY:2013}.

\begin{definition}[\cite{MullerDSFT:2013}]
\label{renyimaindefinition}
Let $\rho$ be a density matrix on $\mathbb{C}^n$.  Let $\sigma$ be
a positive semidefinite matrix on $\mathbb{C}^n$ whose support
contains the support of $\rho$.  Let $\alpha > 1$ be a real number.
Then,
\begin{eqnarray}
\label{divergencedef_d}
d_\alpha ( \rho \| \sigma ) & = & 
\Tr \left[ \left( \sigma^{\frac{1 - \alpha}{2\alpha}} \rho \sigma^{\frac{1 - \alpha}{2\alpha}}
\right)^\alpha \right]^{\frac{1}{\alpha - 1}}.
\end{eqnarray}
More generally, for any positive semidefinite matrix $\rho'$ whose
support is contained in $\textnormal{Supp } \sigma$, let
\begin{eqnarray}
d_\alpha ( \rho' \| \sigma ) & = &
\Tr \left[  \frac{1}{\Tr [ \rho' ] }  \left( \sigma^{\frac{1 - \alpha}{2 \alpha}}
\rho' \sigma^{\frac{1 - \alpha}{2 \alpha}} \right)^\alpha \right]^{\frac{1}{\alpha -1 }}.
\end{eqnarray}
Let
\begin{eqnarray}
\label{divergencedef_D}
D_\alpha ( \rho' \| \sigma ) & = &
\log d_\alpha ( \rho' \| \sigma ).
\end{eqnarray}
\end{definition}

Let $AB$ be a classical quantum system whose state
is given by a density operator $\rho_{AB}$.  One way to quantify the amount of
randomness in $A$ conditioned on $B$ is via an expression of the form
$-D_\alpha ( \rho_{AB} \| \mathbb{I}_A \otimes \sigma_B)$, where
$\sigma_B$ is a density operator.  (Maximizing over expressions
of this form leads to the corresponding notion of conditional Renyi entropy, which
will not be used directly in this paper.  See Definition~10 in
\cite{MullerDSFT:2013}.)

Note that if $\rho$ is a density matrix, then
\begin{eqnarray}
D_\alpha ( \rho \| \mathbb{I} ) =
- \frac{1}{\alpha - 1 } \log \Tr [ \rho^\alpha ].
\end{eqnarray}
For any positive semidefinite operator $\rho$, let $H_\alpha ( \rho )
:= D_\alpha ( \rho \| \mathbb{I} )$.  This is the unconditional $\alpha$-Renyi entropy of $\rho$.

Additionally, we will need a definition of smooth min-entropy.  
There are multiple definitions of smooth min-entropy that are
essentially equivalent.  The definition that we will use is
not the most up-to-date (see \cite{MCR10}) but it is good for
our purposes for its simplicity.
\begin{definition}
\label{smminentdef}
Let $AB$ be a classical-quantumtum system, and let $\rho_{AB}$
be a positive semidefinite operator.  Let
$\epsilon > 0$ be a real number.  Then,
\begin{eqnarray}
H^{\epsilon}_{min} ( A \mid B )_{\rho} =
\max_{\substack{\| \rho' - \rho \|_1 \leq \epsilon \\
\rho' \in \S ( A \otimes B) }} \hskip0.1in \max_{\substack{\sigma \in P ( A ) \\
\mathbb{I}_A \otimes \sigma \geq \rho'}} -\log ( \Tr ( \sigma )).
\end{eqnarray}
\end{definition}

The smooth min-entropy measures the number of
random bits that can be extracted from a classical
source in the presence of quantum information \cite{R05}.
When it is convenient, we will use the notation
$H^\epsilon_{min} ( \rho_{AB} | B )$
instead of $H^\epsilon_{min} ( A | B )_{\rho}$.

Following \cite{datta:relative}, let us define
the \textit{relative smooth max-entropy} of
two operators.
\begin{definition}
Let $\rho, \sigma$ be positive semidefinite operators
on $\mathbb{C}^n$ such that the support of $\sigma$
contains the support of $\rho$.  Then, 
\begin{eqnarray}
D_{max} ( \rho \| \sigma ) & = & 
\log \min_{ \substack{\lambda \in \mathbb{R} \\
\rho \leq \lambda \sigma }} (\lambda).
\end{eqnarray}
For any $\epsilon \geq 0$,
\begin{eqnarray}
D_{max}^\epsilon ( \rho \| \sigma ) & = & 
\inf_{\substack{\| \rho' - \rho \|_1 \leq \epsilon \\
\rho' \in \S ( \mathbb{C}^n ) }} D_{max} ( \rho' \| \sigma ). 
\end{eqnarray}
\end{definition}
The quantity $D_{max}^\epsilon$ is convenient
for computing lower bounds on $H_{min}^\epsilon$.
Note that if $\psi_B$ is any density matrix on 
$B$,
\begin{eqnarray}
H^\epsilon_{min} ( \rho_{AB} | B ) \geq 
- D^\epsilon_{max} ( \rho_{AB} \| \mathbb{I}_{A}
\otimes \psi_B  ).
\end{eqnarray}

The following proposition and corollary relate smooth min-entropy to
Renyi divergence.  The proof of the proposition is an easy derivative
of proofs of similar results (\cite{TomamichelCR:2009},
\cite{DupuisFS:2013}) and
is given in Appendix~\ref{renyismapp}.  The corollary follows easily.

\begin{proposition}
\label{maxrenyiprop}
Let $\alpha \in (1, 2]$.
Let $\rho$ be a density operator on 
a finite-dimensional Hilbert space $V$, and let
$\sigma \in \P ( V )$ such that
$\textnormal{Supp } \sigma \supseteq
\textnormal{Supp } \rho$.  Then,
\begin{eqnarray}
\label{adivbound}
D_{max}^\epsilon ( \rho \| \sigma ) \leq 
D_\alpha ( \rho \| \sigma ) + \frac{2 \log ( 1/\epsilon ) + 1}{\alpha -1 }.
\end{eqnarray}
Additionally, if $\rho$ is a classical-quantum operator
on a bipartite state, then there exists a classical-quantum operator $\rho'$ with $\left\| \rho' - \rho \right\|_1 \leq \epsilon$
and $\rho' \geq 0$
such that $D_{max} ( \rho' \| \sigma )$ satisfies the
above bound. \qed
\end{proposition}

\begin{corollary}
\label{divsmcor}
Let $AB$ be a classical-quantum bipartite system, and let
$\rho_{AB}$ be a density operator.  Let $\sigma_B$ be
a density operator on $B$ whose support contains $\Supp \rho_B$.  Let $\epsilon > 0$
and $\alpha \in (1, 2 ]$ be real numbers.  Then,
for any $\epsilon > 0$,
\begin{eqnarray}
\label{smedivbound}
H^\epsilon_{min} \left( A \mid B \right)_\rho \geq
- D_\alpha \left( \rho \| \mathbb{I}_A \otimes \sigma_B \right) - \frac{2 \log (1/\epsilon) + 1}{
 \alpha - 1 }. \qed
\end{eqnarray}
\end{corollary}

\subsection{Quantum Devices}

\label{canonicalsubsec}

Let us formalize some terminology and notation for describing quantum
devices.  (Our formalism is a variation on that which has
appeared in other papers on untrusted devices, such as \cite{ruv:2013}.)

\begin{definition}
\label{qddef}
Let $n$ be a positive integer.  A \textbf{binary quantum device with $n$
components} $D = (D_1, \ldots, D_n)$ consists
of the following.

\begin{enumerate}
\item Quantum systems $Q_1, \ldots, Q_n$ whose
initial state is specified by a density operator,
\begin{eqnarray}
\Phi \colon ( Q_1 \otimes \ldots \otimes Q_n)  \to 
( Q_1 \otimes \ldots \otimes Q_n) 
\end{eqnarray}

\item For any $k \geq 0$, and any function
\begin{eqnarray}
T \colon \{ 0, 1 \} \times \{ 1, 2, \ldots, k \} \times \{ 1, 2, \ldots, n \} 
\to \{ 0, 1 \}
\end{eqnarray}
a unitary operator
\begin{eqnarray}
U_T \colon \left( {Q}_1 \otimes \ldots \otimes {Q}_n \right) \to 
\left( {Q}_1 \otimes \ldots \otimes {Q}_n \right).
\end{eqnarray}
and a collection of Hermitian operators
\begin{eqnarray}
\left\{ M^{(b)}_{T, j} \colon {Q}_j \to {Q}_j  \right\}_{\substack{b \in \{ 0, 1 \} \\
1 \leq j \leq n }}
\end{eqnarray}
satisfying $\left\| M_{T, j}^{(b)} \right\| \leq 1$.
\end{enumerate}
\end{definition}

The device $D$ behaves as follows.  Suppose that $k$ iterations
of the device have already taken place, and suppose that $T$ is
such that $T ( 0, i, j ) \in \{ 0, 1 \}$
and $T( 1, i, j ) \in \{ 0, 1 \}$ represent the input bit and output bit, respectively, for
the $j$th player on the $i$th round ($i \leq k$).  ($T$ is the \textbf{transcript function}.)
Then, 
\begin{enumerate}
\item The components $D_1, \ldots, D_n$ collectively perform
the unitary operation $U_T$ on ${Q}_1 \otimes  \ldots \otimes {Q}_n$. 

\item Each component $D_j$ receives
its input bit $b_j$, then applies the binary nondestructive measurement
on $Q_i$ 
given by
\begin{eqnarray}
X & \mapsto & \left( \sqrt{ \frac{ \mathbb{I} + M_{T, j}^{(b_j)}}{2} } \right)
X \left( \sqrt{ \frac{ \mathbb{I} + M_{T, j}^{(b_j)}}{2} } \right) \\
X & \mapsto & \left( \sqrt{ \frac{ \mathbb{I} - M_{T, j}^{(b_j)}}{2} } \right)
X \left( \sqrt{ \frac{ \mathbb{I} - M_{T, j}^{(b_j)}}{2} } \right),
\end{eqnarray}
and then outputs the result.
\end{enumerate}

Let us say that one binary quantum device $D'$ \textbf{simulates}
another binary quantum device $D$ if, 
for any purifying systems $E$ and $E'$ (for $D$ and $D'$, respectively),
and any input sequence $\mathbf{i}_1, \ldots, \mathbf{i}_k \in \{ 0, 1 \}^n$,
the joint state of the outputs of $D$ together with $E$ is isomorphic
to the joint state of the outputs of $D'$ together with $E'$
on the same input sequence.  Similarly,
let us say that a protocol X \textbf{simulates} another protocol Y if,
for any purifying systems $E$ and $E'$ for the quantum devices
used by X and Y, respectively, the joint state of E together with the
outputs of X is isomorphic to the joint state of E' together with the
outputs of Y.

\begin{definition}
Let us say that a binary quantum device $D$ is in \textbf{canonical form}
if each of its quantum systems $Q_j$ is such that ${Q}_j = \mathbb{C}^{2m_j}$
for some $m_j \geq 1$, and each measurement operator pair
$(M^{(0)} , M^{(1)} ) = 
(M_{T,j}^{(0)}, M_{T, j}^{(1)})$ has the following $2 \times 2$ diagonal block form:
\begin{eqnarray*}
M^{(0)}   =   \left[ \begin{array}{ccccccc}
0 & 1  \\
1 & 0  \\
&& 0 & 1 \\
&& 1 & 0 \\
& & & & \ddots \\
&& & & & 0 & 1 \\
&& & & & 1 & 0 \\
\end{array} \right] & \hskip0.6in  &
M^{(1)}  =   \left[ \begin{array}{ccccccc}
0 & \zeta_1  \\
\overline{\zeta_1}  & 0 \\
&&0 & \zeta_2  \\
&&\overline{\zeta_2}  & 0 \\
 &&& & \ddots \\
& &&& & 0 & \zeta_{m_j} \\
& & &&& \overline{\zeta_{m_j}} & 0 \\
\end{array} \right],
\end{eqnarray*}
where the complex numbers $\zeta_\ell$ satisfy
\begin{eqnarray}
\left| \zeta_\ell \right| = 1 \textnormal{ and }
\Im ( \zeta_\ell ) \geq 0.
\end{eqnarray}
(Note that the complex numbers $\zeta_\ell$ may be different for each
transcript $T$ and each player $j$.)
\end{definition}

When we discuss quantum devices that are in canonical form,
we will frequently make use of the isomorphism 
$\mathbb{C}^{2m} \cong \mathbb{C}^2 \otimes \mathbb{C}^m$
given by $e_{2k-1} \mapsto e_1 \otimes e_k$, $e_{2k} \mapsto
e_2 \otimes e_k$.  (Here, $e_1, \ldots, e_r$ denote the
standard basis vectors for $\mathbb{C}^r$.)

\begin{proposition}
\label{canonicalsimprop}
Any binary quantum device can be simulated by a device
that is in canonical form.
\end{proposition}

\begin{proof}
This follows from Theorem~\ref{canintermed} in the appendix.
\end{proof}

\section{An Uncertainty Principle}

\label{opineqsection}

In this section, we consider the behavior of the map
$\rho \mapsto \Tr [ \rho^{1 + \epsilon } ]$ when 
measurements are applied to a qubit and
the operator $\rho$ represents the state of a system
that is entangled with the qubit.

We begin by quoting the following theorem, which appears as part of
Theorem 5.1 in the paper \cite{PisierX:2003}.  

\begin{theorem}
\label{pisierthm}
Let $X, Y \colon \mathbb{C}^m \to \mathbb{C}^n$ be linear operators.
Let $p \geq 2$ be a real number, and let $p' = 1/(1 - 1/p)$.  Then,
\begin{eqnarray}
\label{pisierineq}
\left[ \frac{1}{2} \left( \left\| X + Y \right\|_p^p + \left\| X - Y \right\|_p^p
\right) \right]^{1/p} \leq \left( \left\| X \right\|_p^{p'} +
\left\| Y \right\|_p^{p'} \right)^{1/p'}. \qed
\end{eqnarray}
\end{theorem}

Inequality~(\ref{pisierineq}) may alternatively be expressed as
\begin{eqnarray}
\label{altineq}
\left[ \left\| \frac{X + Y}{\sqrt{2}} \right\|_p^p + \left\| \frac{X - Y}{\sqrt{2}}
\right\|_p^p \right]^{1/p} \leq 2^{1/p - 1/2} \left( \left\| X \right\|_p^{p'} +
\left\| Y \right\|_p^{p'} \right)^{1/p'}
\end{eqnarray}
or,
\begin{eqnarray}
\label{altineq2}
\left\| \frac{X + Y}{\sqrt{2}} \right\|_p^p + \left\| \frac{X - Y}{\sqrt{2}}
\right\|_p^p \leq 2^{1 - p/2} \left( \left\| X \right\|_p^{p'} +
\left\| Y \right\|_p^{p'} \right)^{p/p'}.
\end{eqnarray}

Observe the following: if $QW$ is a bipartite quantum system with $Q = \mathbb{C}^2$
and $\Lambda \in \D ( Q \otimes W )$ is a density operator, $\Lambda$ can be written as
\begin{eqnarray}
\Lambda & = & \left[ \begin{array}{c|c} 
X^* X & X^* Y \\
\hline
Y^* X & Y^* Y \end{array} \right]
\end{eqnarray}
for some $X, Y \in \L ( W )$.  Then the reduced state of $W$ is
\begin{eqnarray}
\label{rhostatement}
\rho & := & X^* X + Y^* Y
\end{eqnarray}
Additionally, if we let $\{ \rho_0, \rho_1 \}
$ and $\{ \rho_+, \rho_- \}$ denote the subnormalized states of $W$ that
arise from measurements on $Q$ along the computational and Hadamard bases,
respectively, then
\begin{eqnarray}
\rho_0 & = & X^* X \\
\rho_1 & = & Y^* Y \\
\rho_+ & = & \left( \frac{X + Y}{\sqrt{2}} \right)^* \left( \frac{X + Y}{\sqrt{2}} \right), \\
\label{rhominusstatement}
\rho_- & = & \left( \frac{X - Y}{\sqrt{2}} \right)^* \left( \frac{X - Y}{\sqrt{2}} \right).
\end{eqnarray}

\begin{theorem}
\label{uncertaintythm}
There exists a continuous function $\Pi \colon (0, 1 ] \times [0, 1 ]
\to \mathbb{R}$ such that the following holds.
\begin{enumerate}
\item 
Let $V$ be a quantum system, and let 
$\rho, \rho_0, \rho_0, \rho_+, \rho_- \in \mathcal{S} ( V )$ denote
operators arising from measurements of a qubit entangled with $V$.  Let
\begin{eqnarray}
t = \frac{\Tr ( \rho_1^{1+\epsilon} )}{\Tr ( \rho^{1+\epsilon})}.
\end{eqnarray}
Then, the following inequality always holds:
\begin{eqnarray}
\label{uncertaintyexp}
\log \left[ \frac{ \Tr ( \rho_+^{1 + \epsilon } + \rho_-^{1 + \epsilon} )}{\Tr ( \rho^{1 + \epsilon} )
} \right]  & \leq & - \epsilon \Pi ( \epsilon , t ).
\end{eqnarray}

\item The limiting function $\pi ( z) := \lim_{(x,y )
\to ( 0, z )} \Pi ( x, y )$ is given by
\begin{eqnarray}
\pi ( z ) & = & 1 - 2 z \log \left( \frac{1}{z} \right) 
- 2 (1 -z ) \log \left( \frac{1}{1 -z } \right).
\end{eqnarray}
\end{enumerate}
\end{theorem}

\begin{proof}
Express the states $\rho_*$ in terms of operators $X$ and $Y$ as in
(\ref{rhostatement}--\ref{rhominusstatement}).
Applying (\ref{altineq2}) with $p = 2 + 2 \epsilon$,
and $p' = 1/(1 - 1/p)$ we have the following:
\begin{eqnarray}
\Tr ( \rho_+^{1 + \epsilon} + \rho_-^{1 + \epsilon} ) & = & 
\left\| \frac{X + Y}{\sqrt{2}} \right\|_{2 + 2 \epsilon}^{2 + 2\epsilon} + \left\| \frac{X - Y}{\sqrt{2}}
\right\|_{2+2\epsilon}^{2+2\epsilon} \\ 
& \leq & 2^{1 - p/2} \left( \left\| X \right\|_p^{p'} +
\left\| Y \right\|_p^{p'} \right)^{p/p'} \\
& = & 2^{1 - p/2} \left[ \left( \left\| X \right\|_p^p \right)^{p'/p} +
\left( \left\| Y \right\|_p^p \right)^{p'/p} \right]^{p/p'} \\
& = & 2^{- \epsilon} \left[ \Tr \left(  \rho_0^{1 + \epsilon} \right)^\frac{1}{1 + 2 \epsilon} +
\Tr \left( \rho_1^{1 + \epsilon} \right)^{\frac{1}{1 + 2 \epsilon}} \right]^{1 + 2 \epsilon} 
\end{eqnarray}
Letting
\begin{eqnarray}
s  =  \frac{\Tr ( \rho_0^{1 + \epsilon} )}{\Tr ( \rho^{1 + \epsilon})},
\end{eqnarray}
we have
\begin{eqnarray}
\Tr ( \rho_+^{1 + \epsilon} + \rho_-^{1 + \epsilon} )
& \leq & 2^{- \epsilon} \left[ s^\frac{1}{1 + 2 \epsilon} +
t^{\frac{1}{1 + 2 \epsilon}} \right]^{1 + 2 \epsilon} \Tr ( \rho^{1 + \epsilon} ).
\end{eqnarray}
Since $\Tr ( \rho_0^{1 + \epsilon} ) + \Tr ( \rho_1^{1 + \epsilon} ) \leq
\Tr ( \rho^{1 + \epsilon} )$, we
have $t + s \leq 1$, and therefore,
\begin{eqnarray}
\Tr ( \rho_+^{1 + \epsilon} + \rho_-^{1 + \epsilon} )
& \leq & 2^{- \epsilon} \left[ (1 - t)^\frac{1}{1 + 2 \epsilon} +
t^{\frac{1}{1 + 2 \epsilon}} \right]^{1 + 2 \epsilon} \Tr ( \rho^{1 + \epsilon} ).
\end{eqnarray}
Let
\begin{eqnarray}
\Pi ( x, y ) & = & - \frac{1}{x} \log \left\{ 2^{- x} \left[ (1 - y)^\frac{1}{1 + 2 x} +
y^{\frac{1}{1 + 2 x}} \right]^{1 + 2 x} \right\}.
\end{eqnarray}
The desired limiting condition follows using L'Hospital's rule.
\end{proof}

We note that (\ref{uncertaintyexp}) can be rewritten as
\begin{eqnarray}
\label{uncertaintyexp2}
\left( - \frac{1}{\epsilon} \log \Tr ( \rho_+^{1+\epsilon} + \rho_-^{1+\epsilon} ) \right) - \left( - \frac{1}{\epsilon} \log \Tr ( \rho^{1+\epsilon} ) \right)
& \geq & \Pi ( \epsilon, t ).
\end{eqnarray}
The expression on the left side is the
difference in $(1+\epsilon)$-Renyi entropy between the state $\rho_+ \oplus \rho_-$ and the
state $\rho$.

\section{The Self-Testing Property of Binary Nonlocal XOR Games}

\label{gamessec}

In this section we review some of the known formalism for
binary XOR games, and then prove new results.

\subsection{Definitions and Basic Results}

\label{gamedefsubsec}

\begin{definition}
An \textbf{$n$-player binary nonlocal XOR game} consists of a probability
distribution
\begin{eqnarray}
\{ p_{\mathbf{i}} \mid \mathbf{i} \in \{ 0, 1 \}^n \}
\end{eqnarray} on the
set $\{ 0, 1 \}^n$, together with an indexed set
\begin{eqnarray}
\{ \eta_{\mathbf{i}} \in \{ -1, 1 \}
\mid \mathbf{i} \in \{ 0, 1 \}^n \}.
\end{eqnarray}
\end{definition}

Given any indexed sets $\{ p_{\mathbf{i}} \}$ and $\{ \eta_\mathbf{i} \}$
satisfying the above conditions, we can conduct an $n$-player nonlocal game as follows.
\begin{enumerate}
\item A referee chooses a binary vector $\mathbf{c} \in \{ 0, 1 \}^n$ according
to the distribution $\{ p_\mathbf{i} \}$.  For each $k$, he gives
the bit $c_k$ as input to the $k$th player.

\item Each player returns an output bit $d_k$ to the referee.

\item The referee calculates the score, which is given by
\begin{eqnarray}
\eta_\mathbf{c} (-1)^{d_1 + d_2 + \cdots + d_n}.
\end{eqnarray}
If the score is $+1$, a ``pass'' has occurred.  If the score is $-1$,
a ``failure'' has occurred.
\end{enumerate}

We quote some definitions and results from \cite{MillerS:self-testing:2013} and \cite{WW01a}.

\begin{definition}
An \textbf{mixed $n$-player quantum strategy} is a pair
\begin{eqnarray}
\label{qubitstrat}
\left( \Psi , \{ \{  M_j^{(0)} , M_j^{(1)} \} \}_{j=1}^n \right)
\end{eqnarray}
where $\Psi$ is a density matrix on an $n$-tensor product space
$V_1 \otimes \ldots \otimes V_n$ and
$M_j^{(i)}$ denotes a linear operator on $V_j$ whose
eigenvalues are contained in $\{ -1, 1 \}$.  A \textbf{pure
$n$-player quantum strategy} is a pair
\begin{eqnarray}
\label{purestrat}
\left( \psi , \{ \{  M_j^{(0)} , M_j^{(1)} \} \}_{j=1}^n \right)
\end{eqnarray}
which satisfies the same conditions, except
that $\psi$ is merely a unit vector
on $V_1 \otimes \ldots \otimes V_n$.
A \textbf{qubit strategy}
is a pure quantum strategy in which the spaces $V_i$ are
equal to $\mathbb{C}^2$ and the operators
$M_j^{(i)}$ are all nonscalar. 

The \textbf{score} achieved by a quantum strategy at an
$n$-player binary nonlocal XOR game
$G = (\{ p_\mathbf{i} \} , \{ \eta_{\mathbf{i}} \} )$ is the expected
score when the qubit strategy is used to play the game $G$.  This quantity
can be expressed as follows.  Let $\mathbf{M}$ denote the
\textbf{scoring operator} for $G$, which is given by
\begin{eqnarray}
\label{scoreop}
\mathbf{M} & = & \sum_{\mathbf{i} \in \{ 0, 1 \}^n } 
 p_\mathbf{i} \eta_{\mathbf{i}} M_1^{(i_1)} \otimes M_2^{(i_2)} \otimes
\cdots \otimes M_n^{(i_n)}.
\end{eqnarray}
Then, the score for strategy (\ref{qubitstrat}) at game $G$ is
$\Tr ( \mathbf{M} \Psi )$.  The score for the pure strategy
(\ref{purestrat}) is $\psi^* \mathbf{M} \psi$.

The \textbf{optimal score} for a nonlocal game is the highest score
that can be achieved at the game by qubit strategies.  
We denote this quantity
by $\mathfrak{q}_G$.  (As explained
in \cite{MillerS:self-testing:2013}, this is also the highest score that can be achieved
by arbitrary quantum strategies.)
A game $G$ is a \textbf{self-test} if there is only one qubit
strategy (modulo local unitary operations on the $n$ tensor components
of $\left( \mathbb{C}^2 \right)^{\otimes n}$)
which achieves the optimal score.
A game $G$ is \textbf{winnable} if $\mathfrak{q}_G = 1$.

Note that $\mathfrak{q}_G$ is different from the maximum \textbf{passing probability}
for quantum strategies, which
we denote by $\mathbf{w}_G$.  The two are related
by $\mathbf{w}_G = (1 + \mathfrak{q}_G)/2$.
We will also write $\mathbf{f}_G$ for the \textbf{minimum failing probability},
which is given by $\mathbf{f}_G = 1 - \mathbf{w}_G$.

We define functions that are useful for the study of binary XOR games.
For any nonlocal game $G = (\{ p_\mathbf{i} \} , \{ \eta_{\mathbf{i}} \} )$, define
$P_G \colon \mathbb{C}^n \to \mathbb{C}$ by 
\begin{eqnarray}
P_G ( \lambda_1, \ldots, \lambda_n  ) & = & \sum_{\mathbf{i} \in \{ 0, 1 \}^n }
p_\mathbf{i} \eta_{\mathbf{i}} \lambda_1^{i_1} \lambda_2^{i_2} \ldots \lambda_n^{i_n}.
\end{eqnarray}
Define $Z_G \colon \mathbb{R}^{n+1} \to \mathbb{R}$ by
\begin{eqnarray}
Z_G ( \theta_0 , \theta_1, \ldots , \theta_n ) & = & \sum_{\mathbf{i} \in \{ 0, 1 \}^n }
p_\mathbf{i} \eta_\mathbf{i} \cos \left( \theta_0 + \sum_{k=1}^n i_k \theta_k \right).
\end{eqnarray}
\end{definition}
These functions are related by
\begin{eqnarray}
Z_G ( \theta_0, \ldots, \theta_n ) &=&
\Re \left[ e^{i \theta_0 } P ( e^{i \theta_1} , e^{i \theta_2} ,
\ldots , e^{i \theta_n } ) \right]. \\
|P_G ( e^{i\theta_1} , \ldots , 
e^{i \theta_n} )| & = & \max_{\theta_0 \in [-\pi , \pi ] } Z_G ( \theta_0, \ldots, \theta_n ).
\end{eqnarray}

The functions $P_G$ and $Z_G$ can be used to calculate $\mathfrak{q}_G$.  This 
was observed by Werner and Wolf in \cite{WW01a}.  We sketch a proof here.
(For a more detailed proof, see Proposition 1 in \cite{MillerS:self-testing:2013}.)
\begin{proposition}
For any nonlocal binary XOR game $G$, the following equalities hold.
\begin{eqnarray}
\mathfrak{q}_G & = & \max_{| \lambda_1 | = \ldots = |\lambda_n | = 1}
\left| P_g ( \lambda_1 , \ldots, \lambda_n ) \right| \\
& = & \max_{\theta_0 , \ldots, \theta_n \in \mathbb{R} }
Z_g ( \theta_0, \ldots, \theta_n ).
\end{eqnarray}
\end{proposition}

\begin{proof}[sketch]
Let $( \psi , \{ M_j^{(i)} \} )$ be a qubit strategy for $G$.  By an appropriate
choice of basis, we may assume that
\begin{eqnarray}
M_j^{(0)} = \left[ \begin{array}{cc} 0 & 1 \\ 1 & 0 \end{array} \right] 
\hskip0.2in \textnormal{ and } \hskip0.2in
M_j^{(1)} = \left[ \begin{array}{cc} 0 & \zeta_j \\ \overline{\zeta_j} & 0 \end{array}
\right].
\end{eqnarray}
where $\{ \zeta_j \}$ are complex numbers of length $1$.  The scoring operator
$\mathbf{M}$ can be expressed as a reverse diagonal matrix whose entries
are
\begin{eqnarray}
\label{reverseentries}
\left\{ P_G ( \zeta_1^{b_1} , \ldots, \zeta_n^{b_n} ) \right\}_{(b_1, \ldots, b_n) \in \{ -1, 1 \}^n}.
\end{eqnarray}
The eigenvalues of a reverse diagonal Hermitian matrix whose reverse-diagonal
entries are equal to $z_1, z_2, \ldots, z_{2n}$ is simply $\pm | z_1 | , \pm |z_2| , \ldots, \pm |z_n |$.
Therefore the operator norm of $\mathbf{M}$ is the maximum absolute value that occurs
in (\ref{reverseentries}).

The value $q_f$ is the maximum of the operator norm that occurs among 
all the scoring operators arising from qubit strategies for $G$.  The desired formulas
follow.
\end{proof}

\begin{proposition}
\label{selftestprop}
Let $G$ be a nonlocal binary XOR game.  Then, $G$ is a self-test if and only if
the following two conditions are satisfied.
\begin{enumerate}
\item[(A)] There is a maximum $(\alpha_0, \ldots, \alpha_n)$ for $Z_G$
such that none of $\alpha_1, \ldots, \alpha_n$ is a multiple of $\pi$.

\item[(B)] Every other maximum of $Z_G$ is congruent modulo 
$2 \pi$ to either $(\alpha_0, \ldots, \alpha_n )$ or $(- \alpha_0, \ldots, - \alpha_n)$.
\end{enumerate}
\end{proposition}

\begin{proof}
See Proposition 2 in \cite{MillerS:self-testing:2013}.
\end{proof}

The following definition will be convenient in later proofs.

\begin{proposition}
Let $G$ be a nonlocal game which is a self-test.  Then, $G$
is \textbf{positively aligned} if a maximum for $Z_G ( \theta_0 , \ldots, \theta_n )$
occurs in the region
\begin{eqnarray}
\left\{ (\theta_0, \ldots, \theta_n ) \mid 0 < \theta_i < \pi \hskip0.2in \forall i \geq 1\right\}.
\end{eqnarray}
\end{proposition}

For any binary XOR self-test $G = \left( \{ p_\mathbf{i} \} ,
\{ \eta_\mathbf{i} \} \right)$, we can construct
a positively aligned self-test $G' = \left( \{ p'_\mathbf{i} \} , \{ \eta'_\mathbf{i} \} \right)$ by setting $b_1, \ldots, b_n \in \{ 0, 1 \}$ so that
$b_i = 0$ if $Z_G$ has a maximum with $\theta_i \in (0, \pi )$, and $b_i = 1$ if not,
and letting
\begin{eqnarray}
p'_\mathbf{i} & = & p_\mathbf{i} \\
\eta'_\mathbf{i} & = & \eta_{(\mathbf{i} + \mathbf{b}) \textbf{ mod } 2}.
\end{eqnarray}
It is easy to see that $\mathfrak{q}_{G'} = \mathfrak{q}_G$.

\begin{definition}\label{def:strong}
Let $( \psi , \{ M_j^{(i)} )$ and $( \phi , \{ N_j^{(i)} \})$ be $n$-player qubit
strategies.  Then the \textbf{distance} between these two strategies is
the quantity
\begin{eqnarray}
\max \left( \{ \left\| \psi - \phi \right\| \} \cup \left\{  \left\| M_j^{(i)} - N_j^{(i)} \right\| \mid j \in \{ 1, 2, \ldots, n \},
i \in \{ 0,1 \} \right\} \right).
\end{eqnarray}
(In this formula, the first norm denotes Euclidean distance
and second denotes operator norm.)
Let $G$ be a self-test.  Then, $G$ is a \textbf{strong self-test} if
there exists a constant $K$ such that any qubit strategy that achieves
a score of $\mathfrak{q}_G - \epsilon$ is within distance $K \sqrt{\epsilon}$ from a qubit strategy
that achieves the score $\mathfrak{q}_G$.
\end{definition}

For any twice differentiable $m$-variable function $F \colon \mathbb{R}^m \to \mathbb{R}$,
and any $c = (c_1, \ldots, c_m ) \in \mathbb{R}^m$, we can define the Hessian
matrix for $F$ at $c$, which is the $m \times m$ matrix formed
from the second partial derivatives
\begin{eqnarray}
\frac{\partial^2 F}{\partial x_i \partial x_j } ( c_1 , \ldots, c_m)
\end{eqnarray}
(for $i,j \in \{ 1, 2, \ldots, m \}$).

\begin{proposition}
\label{strongrobustprop}
Let $G$ be an $n$-player self-test.  Then the following conditions are equivalent.
\begin{enumerate}
\item  $G$ is a strong self-test.

\item The function $Z_G$ has nonzero Hessian
matrices at all of its maxima.

\item \label{nearmaxprop} There exists a constant $K > 0$ such that
any $(\beta_0, \ldots, \beta_n ) \in \mathbb{R}^{n+1}$ which
satisfies
\begin{eqnarray*}
Z_G ( \beta_0, \ldots, \beta_n ) \geq \mathfrak{q}_G - \epsilon
\end{eqnarray*}
(with $\epsilon \geq 0$) must be within
distance $K \sqrt{\epsilon}$ from
a maximum of $Z_G$.
\end{enumerate}
\end{proposition}

\begin{proof}
(1) $\Longleftrightarrow$ (2) is Proposition 3 in \cite{MillerS:self-testing:2013}.  
(2) $\Longleftrightarrow$ (3) follows from an easy calculus argument.
\end{proof}

We next prove a proposition and corollary which state consequences
of the strong self-testing conditions.  These will be the basis
for proofs in subsection~\ref{decompsec}.
\begin{proposition}
\label{polarcurveprop}
Let $G$ be a positively-aligned strong self-test.
Let $H$ denote the semicircle $\{ e^{i  \beta} \mid 0 \leq \beta \leq \pi \}
\subseteq \mathbb{C}$. 
Then, there exists $\alpha \in [-\pi, \pi ]$ and $c \geq 0$
such that the set
\begin{eqnarray}
P_G ( H^n ) \subseteq \mathbb{C}
\end{eqnarray}
is bounded by the polar curve
\begin{eqnarray}
\label{polarcurve}
& f \colon [-\pi , \pi ] \to \mathbb{C} \\
\nonumber & f( \theta ) = (\mathfrak{q}_G - c ( \theta - \alpha)^2 ) e^{i \theta}.
\end{eqnarray}
\end{proposition}

\begin{proof}
Since $G$ is positively aligned, we may find a maximum $(\alpha_0, \ldots, \alpha_n)$
for $Z_G$ such that $\alpha_1, \ldots, \alpha_n \in (0, \pi )$.  Choose $K$
according to condition (\ref{nearmaxprop}) from Proposition~\ref{strongrobustprop}.
Let $c = 1/K^2$ and $\alpha = - \alpha_0$.

Suppose, for the sake of contradiction, that there is a point in the 
set $P_f ( H^n )$ which lies outside of (\ref{polarcurve}).  Then,
there exists $\beta_1, \ldots, \beta_n \in [0, \pi]$ such that
\begin{eqnarray}
P_f ( e^{i \beta_1} , \ldots, e^{i \beta_n} ) = r e^{i \theta}
\end{eqnarray}
(with $\theta \in [-\pi, \pi ]$) and
\begin{eqnarray}
r > \mathfrak{q}_G - c ( \theta - \alpha)^2.
\end{eqnarray}
Let $\epsilon = (1/K^2) (\theta - \alpha)^2$. We have
\begin{eqnarray}
Z_G ( - \theta , \beta_1, \ldots, \beta_n ) = r & > & \mathfrak{q}_G - c( \theta - \alpha)^2 \\
& = & \mathfrak{q}_G - \epsilon,
\end{eqnarray}
and the distance between $( - \theta , \beta_1, \ldots, \beta_n )$ 
and $(\alpha_0, \ldots, \alpha_n)$ is
at least $|\theta - \alpha| = K \sqrt{\epsilon}$. (And, it
is easy to see that $(- \theta , \beta_1 , \ldots, \beta_n )$ 
is not any closer to any of the other maxima of $Z_G$
than it is to $(\alpha_0, \ldots , \alpha_n )$.)  This contradicts
condition (\ref{nearmaxprop}) of Proposition~\ref{strongrobustprop}.
\end{proof}

\begin{corollary}
\label{displacedcor}
Let $G$ satisfy the assumptions of Proposition~\ref{polarcurveprop}.
Then, there exists a complex number $\gamma \neq 0$ such that
for all $\zeta_1, \ldots, \zeta_n \in H$, 
\begin{eqnarray}
\left| P_G ( \zeta_1, \ldots, \zeta_n ) - \gamma \right|  + \left| \gamma
\right| \leq \mathfrak{q}_G.
\end{eqnarray}
\end{corollary}

\begin{proof}
Let $R \subseteq \mathbb{C}$ be the region enclosed
by the polar curve (\ref{polarcurve}).  Let $S = \{ z \in \mathbb{C}
\mid |z| = \mathfrak{q}_G \}$.  We have $S \cap R = \{ \mathfrak{q}_G \cdot e^{i \alpha} \}$.  
Since the curvature of the curve $(\ref{polarcurve})$ at
$e^{i \alpha}$ is
strictly greater than $1/\mathfrak{q}_G$, we can find a circle of radius less than $\mathfrak{q}_G$
which lies inside of $S$, which is tangent to $S$ at $\mathfrak{q}_G \cdot e^{i \alpha}$,
and which encloses the region $R$.    Then, if we let $\gamma$ be
the center of this circle, we have 
$|z - \gamma| + | \gamma | \leq \mathfrak{q}_G$ for all $z \in R$.  The desired inequality follows.
\end{proof}

\subsection{Decomposition Theorems}

\label{decompsec}

This subsection proves results on the measurements
that are simulated by strong self-tests.
For any unit-length complex number $\zeta$, let us write $g_\zeta$
for the following modified GHZ state:
\begin{eqnarray}
g_\zeta & = & \frac{1}{\sqrt{2}} \left( \left| 00 \ldots 0 \right>
+ \zeta \left| 11 \ldots 1 \right> \right).
\end{eqnarray}
The next theorem uses the canonical form for binary measurements
from subsection~\ref{canonicalsubsec}.  Note that when 
a collection of four projections $\{ P^{(b,c)} \}$ is in canonical
form over a space $\mathbb{C}^{2m}$, we can naturally
express them as operators on $\mathbb{C}^2 \otimes \mathbb{C}^m$
via the isomorphism $\mathbb{C}^{2m} \to \mathbb{C}^2 \otimes
\mathbb{C}^m$ given by $e_{2k-1} \mapsto e_1 \otimes e_k$,
$e_{2k} \mapsto e_2 \otimes e_k$.

\begin{theorem}
\label{winnabledecompthm}
Let $G = ( \{ p_\mathbf{i} \} , \{ \eta_\mathbf{i} \} )$ be a \textbf{winnable} $n$-player self-test which is such that
\begin{enumerate}
\item $G$ is positively aligned, and
\item $p_{00 \ldots 0} > 0$ and $\eta_{00 \ldots 0} = 1$.
\end{enumerate}
Then, there exists a constant $\delta_G > 0$ such that the following holds.
Let 
$( \Phi , \{ M_j^{(i)} \} )$ be a quantum strategy whose
measurements are in canonical form with underlying
space $(\mathbb{C}^2 \otimes W_1 ) \otimes \ldots \otimes ( \mathbb{C}^2 \otimes
W_n )$.  Then the scoring operator $\mathbf{M}$ can be decomposed as
\begin{eqnarray}
\mathbf{M} & = & \delta_G \mathbf{M}' + (1 - \delta_G) \mathbf{M}'',
\end{eqnarray}
where $\left\| \mathbf{M}'' \right\| \leq 1$, and
\begin{eqnarray}
\label{m0formula}
\mathbf{M}' & = & ( g_1 g_1^* - g_{-1} g_{-1}^* ) \otimes
\mathbb{I}_{W_1 \otimes \cdots \otimes W_n}.
\end{eqnarray}
\end{theorem}

\begin{proof}
Let
\begin{eqnarray}
\label{tplus}
T^+ & = & \left\{ (\theta_0 , \ldots , \theta_n ) \in \mathbb{R}^{n+1}
\mid \theta_i > 0 \hskip0.1in \forall i \geq 1 \right\},
\end{eqnarray}
and
\begin{eqnarray}
\label{tminus}
T^- & = & \left\{ (\theta_0 , \ldots, \theta_n ) \in \mathbb{R}^{n+1}
\mid \theta_i < 0 \hskip0.1in \forall i \geq 1 \right\},
\end{eqnarray}
Let $\mathfrak{q}'_G$ be the maximum value of $Z_G$ that occurs
on the set $[-\pi, \pi]^{n+1} \smallsetminus (T^+  \cup T^-)$.  By 
the criteria from Proposition~\ref{selftestprop}, this set does not include
any of the global maxima for the function $Z_G$, and so $\mathfrak{q}'_G$
is strictly
smaller than the overall maximum $\mathfrak{q}_G = 1$.  Let
\begin{eqnarray}
\delta_G = \min \left\{ p_{00 \ldots 0 } , \mathfrak{q}_G - \mathfrak{q}'_G \right\},
\end{eqnarray}
where $p_{00 \ldots 0 }$ denotes the probability which
$G$ associates to the input string $00 \ldots 0$.

First let us address the case where $\dim W_j = 1$ for all $j$.  Then
\begin{eqnarray}
M_j^{(0)} & = & \left[ \begin{array}{cc} 0 & 1 \\ 1 & 0 \end{array} \right], \\
M_j^{(1)} & = & \left[ \begin{array}{cc} 0 & \zeta_j \\ \overline{\zeta_j} & 0 \end{array} \right].
\end{eqnarray}
We can compute the scoring operator $\mathbf{M}$ using formula (\ref{scoreop}).
When we write this operator as a matrix,
using the computational basis for $(\mathbb{C}^2)^{\otimes n}$ in lexicographical
order, we obtain a reverse diagonal matrix,
\begin{eqnarray}
\mathbf{M} & = & \left[ \begin{array}{ccccccccc} &&&& a_{00 \ldots 0} \\
&&& a_{00 \ldots 1 } \\
&& \iddots \\
& a_{11 \ldots 0} \\
 a_{11 \ldots 1}
\end{array} \right]
\end{eqnarray}
where 
\begin{eqnarray}
a_{b_1, \ldots, b_n} & = & P_G ( \zeta_1^{(-1)^{b_1}} , \zeta_2^{(-1)^{b_2}} , \ldots , \zeta_n^{(-1)^{b_n}} ). \end{eqnarray}
By canonical form, we have $\zeta_j = e^{i \theta_j}$ for some $\theta_j \in [ 0 , \pi ]$.  
Note that can write
\begin{eqnarray}
\left| a_{b_1, \ldots, b_n} \right| & = & 
\max_{\theta_0 \in \mathbb{R} } \hskip0.1in  Z_G ( \theta_0, (-1)^{b_1} \theta_1 ,
(-1)^{b_2} \theta_2, \ldots, (-1)^{b_n} \theta_n ).
\end{eqnarray}
By the definition of $\mathfrak{q}'_G$, all of the values $\left| a_\mathbf{b} \right|$ are bounded
by $\mathfrak{q}'_G$ except possibly $\left| a_{00 \ldots 0} \right|$ and $\left| a_{11 \ldots 1} \right|$, which are
both bounded by $\mathfrak{q}_G = 1$. 

We claim that the matrix
\begin{eqnarray}
\mathbf{N} & = & \left[ \begin{array}{ccccccccc} &&&& a_{00 \ldots 0} - \delta_G \\
&&& a_{00 \ldots 1 } \\
&& \iddots \\
& a_{11 \ldots 0} \\
 a_{11 \ldots 1} - \delta_G
\end{array} \right]
\end{eqnarray}
which arises from subtracting $\delta_G$ from the two corner entries of $\mathbf{M}$,
has operator norm less than or equal to $1 - \delta_G$.  Indeed,
the operator norm of this Hermitian matrix is the maximum of the
absolute values of its entries, and we already know that all of its entries other
than its corner entries are bounded by $\mathfrak{q}'_G \leq 1 - \delta_G$.  To show
that that the absolute values of the corner entries are bounded by
$1 - \delta_G$, it suffices to write them out in terms of
the parameters of the game $G$: we have
\begin{eqnarray}
\left| a_{00\ldots0} - \delta_G \right| & = & \left| P_G ( \zeta_1 , \ldots, \zeta_n ) - \delta_G \right|  \\
& = & \left| \left( \sum_{\mathbf{i} \in \{ 0, 1 \}^n} \eta_\mathbf{i} p_\mathbf{i} \zeta_1^{i_1} 
\zeta_2^{i_2} \cdot 
\ldots \zeta_n^{i_n} \right) - \delta_G \right| \\
& = & \left| ( p_\mathbf{0}  - \delta_G ) + \sum_{\mathbf{i} \neq \mathbf{0}} \eta_\mathbf{i} p_\mathbf{i} \zeta_1^{i_1} 
\zeta_2^{i_2} \cdot 
\ldots \zeta_n^{i_n} \right| \\
& \leq & (p_0 - \delta_G ) + \sum_{\mathbf{i} \neq \mathbf{0}} p_\mathbf{i} \\
& = & 1 - \delta_G,
\end{eqnarray}
and likewise for $(a_{11\ldots1} - \delta_G)$.  We conclude that $\mathbf{N}$
has operator norm less than or equal to $1 - \delta_G$.  Let $\mathbf{M}''
= \mathbf{N}/(1 - \delta_G)$ and $\mathbf{M}' = (\mathbf{M} - \mathbf{N}')/\delta_G$,
and the desired conditions hold.

The proof for the case in which $W_1 , \ldots, W_n$ are of arbitrary dimension
follows by similar reasoning.
\end{proof}

\begin{theorem}
\label{generaldecthm}
Let $G = \left( \left\{ p_\mathbf{i} \right\} , \left\{ \eta_\mathbf{i} \right\} \right)$
be a strong self-test which is positively aligned.  Then, there 
exist $\delta_G > 0$ and $\alpha \in \mathbb{C}$ with $| \alpha | = 1$
such that the following holds.  Let $( \Phi , \{ M_j^{(i)} \} )$ be
a quantum strategy whose
measurements
are in canonical form
with underlying space
$(\mathbb{C}^2 \otimes W_1 ) \otimes \ldots \otimes ( \mathbb{C}^2 \otimes
W_n )$.  Then the scoring operator $\mathbf{M}$ can be decomposed
as
\begin{eqnarray}
\mathbf{M} = \delta_G \mathbf{M}' + ( \mathfrak{q}_G - \delta_G ) \mathbf{M}'',
\end{eqnarray}
where $\left\| \mathbf{M}'' \right\| \leq 1$, and
\begin{eqnarray}
\mathbf{M}' & = & (g_\alpha g_\alpha^* - g_{-\alpha} g_{-\alpha}^* )
\otimes \mathbb{I}_{W_1 \otimes \cdots \otimes W_n }.
\end{eqnarray}
\end{theorem}

\begin{proof}
We repeat elements of the proof of Theorem~\ref{winnabledecompthm}.
It suffices to prove our desired decomposition for the case in
which $\dim W_i = 1$ for all $i$.   Let
$\mathfrak{q}'_G$ be the maximum value of $Z_G$ that occurs on
the set $[-\pi , \pi ]^{n+1} \smallsetminus ( T^+ \cup T^- )$
(where $T^+$ and $T^-$ are defined by (\ref{tplus}) and (\ref{tminus})).
Let $\gamma \neq 0$ be the constant
that is given by Corollary~\ref{displacedcor},
and let
\begin{eqnarray}
\delta_G & = & \min \{ | \gamma | , \mathfrak{q}_G - \mathfrak{q}'_G \}.
\end{eqnarray}

We have
\begin{eqnarray}
\mathbf{M} & = & \left[ \begin{array}{ccccccccc} &&&& a_{00 \ldots 0} \\
&&& a_{00 \ldots 1 } \\
&& \iddots \\
& a_{11 \ldots 0} \\
 a_{11 \ldots 1}
\end{array} \right]
\end{eqnarray}
where 
\begin{eqnarray}
a_{b_1, \ldots, b_n} & = & P_G ( \zeta_1^{(-1)^{b_1}} , \zeta_2^{(-1)^{b_2}} , \ldots , \zeta_n^{(-1)^{b_n}} ).
\end{eqnarray}
for some $\zeta_1, \ldots, \zeta_n \in \mathbb{C}$ such that $| \zeta_i | = 1$ and
$\Im ( \zeta_i ) \geq 0$. 
By Corollary~\ref{displacedcor},
\begin{eqnarray}
\left| P_G ( \zeta_1, \ldots, \zeta_n ) - \gamma \right| + \left| \gamma \right| \leq \mathfrak{q}_G,
\end{eqnarray}
and it is easy to see (by the triangle inequality) that for any $ c \in [0, 1 ]$,
\begin{eqnarray}
\left| P_G ( \zeta_1, \ldots, \zeta_n ) - c \gamma \right| + \left| c \gamma \right| \leq \mathfrak{q}_G.
\end{eqnarray}

Let
\begin{eqnarray}
\mathbf{N} & = & \left[ \begin{array}{ccccccccc} &&&& a_{00 \ldots 0} - \frac{\delta_G}{| \gamma|}  \cdot \gamma \\
&&& a_{00 \ldots 1 } \\
&& \iddots \\
& a_{11 \ldots 0} \\
 a_{11 \ldots 1} - \frac{\delta_G}{| \gamma|}  \cdot \overline{\gamma}
\end{array} \right]
\end{eqnarray}
The absolute values of the corner entries of this matrix are less than
or equal to $\mathfrak{q}_G - \delta_G$, and the
other entries have absolute values less than or equal
to $\mathfrak{q}'_G \leq \mathfrak{q}_G - \delta_G$.  Thus when we set
\begin{eqnarray}
\alpha & = & \gamma / | \gamma|, \\
\mathbf{M}' & = & (g_\alpha g_\alpha^* - g_{-\alpha} g_{-\alpha}^* )
\otimes \mathbb{I}_{W_1 \otimes \cdots \otimes W_n }, \\
\mathbf{M}'' & = & ( \mathbf{M} - \delta_G \mathbf{M}' ) / (\mathfrak{q}_G - \delta_G),
\end{eqnarray}
the desired result follows.
\end{proof}

The operator $(g_\alpha g_\alpha^* - g_{-\alpha} g_{-\alpha}^*)$ from the
statement of Theorem~\ref{generaldecthm} does not describe a projective
measurement.  It is convenient to have a decomposition theorem
involving a projective measurement.  This motivates the next result.

We introduce some additional notation.
Let
\begin{eqnarray}
\mathbf{b} \colon \{ 0, 1, 2, \ldots, 2^n - 1 \} \to \{ 0, 1 \}^n
\end{eqnarray}
be the function which maps $k$ to its base-$2$ representation.  For
any $\zeta \in \mathbb{C}$ with $| \zeta | = 1$, and any
$k \in \{ 0, 1, 2, \ldots, 2^n -1  \}$, let
\begin{eqnarray}
g_{\zeta, k} = \frac{1}{\sqrt{2}} \left( \left| \mathbf{b}(k) \right> \left<
\mathbf{b} (k) \right| + \zeta \left| \overline{\mathbf{b}(k)} \right> \left<
\overline{\mathbf{b}(k)} \right| \right).
\end{eqnarray}

\begin{theorem}
\label{verygeneraldecompthm}
Let $G = \left( \left\{ p_\mathbf{i} \right\} , \left\{ \eta_\mathbf{i} \right\} \right)$
be a strong self-test which is positively aligned.  Then, there 
exist $\delta_G > 0$ and $\alpha \in \mathbb{C}$ with $| \alpha | = 1$
such that the following holds.  Let $( \Phi , \{ M_j^{(i)} \} )$ be
a quantum strategy whose
measurements
are in canonical form
with underlying space
$(\mathbb{C}^2 \otimes W_1 ) \otimes \ldots \otimes ( \mathbb{C}^2 \otimes
W_n )$.  Let $\alpha_0 = \alpha$ and let $\alpha_1, \ldots, \alpha_{2^{n-1}-1}$
be any unit-length complex numbers.
Then the scoring operator $\mathbf{M}$ can be decomposed
as
\begin{eqnarray}
\mathbf{M} = \delta_G \mathbf{M}' + ( \mathfrak{q}_G - \delta_G ) \mathbf{M}'',
\end{eqnarray}
where $\left\| \mathbf{M}'' \right\| \leq 1$, and
\begin{eqnarray}
\mathbf{M}' & = & \left[ \sum_{k=0}^{2^{n-1}-1}
(g_{\alpha_k, k} g_{\alpha_k, k}^* - g_{-\alpha_k, k} g_{-\alpha_k, k}^* ) \right]
\otimes \mathbb{I}_{W_1 \otimes \cdots \otimes W_n }.
\end{eqnarray}
\end{theorem}

\begin{proof}
Again it suffices to prove this result for when $\dim W_i = 1$ for all $i$.
Let $\mathfrak{q}'_G$ be the maximum value of $Z_G$ that occurs on
the set $[-\pi , \pi ]^{n+1} \smallsetminus ( T^+ \cup T^- )$,
where $T^+$ and $T^-$ are defined by (\ref{tplus}) and (\ref{tminus}).
Let $\gamma$ be the constant given by Corollary~\ref{displacedcor},
let $\alpha = \gamma/|\gamma|$, and let
\begin{eqnarray}
\delta_G & = & \min \{ \left| \gamma \right| , \left( \mathfrak{q}_G - \mathfrak{q}'_G \right)/2 \}.
\end{eqnarray}
Write $\mathbf{M}$ as
\begin{eqnarray}
\mathbf{M} & = & \left[ \begin{array}{ccccccccc} &&&& a_{00 \ldots 0} \\
&&& a_{00 \ldots 1 } \\
&& \iddots \\
& a_{11 \ldots 0} \\
 a_{11 \ldots 1}
\end{array} \right]
\end{eqnarray}
Let
\begin{eqnarray}
\mathbf{N} & = & \mathbf{M} - \delta_G\left[ \begin{array}{cccccccc}
&&&&&&& \alpha_0 \\
&&&&&& \alpha_1 \\
&&&&& \iddots \\
&&&& \alpha_{2^{n-1}-1} \\
&&& \overline{\alpha_{2^{n-1}-1}}  \\
&& \iddots  \\
& \overline{\alpha_1}  \\
 \overline{\alpha_0}  \\
\end{array} \right],
\end{eqnarray}
The corner entries of $\mathbf{N}$
have absolute value $\leq \mathfrak{q}_G - \delta_G$ (by Corollary~\ref{displacedcor})
and the same holds for the other anti-diagonal entries by the triangle inequality:
for any $n \in \{ 1, 2, \ldots, 2^{N-1} - 1 \}$,
\begin{eqnarray}
\left| a_{\mathbf{b}(n)} - \delta_G \alpha_n \right| \leq \left| a_{\mathbf{b}(n)} \right|
+ \delta_G \leq \mathfrak{q}'_G + (\mathfrak{q}_G - \mathfrak{q}'_G)/2 \leq \mathfrak{q}_G - \delta_G.
\end{eqnarray}
Thus we let $\mathbf{M}''/(\mathfrak{q}_G - \delta_G )$ and the desired statements hold.
\end{proof}

\section{Randomness Expansion with Partially Trusted Measurements}

\label{REpartialsec}

The goals of this section are to define randomness expansion protocols based
on \textbf{partially} trusted devices, and then to relate
these new protocols to Protocol R.

\subsection{Devices with Trusted Measurements}

\label{trustedsec}
\label{untrustedsec}

\label{simpleprotocolsubsec}

We begin by stating a simple protocol
that involves a device with trusted measurements.

\begin{definition}
A \textbf{device with trusted measurements} consists
of the following data.
\begin{enumerate}
\item A single quantum system $Q$ in an initial state $\Phi$.

\item For every pair $(\mathbf{i}, \mathbf{o} )$ of binary strings of equal length,
two Hermitian operators $M_{\mathbf{i} , \mathbf{o} }^{(0)}, M_{\mathbf{i} , \mathbf{o} }^{(1)}$
representing the measurements performed on $Q$ when the input 
and output histories are $\mathbf{i}$ and $\mathbf{o}$.  These operators
are assumed to satisfy
\begin{eqnarray}
\label{anticom1}
\left( M_{\mathbf{i} , \mathbf{o} }^{(0)}\right)^2 = \left( M_{\mathbf{i} , \mathbf{o} }^{(1)} \right)^2
= \mathbb{I}
\end{eqnarray}
and
\begin{eqnarray}
\label{anticom2}
 M_{\mathbf{i} , \mathbf{o} }^{(0)}
 M_{\mathbf{i} , \mathbf{o}}^{(1)}
= - M_{\mathbf{i} , \mathbf{o}}^{(1)}
M_{\mathbf{i} , \mathbf{o}}^{(0)}.
\end{eqnarray}
\end{enumerate}
\end{definition}
A trusted measurement device is one whose measurements perfectly
anti-commute.
A protocol for trusted measurement devices is given in Figure~\ref{protafig}.
Essentially this protocol is the same as Protocol R, except
that we have skipped the process of generating random
inputs for the game rounds, and have instead simply used the
biased coin flip $g$ itself as input to the device.

\begin{figure}
\setlength{\fboxsep}{10pt}
\fbox{\parbox{5.2in}{
\textbf{Protocol A:}
\vskip0.1in
\textit{Arguments:} 
\[
\begin{array}{rcl}
N & = & \textnormal{positive integer} \\
q & \in & (0, 1) \\
\eta & \in & (0, 1/2) \\
D & = & \textnormal{device with trusted measurements}
\end{array}
\]

\begin{enumerate}
\item A bit $g \in \{ 0, 1 \}$ is chosen
according to a biased $(1 - q, q )$ distribution.
The bit $g$ is given to $D$ as input, and an output
bit $o$ is recorded.
\vskip0.1in
\item If $g = 1$ and the output given by $D$ is
$0$, then the event $P$ (``pass'') is recorded.
If $g = 1$ and the output is $1$, the
event $F$ (``fail'') is recorded.
\vskip0.1in
\item If $g = 0$ and the output given by $D$ is
$0$, then the event $H$ (``heads'') is recorded.
If $g = 0$ and the output is $1$, the
event $T$ (``tails'') is recorded.
\vskip0.1in
\item Steps $1-3$ are repeated $N-1$ (more) times.
Bit sequences $\mathbf{g} = (g_1, \ldots, g_N)$
and $\mathbf{o} = (o_1, \ldots, o_N)$ are obtained.
\vskip0.1in
\item If the total number of failures is more than $\eta q N$,
the protocol \textbf{aborts}.  Otherwise,
the protocol \textbf{succeeds}.
If the protocol succeeds, it outputs
the bit sequences $\mathbf{g}$ and $\mathbf{o}$.
\end{enumerate}
}}
\caption{A randomness expansion protocol for a 
trusted measurement device.}
\label{protafig}
\end{figure}

\subsection{Devices with Partially Trusted Measurements}

\label{partialtrustsubsec}

\begin{definition}
\label{partiallytrusteddef}
Let $v \in (0, 1]$ and $h \in [0, 1]$ be real numbers
such that $v + h \leq 1$.
Then a \textbf{partially trusted device with parameters $(v, h)$} consists
of the following data.
\begin{enumerate}
\item A single quantum system $Q$ in an intial state $\Phi$.

\item For every pair $(\mathbf{i}, \mathbf{o} )$ of binary strings of equal length,
two Hermitian operators $M_{\mathbf{i} , \mathbf{o} }^{(0)}, M_{\mathbf{i} , \mathbf{o} }^{(1)}$
on ${Q}$
(representing measurements) that satisfy the following conditions:
\begin{itemize}
\item There exist perfectly anti-commuting measurement pairs
$(T_{\mathbf{i}, \mathbf{o}}^{(0)} , T_{\mathbf{i}, \mathbf{o}}^{(1)} )$
such that $M_{\mathbf{i} , \mathbf{o}}^{(0)} = T_{\mathbf{i}, \mathbf{o}}^{(0)}$
for all $\mathbf{i}, \mathbf{o}$, and 

\item The operator $M_{\mathbf{i}, \mathbf{o}}^{(1)}$ decomposes as
\begin{eqnarray}
\label{partiallytrustdecomp}
M_{\mathbf{i} , \mathbf{o} }^{(1)} & = & (v ) T_{\mathbf{i} , \mathbf{o}}^{(1)}
+ (1 - v - h) N_{\mathbf{i} , \mathbf{o}}
\end{eqnarray}
with $\left\| N_{\mathbf{i} , \mathbf{o}} \right\| \leq 1$.
\end{itemize}
\end{enumerate}
\end{definition}

The operators $M_{\mathbf{i}, \mathbf{o}}^{(0)}, M_{\mathbf{i},
\mathbf{o}}^{(1)}$ determine the measurements 
performed by the device on inputs $0$ and $1$, respectively.
Intuitively, a partially trusted device is a device $D$ which
always performs a trusted measurement $T^{(0)}$ on input $0$,
and on input $1$, selects one of the three
operators $(T^{(1)}, N, 0)$ at random according to the probability
distribution $(v, 1 - v - h , h )$. 

We will call the parameter $v$ the \textbf{trust coefficient}, 
and we will call $h$ the \textbf{coin flip coefficient}.
The parameter $h$ measures the extent to which
the output of $D$ on input $1$ is determined by a fair 
coin flip.
Note that when the input to the device $D$ is $1$,
then the probability that $D$ gives an output
of $1$ is necessarily between $h/2$ and $(1 - h/2)$.  

Figure~\ref{protaprimefig} gives a randomness
expansion protocol for partially trusted devices.  It is the
same as Protocol A, except that the
trusted device has been replaced by a partially trusted device.

\begin{figure}
\setlength{\fboxsep}{10pt}
\fbox{\parbox{5.2in}{
\textbf{Protocol A':}
\vskip0.1in
\textit{Arguments:} 
\[
\begin{array}{rcl}
v & = & \textnormal{real number such that } v \in  (0, 1]. \\
h & = & \textnormal{real number such that } h \in [0, 1 - v]. \\
N& = & \textnormal{ positive integer} \\
q & \in & (0, 1) \\
\eta & \in & (0, v/2) \\
D & = & \textnormal{partially trusted device with parameters } (v, h ). \\
\end{array}
\]
\begin{enumerate}
\item A bit $g \in \{ 0, 1 \}$ is chosen
according to a biased $(1 - q, q )$ distribution.
The bit $g$ is given to $D$ as input, and the output
bit $o$ is recorded.
\vskip0.1in
\item If $g = 1$ and the output given by $D$ is
$0$, then the event $P$ (``pass'') is recorded.
If $g = 1$ and the output is $1$, the
event $F$ (``fail'') is recorded.
\vskip0.1in
\item If $g = 0$ and the output given by $D$ is
$0$, then the event $H$ (``heads'') is recorded.
If $g = 0$ and the output is $0$, the
event $T$ (``tails'') is recorded.
\vskip0.1in
\item Steps $1-3$ are repeated $N-1$ (more) times.
Bit sequences $\mathbf{g} = (g_1, \ldots, g_N )$ and
$\mathbf{o} = (o_1, \ldots, o_N )$ are obtained.
\vskip0.1in
\item If the total number of 
failures is greater than
$(h/2 + \eta)qN$, then the protocol \textbf{aborts}.  
Otherwise,
the protocol \textbf{succeeds}.
If the protocol succeeds, it outputs
the bit sequences $\mathbf{g}$ and $\mathbf{o}$.
\end{enumerate}
}}
\caption{A randomness expansion protocol for a partially trusted
device.}
\label{protaprimefig}
\end{figure}

\subsection{Entanglement with a Partially Trusted
Measurement Device}

\label{entpartialsubsec}

Suppose that $D$ is a partially trusted measurement device (see
Definition~\ref{partiallytrusteddef}) with parameters $(v, h)$.  Suppose
that $E$ is a quantum system that is entangled with $D$, and let $\rho =
\rho_E$ denote the initial state of $E$.  We will use the following
notation: let $\rho_+$ and $\rho_-$ denote the subnormalized operators
which represent the states of $E$ when the input bit is $0$ and the output bit
is $0$ or $1$, respectively.  Let $\rho_P$ and $\rho_F$ denote
the operators which represent an input of $1$ and an output of $0$ or $1$,
respectively.  Also (using notation from Definition~\ref{partiallytrusteddef}),
let us write $\rho_0$ and $\rho_1$ denote the states of $E$ that would
occur if the trusted measurement $T^{(1)}$ was applied to $Q$ (instead
of the partially trusted measurement $M^{(1)}$).  (Note that $T^{(1)}$
is perfectly anticommuting with $M^{(0)}$.)

The following proposition expresses the possible behavior of the system $E$.

\begin{proposition}
\label{entpartialprop}
Let $v \in (0, 1]$ and $h \in [0, 1]$ be such that $v + h \leq 1$.  Let
$D$ be a partially trusted device with parameters $(v,h)$, let $E$
be a quantum system that is entangled with $D$, and 
let $\rho = \rho_E$.  Then,
\begin{eqnarray}
(h/2) \rho + v \rho_0 \leq \rho_P \leq (1 - h/2) \rho - v \rho_1
\end{eqnarray}
and
\begin{eqnarray}
(h/2) \rho + v \rho_1 \leq \rho_F \leq (1 - h/2) \rho - v \rho_0.
\end{eqnarray}
\end{proposition}

\begin{proof}
Let $N$ be the measurement operator from the decomposition
of $M^{(1)}$ given in Definition~\ref{partiallytrusteddef}.  Let
$\rho'$ be the subnormalized
operator on $E$ which denotes the state that would be produced
if $N$ were applied to $Q$ and the outcome were $0$.  
Clearly, $0 \leq \rho' \leq \rho$.  
From the decomposition (\ref{partiallytrustdecomp}), $\rho_P$
is a convex combination of the operators $\rho_0$, $\rho'$
and $(\rho/2)$:
\begin{eqnarray}
\rho_P & = & v \rho_0 + (1 - v - h ) \rho' + h(\rho/2).
\end{eqnarray}
Since $\rho' \leq \rho$, we have
\begin{eqnarray}
\rho_P & \leq & v \rho_0 + (1 - v - h ) \rho + h(\rho/2) \\
& = & v \rho_0 + (1 - v - h/2 ) \rho \\
& = & (1 - h/2) \rho + v(\rho_0 - \rho) \\
& = & (1 - h/2) \rho - v \rho_1.
\end{eqnarray}
The other inequalities follow similarly.
\end{proof}

\subsection{Simulation}

To any binary XOR game $G$, we have associated three quantities:
$\mathfrak{q}_G, \mathbf{w}_G$, and $\mathbf{f}_G$. These
are respectively the optimal quantum score, optimal quantum winning
probability, and least quantum failure probability for $G$.
The quantities are related by $\mathbf{w}_G =
(1 + \mathfrak{q}_G)/2$ and $\mathbf{f}_G = 1 - \mathbf{w}_G$.

\begin{theorem}
\label{partialtrustsimthm}
For any $n$-player strong self-test $G$ which is positively aligned, there exists
$\delta_G > 0$ such that the following holds.  For any 
any $n$-part binary quantum device $D$, there exists
a partially trusted device $D'$ with parameters $\mathfrak{q}_G,
\delta_G$ such that Protocol A' (with arguments $\delta_G, 2 \mathbf{f}_G, N , q, \eta, D'$) simulates Protocol R (with arguments $N, \eta, q, G, D$).
\end{theorem}

\begin{proof}
Choose $\delta_G$ according to Theorem~\ref{verygeneraldecompthm}.

Consider the behavior of the device $D$ in the first round.  We may assume
that the measurements performed by $D_1 , \ldots, D_n$ are in 
canonical form.  Write the underlying space as $(\mathbb{C}^2 \otimes W_1 )
\otimes \cdots \otimes (\mathbb{C}^2 \otimes W_n)$.  If $g = 0$,
the measurement performed by $D_1$ is given by the operator
\begin{eqnarray}
\label{doneop}
\left[ \begin{array}{cccc|cccc}
&&&&1\\
&&&&& 1\\
&&&&& & \ddots \\
&&&&& & & 1 \\
\hline
1&&&& \\
&1&&&& \\
&&\ddots &&&\\
&&&1
\end{array} \right] \otimes \mathbb{I}_{W_1 \otimes \cdots \otimes W_n }
\end{eqnarray}
(where the matrix on the left is an operator on $\left( \mathbb{C}^2 \right)^{\otimes n }$,
with the basis taken in lexiographic order as usual).  

If $g = 1$ the measurement performed by $D$ is given by the scoring operator
$\mathbf{M}$.
Theorem~\ref{verygeneraldecompthm} guarantees that for some unit-length
complex number $\alpha$, and for any choices of unit-length complex
numbers $\alpha_1, \ldots, \alpha_{2^{n-1}-1}$, there is a decomposition
for $\mathbf{M}$ in the form $\mathbf{M} = \delta_G \mathbf{M}' +
(\mathfrak{q}_G - \delta_G ) \mathbf{M}''$ with
\begin{eqnarray*}
\mathbf{M}' & = & \left[ 
\begin{array}{cccc|cccc}
&&&&&&& \alpha \\
&&&&&& \alpha_1 \\
&&&&& \iddots \\
&&&& \alpha_{2^{n-1}-1}\\
\hline
&&& \overline{\alpha}_{2^{n-1}-1} &&&\\
&& \iddots &&&\\
& \overline{\alpha_1}&&& \\
\overline{\alpha}&&&&&
\end{array} \right] \otimes \mathbb{I}_{W_1 \otimes \cdots \otimes W_n}
\end{eqnarray*}
and $\left\| \mathbf{M}'' \right\| \leq 1$.
To simulate the behavior of $D$ with a partially trusted
device, we need only choose $\alpha_1, \ldots, \alpha_{2^{n-1}-1}$
so that $\mathbf{M}'$ is perfectly anti-commutative
with the operator~\ref{doneop}.    This can be done,
for example, by setting $\alpha_1 , \alpha_2, \ldots, \alpha_{2^{n-2}-1}$
to be equal to $\alpha$, and $\alpha_{2^{n-2}} , \ldots , \alpha_{2^{n-1}+1}$
to be equal to $-\alpha$.  Thus the behavior of the device $D$
in the first round of Protocol R can be simulated by a partially trusted device
with parameters $(\delta_G, 1 - \mathfrak{q}_G ) =
(\delta_G , 2 \mathbf{f}_G)$.
Similar reasoning shows the desired simulation result across all rounds.
\end{proof}

The following corollary is easy to prove.
\begin{corollary}
\label{partialtrustsimcor}
Theorem~\ref{partialtrustsimthm} holds true without the assumption
that $G$ is positively aligned. $\qed$
\end{corollary}

Essentially, the above corollary implies that any security result
for Protocol A' can be converted immediately into an
identical security result
for Protocol R.  This will be the basis for our eventual
full proof of randomness expansion.

\section{The Proof of Security for Partially Trusted Devices}

\label{partialtrustapp}

In this section we provide the proof of security for Protocol A'
(see Figure~\ref{protaprimefig}).
Our approach, broadly stated, is as follows: we show
the existence of a function $T ( v, h, \eta, q, \kappa)$
which provides a lower bound on the linear rate of entropy
of the protocol.  (The variables $v, h, \eta, q$ are
from the protocol, and $\kappa$ is a positive constant
that can be chosen to be arbitrarily small.)  The main point of our proofs is that,
although $T$ depends on several variables, it does
\textit{not} depend on the particular device used in Protocol A'.
Thus, we have a uniform security result.

The definition of $T$ is multi-layered and is developed over the course
of the section.  For the reader's convenience, we have collected all the
definitions of the functions that we use, including $T$, in appendix subsection~\ref{varfuncapp}.  The full expression for $T$ is quite complicated,
but for our purposes it suffices to calculate the limit
$\lim_{(q, \kappa) \to (0, 0)} T ( v, h, \eta, q, \kappa )$, since
this will tell us what rate Protocol $A'$ approaches when $q$ is small.
This limit will be shown to be equal to
$\pi ( \eta / v )$, where $\pi$ denotes the function
from Theorem~\ref{uncertaintythm}.

Our proof involves several parameters.  For convenience,
we include a table here which assigns a name
to each parameter (Figure~\ref{parametersfig}.)

\begin{figure}
\begin{center}
\begin{tabular}{| c c l | l | }
\hline
$N$ & $\in$ & $\mathbb{N}$ & number of rounds \\
\hline
$q$ & $\in$ & $( 0 , 1)$ & test probability \\
\hline
$t$ & $\in$ & $[ 0 , 1]$ & failure parameter \\
\hline
$v$ &$\in$ & $(0, 1]$
&  trust coefficient \\
\hline
$h$ & $\in$ & $[ 0, 1 - v]$ &
coin flip coefficient \\
\hline 
$\eta$ & $\in$ & $( 0, v/2)$ & error tolerance \\
\hline
$\kappa$ & $\in $ & $(0 , \infty)$ & failure penalty \\
\hline
$r$ & $\in$ & $( 0, 1/(q \kappa ) ]$ & multiplier for R{\'e}nyi coefficient \\
\hline
$\epsilon$ & $\in$ & $(0, \sqrt{2} ]$ & error parameter for smooth min-entropy \\
\hline
\end{tabular}
\caption{Variables used in section~\ref{partialtrustapp}.}
\label{parametersfig}
\end{center}
\end{figure}

To avoid unnecessary repetition, we will use the following conventions
in this section.

\begin{itemize}
\item Unless otherwise stated, we will assume that the variables
from Figure~\ref{parametersfig} are always restricted to the domains
given.  (The reader can assume that all unquantified statements
are prefaced by, ``for all $q \in (0, 1 ]$, all
$\epsilon \in (0, \sqrt{2})$,'' etc.)
If we say ``$F ( q , \kappa )$
is a real-valued function,'' we mean that it is a real
valued function on $(0, 1) \times (0, \infty )$.
If we say ``let $x = \kappa q$,'' we mean
that $x$ is a real valued function on $(0, 1 ) \times
(0, \infty )$ defined by $x ( \kappa , q ) = \kappa q$.
If the domain of one parameter
of a function depends on another variable
(as can occur, e.g., for the variable $h$)
we always include the other
variable as a parameter of the function.

\item When we discuss a single iteration
of Protocol A', will use notation from subsection~\ref{entpartialsubsec}:
If $D$ is a partially trusted measurement device, and $E$
is a purifying system for $D$ with initial state $\rho = \rho_E$,
then $\rho = \rho_H + \rho_T$ and $\rho = \rho_P + \rho_F$ 
denote the decompositions that occur for a single use of
the device on input $0$ and $1$, respectively.  We denote
by $\rho_+, \rho_-, \rho_0, \rho_1$ the respective states that would occur
if the corresponding fully trusted measurements were used instead.
(Note that $\rho_H = \rho_+$ and $\rho_T = \rho_-$.)  Let
$\overline{\rho}$ denote the operator on ${E} \oplus
{E} \oplus {E} \oplus {E}$ given
by
\begin{eqnarray}
\overline{\rho} = (1-q) \rho_H \oplus (1-q) \rho_T \oplus q \rho_P \oplus q \rho_F.
\end{eqnarray}
This operator represents the state of $E$ taken together with the input
bit and output bit from the first iteration of Protocol A'.

\item When we discuss multiple iterations of Protocol A',
we will use the following notation: let $G$ and $O$
denote classical registers which consist
of the bit sequences $\mathbf{g} = (g_1, \ldots, g_N)$ and $\mathbf{o} = (o_1, \ldots, o_n)$,
respectvely.  We denote basis states for the joint
system $GO$ by $\left| \mathbf{g} \mathbf{o} \right>$.
We denote the joint state of the system $EGO$
at the conclusion of Protocol A' by $\Gamma_{EGO}$.

\item If $D$ is a partially trusted measurement device, $E$
is a purifying system, and $\alpha > 0$, then we refer to
the quantity
\begin{eqnarray}
\frac{\Tr ( \rho_1^\alpha ) }{\Tr ( \rho^\alpha ) } \in [0, 1]
\end{eqnarray}
as the \textbf{$\alpha$-failure parameter} of $D$.  (Note that
we used the operator $\rho_1$ in the above expression, \textit{not}
the operator $\rho_F$.  This parameter measures ``honest'' failures
only.)

\item Let $\Pi ( x, y)$ and $\pi ( y)$ denote the functions from
Theorem~\ref{uncertaintythm}.
\end{itemize}

\subsection{Proof Idea}

Let $D$ be a partially trusted measurement device with parameters $v,h$,
and let $E$ be a purifying system with initial state $\rho$.  Let $\overline{\rho}$
be the operator on ${E} \oplus {E} \oplus {E}
\oplus {E}$ which represents the joint state of $E$ together with the
input and output of a single iteration of Protocol A':
\begin{eqnarray}
\overline{\rho} = (1-q) \rho_H \oplus (1-q) \rho_T \oplus q \rho_P \oplus q \rho_F.
\end{eqnarray}
We wish to show that the state $\overline{\rho}$ is more random than
the original state $\rho$.  Therefore, we wish to show that the ratio
\begin{eqnarray}
\label{thedivratio}
\frac{d_{1 + \gamma}( \overline{\rho} \| \overline{\sigma} )}{d_{1 + \gamma} ( \rho \| \sigma )},
\end{eqnarray}
for some appropriate $\gamma, \sigma, \overline{\sigma}$,
is significantly smaller than $1$.  For simplicity, we will for
the time being take $\sigma = \mathbb{I}$ for the initial bounding operator.
(Later in this section we will generalize this choice.)

A natural choice of bounding operator for $\overline{\rho}$ would be
\begin{eqnarray}
\label{initboundingop}
(1-q) \mathbb{I} \oplus (1-q) \mathbb{I}
\oplus q \mathbb{I} \oplus q \mathbb{I}.
\end{eqnarray}
Computing $d_{1 + \gamma} ( \overline{\rho} \| \cdot )$ with this bounding operator
would yield
\begin{eqnarray}
\label{primitivebounding}
\left\{ (1-q) \Tr [ \rho_+^{1 + \gamma} ] + (1-q) \Tr [ \rho_-^{1 + \gamma} ]
+ q \Tr [ \rho_P^{1 + \gamma} ]  + q \Tr [ \rho_F^{1 + \gamma} ] \right\}^{1/\gamma}
\end{eqnarray}
Computing this quantity 
would have the effect, roughly speaking, of measuring the randomness
of the output bit of Protocol A' conditioned on $E$ and on the input bit $g$.
However this is not adequate for our purposes, since it treats ``passing'' rounds
the same as ``failing'' rounds, and does not take into account that the
device is only allowed a limited number of failures.  (And indeed, this measurement
of randomness does not work: if $D$ performs anticommuting measurements
on a half of a maximally entangled qubit pair, the divergence quantity
$d_{1 + \gamma} ( \overline{\rho} \| \cdot )$ with bounding operator (\ref{initboundingop})
is the same as $d_{1 + \gamma} ( \rho \| \mathbb{I} )$.)

We will use a slightly different expression to measure the output
of Protocol A'.  We introduce a single cofficient $2^{-\kappa}$ (with $\kappa > 0$)
into the fourth term of the expression:
\begin{eqnarray}
\label{oneshotexp}
\left\{ (1-q) \Tr [ \rho_+^{1 + \gamma} ] + (1-q) \Tr [ \rho_-^{1 + \gamma} ]
+ q \Tr [ \rho_P^{1 + \gamma} ]  + q 2^{-\kappa} \Tr [ \rho_F^{1 + \gamma} ] \right\}^{1/\gamma}
\end{eqnarray}
The reason for the introduction of the coefficient
$2^{-\kappa}$ is this: in effect, if a game round occurs and the device
fails, we lower our expectation for the amount of randomness produced.
The quantity (\ref{oneshotexp}) is equal to $d_{1 + \gamma} ( \overline{\rho} \| \overline{\sigma} )$
where
\begin{eqnarray}
\overline{\sigma} & = & (1-q) \mathbb{I} \oplus (1-q) \mathbb{I} \oplus
q \mathbb{I} \oplus q 2^{\kappa/\gamma} \mathbb{I}.
\end{eqnarray}

Having chosen the bounding operator $\overline{\sigma}$, we need only
to choose the coefficient $\gamma \in (0, 1]$.  We will take $\gamma$ to be
of the form $\gamma = r q \kappa$, where
$r \in (0, 1/(q \kappa ) ]$.\footnote{The
reason for this choice of interval for $r$ is that we need $\gamma \leq 1$
for the application of results from section~\ref{renyisection}.}  (Expressing
$\gamma$ this way enables clean calculations in our proofs.)

The proof proceeds by showing an upper bound on (\ref{thedivratio}),
then applying induction to get a similar upper bound for $N$ uses
of the device, and then applying the relationship between
Renyi divergence and smooth min-entropy to get a lower
bound on the number of extractable bits produced by
Protocol $A'$.

\subsection{One-Shot Results}

We begin by proving a one-shot security result under the assumption
that some limited information about the device is available.

\begin{proposition}
\label{knowndeviceprop}
There is a continuous real-valued function $\Lambda ( v, h, q, \kappa, r , t)$ such that
the following conditions hold.
\begin{enumerate}
\item Let $D$ be a partially trusted measurement device
with parameters $(v,h)$, and let $E$ be a purifying system for $D$.
Let $\gamma = r q \kappa$, and let
\begin{eqnarray}
\overline{\sigma} & = & ( 1 - q ) \mathbb{I} \oplus (1 -q ) \mathbb{I}
\oplus q \mathbb{I} \oplus q 2^{\kappa / \gamma} \mathbb{I}.
\end{eqnarray}
Then,
\begin{eqnarray}
\label{oneshotdesineq}
d_{1 + \gamma} ( \overline{\rho}  \| \overline{\sigma} ) 
& \leq & 2^{- \Lambda ( v, h, q, \kappa, r , t )}  \cdot d_{1 + \gamma} 
( \rho \| \mathbb{I} ),
\end{eqnarray}
where $t = \Tr ( \rho_1^{1+\gamma} ) / \Tr ( \rho^{1 + \gamma } )$ denotes
the $(1 + \gamma)$-failure parameter of $D$.

\item The following limit condition is satisfied: for any 
$t_0 \in [0, 1]$,
\begin{eqnarray}
\label{lambdalimitcond}
\lim_{\substack{(q, \kappa, t ) \to  (0, 0, t_0)}} \Lambda ( v, h, q, \kappa, r, t )
& = & \pi ( t_0 ) + \frac{h /2 + v t_0}{r},
\end{eqnarray}
where $\pi$ is the function from Theorem~\ref{uncertaintythm}.
\end{enumerate}
\end{proposition}

\begin{proof}
We have
\begin{eqnarray}
\nonumber
& d_{1 + \gamma} ( \overline{\rho} \| \overline{\sigma} ) = &  \\ 
\label{initformulafordiv}
& \left\{ (1-q) \Tr [ \rho_+^{1 + \gamma} ] + (1-q) \Tr [ \rho_-^{1 + \gamma} ]
+ q \Tr [ \rho_P^{1 + \gamma} ]  + q 2^{-\kappa} \Tr [ \rho_F^{1 + \gamma} ] \right\}^{1/\gamma}
\end{eqnarray}
We will compute a bound on this quantity
by grouping the first and second summands together,
and then by grouping the third and fourth summands together.  Note
that by Theorem~\ref{uncertaintythm},
we have
\begin{eqnarray}
\label{summands12}
\Tr [ \rho_+^{1 + \gamma} ] +  \Tr [ \rho_-^{1 + \gamma} ] & \leq &
2^{- \gamma \Pi ( \gamma , t ) } \Tr [ \rho^{1+\gamma} ]
\end{eqnarray}

Now consider the sum $\Tr [ \rho_P^{1+ \gamma} ] + 2^{-\kappa} \Tr [ \rho_F^{1+\gamma} ]$.
By superaddivity (see Proposition~\ref{powmatrixprop}),
\begin{eqnarray}
\Tr [ \rho_P^{1+ \gamma} ] + 2^{-\kappa} \Tr [ \rho_F^{1+\gamma} ] & = &
\Tr \left[ 2^{-\kappa} (\rho_P^{1+\gamma} + \rho_F^{1+\gamma}) + (1 - 2^{-\kappa}) \rho_P^{1+\gamma} \right] \\
& \leq & \Tr \left[ 2^{-\kappa} \rho^{1+\gamma} + (1 - 2^{-\kappa} )\rho_P^{1+\gamma} \right].
\end{eqnarray}
By Proposition~\ref{entpartialprop},
\begin{eqnarray}
\Tr [ \rho_P^{1+ \gamma} ] + 2^{-\kappa} \Tr [ \rho_F^{1+\gamma} ] 
& \leq & \Tr \left\{ 2^{-\kappa} \rho^{1+\gamma} + (1 - 2^{-\kappa} ) [ \rho - (h/2) \rho - v \rho_1 ]^{1 + \gamma} \right\}.
\end{eqnarray}
Applying the rule $\Tr [(X - Y )^{1 + \gamma} ]
\leq \Tr [ X^{1+\gamma}] - \Tr [Y^{1+\gamma}]$, followed
by the fact that $\Tr [ \rho_1^{1+\gamma} ] =
t \Tr [\rho^{1+\gamma} ]$, we have
the following:
\begin{eqnarray}
\Tr [ \rho_P^{1+ \gamma} ] + 2^{-\kappa} \Tr [ \rho_F^{1+\gamma} ] 
\nonumber
& \leq & \Tr \left\{ 2^{-\kappa} \rho^{1+\gamma} + (1 - 2^{-\kappa} ) [ \rho^{1+\gamma} - (h/2)^{1+\gamma} \rho^{1 +\gamma} - v^{1+\gamma} \rho_1^{1 + \gamma} ] \right\} \\
\nonumber
& = & \Tr \left\{ 2^{-\kappa} \rho^{1+\gamma} + (1 - 2^{-\kappa} ) [ \rho^{1+\gamma} - (h/2)^{1+\gamma} \rho^{1 +\gamma} - v^{1+\gamma} t \rho^{1 + \gamma} ] \right\} \\
& = & \left\{ 2^{-\kappa} + (1 - 2^{-\kappa} ) [ 1 - (h/2)^{1+\gamma}  - v^{1+\gamma} t  ] \right\} \Tr [ \rho^{1+\gamma} ]  \\
& = & \left\{ 1 - (1 - 2^{-\kappa} ) [(h/2)^{1+\gamma} + v^{1+\gamma} t ]
\right\} \Tr [ \rho^{1+\gamma} ].
\label{summands34}
\end{eqnarray}
Combining (\ref{initformulafordiv}), (\ref{summands12}), and (\ref{summands34}), we
find the following: if we set
\begin{eqnarray*}
\lambda ( v, h, q, \kappa, r, t ) & = & \left( (1 - q) 2^{- \gamma \Pi ( \gamma , t ) }
+ q \left\{ 1 - (1 - 2^{-\kappa} ) [(h/2)^{1+\gamma} + v^{1+\gamma} t ]
\right\} \right)^{1/\gamma},
\end{eqnarray*}
then
\begin{eqnarray}
d_{1 + \gamma} ( \overline{\rho}  \| \overline{\sigma} ) 
& \leq & \lambda ( v, h, q, \kappa, r, t )  \cdot d_{1 + \gamma} 
( \rho \| \mathbb{I} ).
\end{eqnarray}
Therefore setting $\Lambda = - \log \lambda$ yields (\ref{oneshotdesineq}).

It remains for us to evaluate the limiting behavior of $\Lambda$
as $(q, \kappa, t) \to (0, 0, t_0)$.  We can rewrite
the formula for $\lambda$ as
\begin{eqnarray*}
\lambda ( v, h, q, \kappa, r, t ) & = & \left( 1
+ \left\{ (1 - q) (2^{- \gamma \Pi ( \gamma , t ) } - 1 )
+ q ( 2^{-\kappa} - 1) [(h/2)^{1+\gamma} + v^{1+\gamma} t ] \right\} \right)^{1/\gamma}
\end{eqnarray*}
Applying Proposition~\ref{etothecprop} to this
expression (with $g = \gamma$, and
$f$ equal to the function enclosed by braces),
we have
\begin{eqnarray*}
&& \ln \left[  \lim_{\substack{(q, \kappa, t ) \to  (0, 0, t_0 )}}
\lambda ( v, h, q, \kappa, r, t ) \right] \\
& = & 
\lim_{(q, \kappa, t) \to (0, 0, t_0)} \left\{ 
(1 - q ) \left( \frac{2^{-\gamma \Pi ( \gamma, t ) } - 1 }{\gamma} \right) +
\left( \frac{q ( 2^{-\kappa} - 1 ) }{\gamma}  \right)
[(h/2)^{1+\gamma} + v^{1+\gamma} t ] \right\} \\
& = & (1  ) (- \ln 2 ) \pi (t_0 ) 
+ (- \ln 2 ) (r^{-1}) [ (h/2) + vt_0 ],
\end{eqnarray*}
which implies (\ref{lambdalimitcond}) as desired.
\end{proof}

Proposition~\ref{knowndeviceprop} is not sufficient for our ultimate
proof of security because it assumes that additional information (beyond
the trust parameters $v,h$) is is known about the device $D$.  The next proposition
avoids this limitation.  (It makes no use of the failure parameters of
the device.)

\begin{proposition}
\label{oneshotprop}
There is a continuous real-valued function $\Delta ( v, h, q, \kappa, r )$ such that
the following conditions hold.
\begin{enumerate}
\item Let $D$ be a partially trusted measurement device
with parameters $(v,h)$, and let $E$ be a purtifying system for $D$.
Let $\gamma = r q \kappa$, and let
\begin{eqnarray}
\overline{\sigma} & = & ( 1 - q ) \mathbb{I} \oplus (1 -q ) \mathbb{I}
\oplus q \mathbb{I} \oplus q 2^{\kappa / \gamma} \mathbb{I}.
\end{eqnarray}
Then,
\begin{eqnarray}
\label{oneshotdesineq2}
d_{1 + \gamma} ( \overline{\rho}  \| \overline{\sigma} ) 
& \leq & 2^{- \Delta ( v, h, q, \kappa, r )}  \cdot d_{1 + \gamma} 
( \rho \| \mathbb{I} ).
\end{eqnarray}

\item The following limit condition is satisfied:
\begin{eqnarray}
\label{lambdalimitcond2}
\lim_{(q, \kappa ) \to (0, 0 )} \Delta ( v, h, q, \kappa, r)
& = & \min_{s \in [0,1]} \left( \pi ( s )+ \frac{ h/2 + vs}{r}  \right),
\end{eqnarray}
where $\pi$ is the function from Theorem~\ref{uncertaintythm}.
\end{enumerate}
\end{proposition}

\begin{proof}
Let $\Lambda$ be the function from Proposition~\ref{knowndeviceprop}, and
let
\begin{eqnarray}
\Delta ( v, h, q, \kappa, r ) & = & \min_{t \in [0, 1 ]} \Lambda ( v, h, q, \kappa, r, t ).
\end{eqnarray}
Clearly, (\ref{oneshotdesineq2}) holds by Proposition~\ref{knowndeviceprop}.  Equality
(\ref{lambdalimitcond2}) follows via Proposition~\ref{compactprop}.
\end{proof}

\subsection{Multi-Shot Results}

The goal of this subsection is to deduce consequences
of Proposition~\ref{oneshotprop} across multiple iterations.
Let $\Gamma_{EGO}$ denote the joint state of the
registers $E$, $G$, and $O$.  (Note that $\Gamma$
is a classical-quantum state with respect
to the partition $(GO | E )$.)

The following proposition follows immediately from 
Proposition~\ref{oneshotprop} by induction.

\begin{proposition}
\label{initmultishotprop}
Let $D$ be a partially trusted measurement device with
parameters $(v,w)$, and let $E$ be a purifying system
for $D$.  Let $\gamma = rq \kappa$, and let
$\Phi$ be the operator on ${E}
\otimes {G} \otimes {O}$ given
by
\begin{eqnarray}
\label{initsigmaop}
\Phi & = & \mathbb{I}_{{E}}
\otimes \left( \sum_{\mathbf{g},\mathbf{o}
\in \{ 0, 1 \}^N } (1 - q)^{\sum_i (1-g_i ) }
q^{\sum_i g_i } 2^{(
\sum_i g_i o_i)/(qr ) } \left| \mathbf{go} \right>
\left< \mathbf{go} \right| \right).
\end{eqnarray}
Then,
\begin{eqnarray}
\label{initmspropineq}
D_{1+ \gamma} \left( \Gamma_{EGO} \|
\Phi \right) \leq D_{1 + \gamma} \left(
\Gamma_E \| \mathbb{I} \right) -
N \cdot  \Delta ( v, h, q, \kappa, r),
\end{eqnarray}
where $\Delta$ denotes the function from Proposition~\ref{oneshotprop}. $\qed$
\end{proposition}

We note the significance of the exponents in
(\ref{initsigmaop}): the quantity $\sum_{i=1}^N (1-g_i)$
is the number of generation rounds that occured
in Protocol A', the quantity $\sum_{i=1}^N g_i$
is the number of game rounds, and the
quantity $\sum_{i=1}^N g_i o_i$ is the number
of times the ``failure'' event occurred during the protocol.

As stated, Proposition~\ref{initmultishotprop} is not
useful for bounding the randomness of
$\Gamma_{EGO}$ because the quantity
$D_{1+\gamma} ( \Gamma_E \| \mathbb{I} )$
could be arbitrarily large.  We therefore prove the following
alternate version of the proposition.  The statement
is the same, except that we replace $\mathbb{I}_{E}$
in (\ref{initsigmaop}) with $\Gamma_E$,
and we remove the term $D_{1+\gamma} ( \Gamma_E
\| \mathbb{I} )$ from (\ref{initmspropineq}).

\begin{proposition}
\label{multishotprop}
Let $D$ be a partially trusted measurement device with
parameters $(v,w)$, and let $E$ be a purifying system
for $D$.  Let $\gamma = rq \kappa$, and let
$\Sigma$ be the operator on ${E}
\otimes {G} \otimes {O}$ given
by
\begin{eqnarray}
\label{sigmaop}
\Sigma & = & \Gamma_E
\otimes \left( \sum_{\mathbf{g},\mathbf{o}
\in \{ 0, 1 \}^N } (1 - q)^{\sum_i (1-g_i ) }
q^{\sum_i g_i } 2^{(\sum_i g_i o_i)/(qr) } \left| \mathbf{go} \right>
\left< \mathbf{go} \right| \right).
\end{eqnarray}
Then,
\begin{eqnarray}
\label{mspropineq}
D_{1+ \gamma} \left( \Gamma_{EGO} \|
\Sigma \right) \leq -
N \cdot  \Delta ( v, h, q, \kappa, r),
\end{eqnarray}
where $\Delta$ denotes the function from Proposition~\ref{oneshotprop}.
\end{proposition}

\begin{proof}
Let $\Gamma = \Gamma_E$.
Let $(D, E')$ be the device-environment pair
that arises from taking the pair $(D, E)$ and applying
the stochastic operation
\begin{eqnarray}
X \mapsto \Gamma^{\frac{-\gamma}{2+2\gamma}}
X \Gamma^{\frac{-\gamma}{2+2\gamma}}
\end{eqnarray}
to the system $E$.  The state $\Gamma_{E'}$ of the resulting
system $E'$ satisfies
\begin{eqnarray}
\Gamma_{E'} & = & \frac{\Gamma^{1/(1+\gamma)}}{K},
\end{eqnarray}
where $K = \Tr ( \Gamma^{1/(1+\gamma)} )$.

By directly applying the definition of $D_\alpha$
(see Definition~\ref{renyimaindefinition}) we can see
that certain divergences of
$\Gamma_{EGO}$ and $\Gamma_{E'GO}$ can
be computed from one another:
\begin{eqnarray}
\label{paralleleqn1}
D_{1+ \gamma} ( \Gamma_{E'GO} \| \Phi )
& = & - \frac{1 + \gamma}{\gamma} \cdot \log K
+ D_{1+\gamma} ( \Gamma_{EGO} \| \Sigma ) \\
\label{paralleleqn2}
D_{1+\gamma} ( \Gamma_{E'} \| \mathbb{I} ) & = &
- \frac{1+\gamma}{\gamma} \cdot \log K
+ D_{1+\gamma} ( \Gamma_E \| \Gamma ).
\end{eqnarray}
Applying Proposition~\ref{initmultishotprop} to
$(D, E')$, we find that
\begin{eqnarray}
D_{1+ \gamma} ( \Gamma_{E'GO} \| \Phi ) - 
D_{1+\gamma} ( \Gamma_{E'} \| \mathbb{I} )
\leq - N \Delta ( v, h, q, \kappa , r ).
\end{eqnarray}
By (\ref{paralleleqn1})--(\ref{paralleleqn2}),
the same bound holds when $E', \Phi, \mathbb{I}$ are replaced
$E, \Sigma, \Gamma$.
Since $D_{1+\gamma} ( \Gamma_E \| \Gamma ) = 0$,
the desired inequality is obtained.
\end{proof}

The following corollary of Proposition~\ref{multishotprop}
provides final preparation
for the proof of the main result.

\begin{corollary}
\label{pushdowncor}
Let $\epsilon > 0$.
Then, there exists a positive semidefinite operator
$\overline{\Gamma}_{EGO}$ which
is classical with respect to the systems $E$ and $G$ 
such that
\begin{eqnarray}
\left\| \overline{\Gamma}_{EGO} - \Gamma_{EGO}
\right\|_1 \leq \epsilon
\end{eqnarray}
and
\begin{eqnarray}
D_{max} ( \overline{\Gamma}_{EGO} \|
\Sigma ) \leq - N \cdot \Delta ( v, h, q, \kappa, r )
+ \frac{\log ( 2 / \epsilon^2) }{q \kappa r}
\end{eqnarray}
(where $\Delta$ and $\Sigma$ are as in
Proposition~\ref{oneshotprop} and Proposition~\ref{multishotprop}, respectively).
\end{corollary}

\begin{proof}
This follows from Proposition~\ref{maxrenyiprop}.
\end{proof}

\subsection{The Security of Protocol A'}

Let $s$ denote the event
that Protocol A' succeeds, and
let $\Gamma^s_{EGO}$ denote
the corresponding (subnormalized)
operator on ${E} \otimes
{G} \otimes {O}$.

\begin{proposition}
\label{mainresultprop}
There exists a continuous real-valued
function $R(v, h, \eta, q, \kappa, r)$
such that the following holds.
\begin{enumerate}
\item Let $\epsilon > 0$.
If Protocol A' is executed with parameters
$(v, h, N, q, \eta,  D)$, then
\begin{eqnarray}
H_{min}^\epsilon ( \Gamma_{EGO}^s \mid EG )
\geq N \cdot R ( v, h, \eta , q, \kappa, r )
- \frac{\log ( 2/\epsilon^2 ) }{q \kappa r}.
\end{eqnarray}

\item The following equality holds:
\begin{eqnarray}
\lim_{(q, \kappa) \to (0, 0)}
R ( v, h, \eta, q, \kappa, r )
& = & \min_{s \in [0, 1 ]} \left[ \pi ( s ) + \frac{vs - \eta}{r}
\right]
\end{eqnarray}
\end{enumerate}
\end{proposition}

\begin{proof}
The ``success'' event for Protocol A' is defined
by the inequality
\begin{eqnarray}
\label{passineq1}
\sum_i g_i o_i  & \leq &  (h/2 + \eta) q N. 
\end{eqnarray}
Let $S \subseteq {G}
\otimes {O}$ be the span of the vectors
$\left| \mathbf{go} \right>$ where $(\mathbf{g},
\mathbf{o} )$ varies over all pairs of sequences
satisfying (\ref{passineq1}).
For any operator $X$ on ${E}
\otimes {G} \otimes {O}$
which is classical-quantum
with respect to $(GO | E )$, let $X^s$ denote the restriction of $X$ to ${E} \otimes S$.
Applying this construction to the operators
$\Gamma_{EGO}, \overline{\Gamma}_{EGO}$
and $\Sigma$ from
Corollary~\ref{pushdowncor}, and
using the fact that $D_{max}$
and $\left\| \cdot \right\|_1$ are monotonically
decreasing under restriction to $S$, we find that
\begin{eqnarray}
\label{thedminineq}
D_{max}^\epsilon ( \Gamma_{EGO}^s
\| \Sigma^s ) \leq - N \cdot \Delta ( v, h, q, \kappa, r )
+ \frac{\log ( 2 / \epsilon^2) }{q \kappa r}.
\end{eqnarray}

In order to give a lower bound on the smooth
min-entropy of $\Gamma_{EGO}^s$,
we need to compute its divergence with respect
to an operator on ${E}
\otimes {G} \otimes {O}$
that is of the form $X
\otimes \mathbb{I}_{O}$, where $X$
is a density matrix.
Define a new operator $\Sigma'$ on ${E} \otimes
{G} \otimes {O}$ by
\begin{eqnarray}
\Sigma' & = & \Gamma_E
\otimes 
\left( \sum_{( \mathbf{g},\mathbf{o} )
\in S } (1 - q)^{\sum_i (1-g_i ) }
q^{\sum_i g_i } 2^{(h/2+\eta) N / r}
\left| \mathbf{go} \right>
\left< \mathbf{go} \right| \right)
\end{eqnarray}
(recalling that $\gamma = q \kappa r$).
Comparing this definition with (\ref{sigmaop})
and using the success criterion (\ref{passineq1}),
we find that $\Sigma' \geq \Sigma^s$.
Therefore, the bound in (\ref{thedminineq})
holds also when $\Sigma^s$
is replaced by $\Sigma'$.

When we let $\Psi$ be the operator
on ${E} \otimes {G}$ defined by
\begin{eqnarray}
\Psi & = & \Gamma_E \otimes \sum_{\mathbf{g} \in \{ 0, 1 \}^N}
(1-q)^{\sum_i (1 - g_i ) } q^{\sum_i g_i}
\left| \mathbf{g} \right> \left< \mathbf{g } \right|
\end{eqnarray}
and rewrite $\Sigma'$ as
\begin{eqnarray}
\Sigma' & = & 2^{(h/2 + \eta)N/r} ( \Psi \otimes \mathbb{I}_{O} ),
\end{eqnarray}
we find (using the rule $D_{max}^\epsilon ( X \|  Y )
= \log c + D^\epsilon_{max} ( X \| cY )$) that
\begin{eqnarray*}
D_{max}^\epsilon ( \Gamma_{EGO}^s
\| 
\Psi \otimes \mathbb{I}_{O} ) &\leq &
(h/2 + \eta) N / r
- N \cdot \Delta ( v, h, q, \kappa, r )  \\
& & + 
\frac{\log ( 2 / \epsilon^2) }{q \kappa r}.
\end{eqnarray*}
Since $\Psi$ is a density matrix, we have
\begin{eqnarray}
H_{min}^\epsilon ( \Gamma^s_{EGO} \mid EG ) 
\geq - D_{max}^\epsilon ( \Gamma_{EGO}^s
\| 
\Psi \otimes \mathbb{I}_{O} ).
\end{eqnarray}
Therefore if we let
\begin{eqnarray}
\label{defofr}
R ( v, h, \eta, q, \kappa, r ) &  = & -
\frac{h/2 + \eta}{r} + \Delta ( v, h, q, \kappa, r),
\end{eqnarray}
condition 1 of the theorem is fulfilled.
Condition 2 follows easily from the formula
for the limit of $\Delta$
(\ref{lambdalimitcond2}).
\end{proof}

A final improvement can be made on the previous
result by optimizing the coefficient $r$.

\begin{theorem}
\label{ptsecuritythm}
There exist continuous real-valued
functions $T(v, h, \eta,  q, \kappa)$
and $F ( v, h, \eta, q, \kappa )$
such that the following holds.  \begin{enumerate}
\item If Protocol A' is executed with parameters $(v, h, N, q, \eta,  D)$, then for any $\epsilon \in (0, \sqrt{2}]$
and $\kappa \in (0, \infty)$,
\begin{eqnarray}
\label{thefinalbound}
H_{min}^\epsilon ( \Gamma_{EGO}^s \mid EG )
\geq N \cdot T ( v, h , \eta, q, \kappa) 
- \left( \frac{\log ( \sqrt{2} / \epsilon)}{q \kappa} \right)
F (v, h, \eta, q, \kappa ).
\end{eqnarray}

\item The following equalities hold, where $\pi$
denotes the function from Theorem~\ref{uncertaintythm}.
\begin{eqnarray}
\label{finalthmlc1}
\lim_{(q, \kappa) \to ( 0, 0)}
T ( v, h, \eta, q, \kappa)
& = & \pi ( \eta / v ), \\
\label{finalthmlc2}
\lim_{(q, \kappa) \to (0, 0)}
F ( v, h, \eta , q , \kappa ) & = &
\frac{- 2 \pi' ( \eta / v ) }{v}.
\end{eqnarray}
\end{enumerate}
\end{theorem}

\begin{proof}
Let
\begin{eqnarray}
\mathbf{r} = \min \left\{ \frac{ v }{- \pi'(\eta / v )} , 
\frac{1}{q \kappa } \right\}.
\end{eqnarray}
Define the function $T$ by 
\begin{eqnarray}
T ( v, h, \eta,  q, \kappa) & = &
R ( v, h, \eta,  q, \kappa, \mathbf{r} ).
\end{eqnarray}
By substitution into Proposition~\ref{mainresultprop},
the bound (\ref{thefinalbound}) will
hold when we set $F$ to be
equal to  $2/( \mathbf{r} )$.

To prove (\ref{finalthmlc1}), note that 
\begin{eqnarray}
\lim_{( q, \kappa) \to (0, 0)}
T ( v, h, \eta, q, \kappa) & = & 
\lim_{(q, \kappa ) \to (0, 0)} R \left(  v, h, \eta, q, \kappa,
\frac{ v}{- \pi' ( \eta / v ) }  \right) \\
& = & \min_{s \in [0, 1 ]} \left[ \pi ( s ) - \frac{\pi' ( \eta/v )}{v}(vs - \eta) \right]\\
\label{thefinalinch}
& = & 
\min_{s \in [0, 1 ]} \left[ \pi ( s ) - \pi' ( \eta / v ) (s - \frac{\eta}{v})
\right].
\end{eqnarray}
The function enclosed by square brackets in (\ref{thefinalinch})
is a convex function of $s$ (by Theorem~\ref{uncertaintythm}) and its derivative
at $s = \eta / v$ is zero.  Therefore, a minimum is achieved
at $s = \eta / v$, and the expression in (\ref{thefinalinch})
thus evaluates simply to $\pi ( \eta / v )$.

Equality (\ref{finalthmlc2}) is immediate.  This completes
the proof.
\end{proof}

\section{Randomness Expansion from an Untrusted Device}

\label{REUDsec}

In this section, we will combine the results of previous sections
to prove that randomness expansion from an untrusted
device is possible.

\subsection{The Trust Coefficient of a Strong Self-Test}

Corollary~\ref{partialtrustsimcor} proves that if $G$
is a strong self-test, then for some $\delta_G > 0$,
the behavior of an untrusted device under $G$ can be simulated
by a partially trusted device with parameters
$(\delta_G , 2 \mathbf{f}_G )$.  Let us say that the 
\textbf{trust coefficient of $G$} is the largest value of $\delta_G$
which makes such a simulation possible.

As a consequence of the theory in section~\ref{gamessec},
we have the following formal definition for the trust coefficient of $G$.

\begin{definition}
\label{trustcoeffdef}
Suppose that $G$ is an $n$-player binary XOR game.  Then the 
\textbf{trust coefficient of $G$}, denoted $\mathbf{v}_G$,
is the maximum value of $c \geq 0$ such that there exists
a Hermitian operator $N$ on
$\left( \mathbb{C}^2 \right)^{\otimes n}$ satisfying
the following conditions.
\begin{enumerate}
\item The square of $N$ is the identity operator on $\left( \mathbb{C}^2 \right)^{\otimes n}$.

\item The operator $N$ anticommutes with the operator
$\left[ \begin{array}{cc}
0 & 1 \\ 1 & 0 \end{array} \right] \otimes \mathbb{I} \otimes \ldots \otimes \mathbb{I}$.

\item For any complex numbers $\zeta_1 , \ldots, \zeta_n 
\in \{ \zeta \mid |\zeta | = 1 , \Im ( \zeta  ) \geq 0 \}$, the operator
given by
\begin{eqnarray}
\label{Mfromzetas1}
M & = & \left[ \begin{array}{ccccccccc} &&&& a_{00 \ldots 0} \\
&&& a_{00 \ldots 1 } \\
&& \iddots \\
& a_{11 \ldots 0} \\
 a_{11 \ldots 1}
\end{array} \right],
\end{eqnarray}
where 
\begin{eqnarray}
\label{Mfromzetas2}
a_{b_1, \ldots, b_n} & = & P_G ( \zeta_1^{(-1)^{b_1}} , \zeta_2^{(-1)^{b_2}} , \ldots , \zeta_n^{(-1)^{b_n}} ),
\end{eqnarray}
satisfies
\begin{eqnarray}
\left\| M - c N \right\| \leq \mathfrak{q}_G - c.
\end{eqnarray}
\end{enumerate}
\end{definition}

\subsection{The Security of Protocol R}

\label{finalsecsubsec}

Combining Theorem~\ref{ptsecuritythm}, Corollary~\ref{partialtrustsimcor},
and the definition from the previous subsection, we have the following.
As with Protocol A', let us record the outputs of Protocol R as bit sequences
$G = (g_1, \ldots, g_N)$ and $O = (o_1, \ldots, o_N)$, where $o_i = 0$ if
the outcome of the $i$th round is $H$ or $P$, and $o_i = 1$ otherwise.
If $E$ is a purifying system for the device $D$ used in
Protocol $R$, then we denote by $\Gamma_{EGO}$ the state
of $E$, $G$, and $O$, and by $\Gamma^{s}_{EGO}$ the subnormalized
state corresponding to the ``success'' event.

\begin{theorem}
\label{maintheorem}
There exists continuous real-valued
functions $T ( v, h, \eta, q, \kappa )$
and $F ( v, h , \eta , q , \kappa )$ (with the domains
specified in Figure~\ref{parametersfig}) such that the following 
statements hold.

\begin{enumerate}

\item
Let $G$ be an $n$-player strong self-test.  Let $D$ be an
untrusted device with $n$ components, and let $E$ be a purifying
system for $D$.  Suppose that Protocol R is executed
with parameters $N, \eta, q, G, D$.  Then,
for any $\kappa \in (0, \infty)$ and $\epsilon \in (0, \sqrt{2} ]$,
the following bound holds.
\begin{eqnarray}
H_{min}^\epsilon ( \Gamma_{EGO}^s \mid EG )
\geq N \cdot T ( \mathbf{v}_G, 2 \mathbf{f}_G , \eta, q, \kappa) 
- \left( \frac{\log ( \sqrt{2} / \epsilon)}{q \kappa} \right)
F (\mathbf{v}_G, 2 \mathbf{f}_G, \eta, q, \kappa ),
\end{eqnarray}

\item The following limit conditions are satisfied, where
$\pi$ denotes the function from Theorem~\ref{uncertaintythm}.
\begin{eqnarray}
\lim_{(q, \kappa) \to ( 0, 0)}
T ( v, h, \eta, q, \kappa)
& = & \pi ( \eta / v ), \\
\lim_{(q, \kappa) \to (0, 0)}
F ( v, h, \eta , q , \kappa ) & = &
\frac{- 2 \pi' ( \eta / v ) }{v}.
\end{eqnarray}
\end{enumerate}
\end{theorem}

The following corollary shows that the linear rate of
Protocol R can be lower bounded by the function $\pi$
from Theorem~\ref{uncertaintythm}.

\begin{corollary}
\label{noisyratecorprelim}
Let $G$ be a strong self-test, and let
$\eta > 0$ and $\delta > 0$ be real numbers. 
Then, there exists positive reals $b$ and $q_0$ such that
the following holds.  If Protocol R is executed with parameters
$N, \eta, q, G, D$, where $q \leq q_0$, then
\begin{eqnarray}
H_{min}^\epsilon ( \Gamma_{EGO}^s \mid EG )
& \geq & N \cdot ( \pi ( \eta / \mathbf{v}_G ) - \delta ),
\end{eqnarray}
where $\epsilon = \sqrt{2} \cdot 2^{-bq N}$.
\end{corollary}

\begin{proof}
By the limit conditions for $T$ and $F$, we 
can find $q_0, \kappa_0 > 0$ sufficiently small and 
$M > 0$ sufficiently large so that
for any $q \in (0, q_0]$ and $\kappa \in (0, \kappa_0]$,
\begin{eqnarray}
T ( \mathbf{v}_G , 2 \mathbf{f}_G , \eta, q, \kappa )
& \geq & \pi ( \eta / \mathbf{v}_G ) - \delta/2 \\
F ( \mathbf{v}_G , 2 \mathbf{f}_G , \eta, q, \kappa )
& \leq & M.
\end{eqnarray}
Let $b = \delta \kappa_0/(2 M)$, and let
$\epsilon = \sqrt{2} \cdot 2^{-bq N}$.
Then, provided that $q \leq q_0$,
the output of Protocol R satisfies
\begin{eqnarray}
\nonumber
H_{min}^\epsilon ( \Gamma_{EGO}^s \mid EG )
& \geq &  N \cdot T ( \mathbf{v}_G, 2 \mathbf{f}_G , \eta, q, \kappa_0) 
- \left( \frac{\log ( \sqrt{2} / \epsilon)}{q \kappa_0 } \right)
F (\mathbf{v}_G, 2 \mathbf{f}_G, \eta, q, \kappa_0 ) \\
\label{steptowardsfinalbound}
& \geq &  N ( \pi ( \eta / \mathbf{v}_G ) - \delta/2 ) 
- \left( \frac{ bq N}{q \kappa_0 } \right) M \\
& = & N ( \pi ( \eta / \mathbf{v}_G ) - \delta / 2 )
- (\delta/2) N,
\end{eqnarray}
which simplifies to the desired bound.
\end{proof}

We will prove some additional corollaries in order to
achieve a security result at full strength.
First
wish to show that the output register $O$ has high min-entropy
even when conditioned on the original inputs
to the device $D$.  The above corollary takes into account 
the biased coin flips $g_1, \ldots, g_N$ used in the protocol,
but it does not take into account the inputs that are given to $D$
during game rounds.

For each $k \in \{ 1, \ldots, N \}$, let $I_k$ denote a classical register
consisting of $n$ bits which records the input used at the $k$th round.
Let $I$ be the collection of the all the registers $I_1, \ldots, I_N$.

\begin{corollary}
\label{noisyratecor}
Let $G$ be a strong self-test, and let $\eta > 0$ and $\delta > 0$
be real numbers.  Then, there exist positive reals
$b$, $K$, and $q_0$ such that the following holds.  If Protocol R
is executed with parameters $N, \eta, q, G, D$, where
$q \leq q_0$, then
\begin{eqnarray}
H_{min}^\epsilon ( \Gamma^s_{EGIO} \mid EGI )
\geq N \cdot ( \pi ( \eta / \mathbf{v}_G ) - \delta ),
\end{eqnarray}
where $\epsilon = K \cdot 2^{-bqN}$.
\end{corollary}

\begin{proof}
Let $\delta' = \delta/2$.  By Corollary~\ref{noisyratecor}, we can
find $b'$ and $q_0$ such that whenever Protocol
$R$ is executed with $q \leq q_0$,
\begin{eqnarray}
H_{min}^{\epsilon'} ( \Gamma_{EGO}^s \mid EG )
& \geq & N \cdot ( \pi ( \eta / \mathbf{v}_G ) - \delta/2 ),
\end{eqnarray}
where $\epsilon' = \sqrt{2} \cdot 2^{-b'qN}$.  
By decreasing $q_0$ if necessary, we will assume
that $q_0 < \delta/(2n)$.

For each $k \in \{ 1, 2, \ldots, \lfloor N \delta / (2n) \rfloor \}$,
let $\overline{I}_k$ denote the input string
that was given to the device $D$
on the $k$th \textit{game} round.  If there were fewer 
than $k$ game rounds, then simply let $\overline{I}_k$
be the sequence $00 \ldots 0$.  Let $\overline{I}$
denote the collection of the registers $\overline{I}_1, \ldots, \overline{I}_{\lfloor N \delta/(2n)
\rfloor }$.

Let $d$ denote the event that
\begin{eqnarray}
\sum G_i \leq N \delta / (2n).
\end{eqnarray}
(That is, $d$ denotes the event that the number
of game rounds is not more than $N \delta/2$.)  By
the Azuma-Hoeffding inequality,
\begin{eqnarray}
\label{boundford}
P ( d ) \leq e^{-N [ \delta/(2n) - q_0]^2/2}.
\end{eqnarray}
Let $\epsilon$ be the sum
of $\epsilon'$ and the quantity on the right of (\ref{boundford}), and
let $sd$ denote the intersection of the event $d$ and
the success event $s$.  Observe the following
sequence of inequalities, where we first use the fact that
the operator $\Gamma^{sd}_{EGIO}$ can be reconstructed
from the operator $\Gamma^{sd}_{EG \overline{I} O}$, and then use the fact
that the register $\overline{I}$ consists of $\leq (N \delta/2)$ bits:
\begin{eqnarray}
H_{min}^\epsilon ( \Gamma^s_{EGIO} \mid EGI ) 
& \geq & H_{min}^{\epsilon'} ( \Gamma^{sd}_{EGIO} \mid EGI )  \\
& = & H_{min}^{\epsilon'} ( \Gamma^{sd}_{EG\overline{I} O} \mid EG \overline{I} ), \\
& \geq &  H_{min}^{\epsilon'} ( \Gamma^{sd}_{EGO} \mid EG ) - N \delta / 2 \\
& \geq & H_{min}^{\epsilon'} ( \Gamma^s_{EGO} \mid EG ) - N \delta / 2 \\
& \geq & N \cdot ( \pi ( \eta / \mathbf{v}_G ) - \delta ).
\end{eqnarray}

We wish to show that $\epsilon$ is upper bounded by a decaying exponential function of
$qN$ (i.e., a function of the form $J \cdot 2^{- c q N}$,
where $J$ and $c$ are positive constants depending only on $\delta$, $\eta$,
and $G$).    We already know that $\epsilon'$ has such an upper bound.
The expression on the right side of (\ref{boundford}) also has such
a bound --- indeed, it has a bound of the form $J \cdot 2^{ - c N}$, which is
stronger.  Therefore $\epsilon$ (which is the sum of the aforementioned
quantities) is also bounded by a decaying exponential function.  This completes the proof.
\end{proof}

Finally, we wish to state a result using the language
of extractable bits from subsection~\ref{statementressubsec}.
Note that if $\rho_{XZ}$ is a subnormalized classical
quantum state of a system $(X, Z)$ that is such that
\begin{eqnarray}
H_{min}^\epsilon ( \rho_{XZ} \mid Z ) & \geq & C,
\end{eqnarray}
Then either $\Tr ( \rho ) \leq 2 \epsilon$, in which 
case $\rho$ is within trace distance $2 \epsilon$
of the zero state (which has an infinite number of
extractable bits) or $\Tr ( \rho ) > 2 \epsilon$, in which
case $\rho$ is within $\epsilon$ of a nonzero
state $\rho'$ satisfying $H_{min} ( \rho' \mid Z )
\geq C$.  In the latter case, since $\Tr ( \rho' ) \geq \Tr ( \rho ) - \epsilon > \epsilon$, we must have
\begin{eqnarray}
H_{min}^\epsilon ( \rho'/\Tr ( \rho' ) \mid Z )
& \geq & C - \log ( 1/\epsilon ).
\end{eqnarray}
Thus $\rho_{XZ}$ is within trace-distance $2\epsilon$
of a state that has $C - \log ( 1/\epsilon  )$ extractable bits.

The next collary follows easily.

\begin{corollary}
\label{rateextractablecor}
Let $G$ be a strong self-test, and let $\eta, \delta > 0$ be real numbers.  Then, there
exist positive reals $b, K$ and $q_0$ such that
the following holds.  If Protocol $R$ is executed with
parameters
$N, \eta, q, G, D$ with $q \leq q_0$, then
it produces 
\begin{eqnarray}
N \cdot ( \pi ( \eta / \mathbf{v}_G ) - \delta )
\end{eqnarray}
extractable bits with soundness error $K
\cdot 2^{-bqN}$.
\end{corollary}

\begin{remark}
\label{trustcoeffrem}
Corollary~\ref{noisyratecor} implies that
if $\pi ( \eta / \mathbf{v}_G ) > 0$, then (provided
$q$ is sufficiently small) a positive linear rate of output
entropy is achieved by Protocol R.  Using Theorem~\ref{uncertaintythm}, this means that a 
positive linear rate is achieved if $\eta < 0.11 \cdot \mathbf{v}_G$.
\end{remark}

Recall that $\pi ( 0 ) = 1$.  The next corollary follows
easily from Corollary~\ref{noisyratecor}.

\begin{corollary}
\label{highratecor}
Let $G$ be a strong self-test, and let $\delta > 0$ be a real number.  Then, there
exist positive reals $b, K, \eta$ and $q_0$ such that
the following holds.  If Protocol $R$ is executed with
parameters
$N, \eta, q, G, D$ with $q \leq q_0$, then
it produces  $N \cdot ( 1 - \delta )$
extractable bits with soundness error $K
\cdot 2^{-bqN}$.
\end{corollary}

\subsection{Example: The GHZ game}

\label{ghzsubsec}

Let $H$ denote the $3$-player binary XOR game whose
polynomial $P_H$ is given by
\begin{eqnarray}
P_H ( \zeta_1 , \zeta_2 , \zeta_3 ) &  = & 
\frac{1}{4}  \left( 1 - \zeta_1 \zeta_2
- \zeta_2 \zeta_3 - \zeta_1 \zeta_3 \right).
\end{eqnarray}
This is the Greenberger-Horne-Zeilinger (GHZ) game.

\begin{proposition}
\label{GHZprop}
The trust coefficient for the GHZ game $H$ is at least $0.14$.
\end{proposition}

For the proof of this result we will need the following
lemma (which the current authors also used 
in \cite{MillerS:self-testing:2013}):
\begin{lemma}
\label{threevarlemma}
Let $a, b, c$ be unit-length complex numbers such that
$\Im ( a ) \geq 0$ and $\Im ( b ) , \Im (c ) \leq 0$.  Then,
\begin{eqnarray}
| 1 - ab - bc - ca | \leq \frac{\sqrt{2}}{2}.
\end{eqnarray}
\end{lemma}

\begin{proof}
We have
\begin{eqnarray}
- 1 + ab + bc + ca & = & (-1 + bc ) + a (b + c ).
\end{eqnarray}
The complex number $(b+ c)$ lies at an angle of
$\pi/2$ (in the counterclockwise direction) from $(-1 + bc )$.
Since $a$ has nonnegative imaginary part, the angle formed
by $a(b+c)$ and $(-1 + bc)$ must be an obtuse or a right angle.
Therefore,
\begin{eqnarray}
\left| (-1 + bc ) + a ( b + c ) \right|^2 & \leq &
\left| -1 + bc \right|^2 + \left| a ( b + c ) \right|^2 \\
 & \leq&
4 + 4 \\
& = & 8.
\end{eqnarray}
The desired result follows.
\end{proof}

\begin{proof}[of Proposition~\ref{GHZprop}]
We proceed from Definition~\ref{trustcoeffdef}.  Let $N$
be the reverse-diagonal matrix
\begin{eqnarray}
N & = & \left[ \begin{array}{cccccccc}
&&&&&&& 1 \\
&&&&&&  1 \\
&&&&& -1 \\
&&&& -1 \\
&&& -1 \\
&& -1 \\
& 1 \\
1
\end{array} \right].
\end{eqnarray}
Clearly, $N$ anticommutes with $\sigma_x \otimes \mathbb{I}
\otimes \ldots \otimes \mathbb{I}$.

Let $\zeta_1, \zeta_2, \zeta_3$ be unit-length complex numbers
with nonnegative imaginary part, and let $M$ be the operator
given by (\ref{Mfromzetas1})--(\ref{Mfromzetas2}).  We wish
to show that the operator norm of $M - (0.14) N$ is bounded
by $\mathfrak{q}_H - 0.14 = 0.86$.  

Note that
\begin{eqnarray}
\left| \frac{1}{4}  \left( 1 - \zeta_1 \zeta_2
- \zeta_2 \zeta_3 - \zeta_1 \zeta_3 \right) - 0.14 \right| & = &
\left| 0.11 - \frac{1}{4} \left(  \zeta_1 \zeta_2
+ \zeta_2 \zeta_3 + \zeta_1 \zeta_3 \right) \right| 
\\ & \leq & 0.11 + 0.75 \\
& = & 0.86.
\end{eqnarray}
Also, by applying Lemma~\ref{threevarlemma},
\begin{eqnarray}
\left| \frac{1}{4}  \left( 1 - \zeta_1 \overline{\zeta_2}
- \overline{\zeta_2 } \overline{ \zeta_3}  - 
\zeta_1  \overline{\zeta_3 } \right) + 0.14 \right| & \leq &
\left| \frac{1}{4}  \left( 1 - \zeta_1 \overline{\zeta_2}
- \overline{\zeta_2 } \overline{ \zeta_3}  - 
\zeta_1  \overline{\zeta_3 } \right) \right| + 0.14 \\
& \leq & \frac{\sqrt{2}}{2} + 0.14 \\
& \leq & 0.86.
\end{eqnarray}
Applying similar arguments shows that every reverse-diagonal
entry of $(M - 0.14 \cdot N)$ has absolute value bounded by $0.86$.  This
completes the proof.
\end{proof}

\begin{remark}
\label{GHZrem}
By the above result and Remark~\ref{trustcoeffrem}, we have the following.
If $\eta$ is a positive real
smaller than $0.0154$ ($ = 0.11 \cdot 0.14$) and if $q > 0$
is sufficiently small, then executing Protocol R with the GHZ game
yields a positive linear rate of entropy.
\end{remark}

\subsection{Completeness}
\label{completenesssubsec}

Let $D$ be
an $n$-component binary quantum device.  For any $j \geq 1$, we will
use the expressions $I_j$ and $Y_j$ to
denote the input strings and output strings (each in $\{ 0, 1 \}^n$)
for $D$ from the $j$th iteration.

\begin{definition}
\label{noisemodeldef}
Let $G$ be a strong self-test.
For each input string
$i = (i^1, \ldots, i^n ) \in \{ 0, 1 \}^n$, the unique optimal
strategy for $G$ (see section~\ref{gamessec}) determines
a distribution on output strings $y \in \{ 0, 1 \}^n$
which we denote by $\{ p_i^y \mid
y \in \{ 0, 1 \}^n \}$.  We will say that $D$ has \textbf{noise level $\beta$}
(for the game $G$)
if, for any  $k \geq 1$, and  $i_1, \ldots, i_k, y_1, \ldots, y_{k-1} \in \{ 0, 1 \}^n$ such that
\begin{eqnarray}
\mathbf{P} ( (Y_1, \ldots, Y_{k-1} ) = (y_1, \ldots, y_{k-1}) \mid (I_1 , \ldots
, I_{k-1} ) = (i_1 , \ldots, i_{k-1}) ) & > & 0,
\end{eqnarray}
the conditional distribution
\begin{eqnarray}
\{ \mathbf{P} ( (Y_1, \ldots, Y_{k-1} ) = (y_1, \ldots, y_k)
\wedge Y_k = y \mid (I_1 , \ldots
, I_k ) = (i_1 , \ldots, i_k) )\}_y
\end{eqnarray}
is within statistical distance $(2 \beta)$ from $\{ p_{i_k}^y\}_y$.
\end{definition}
Note that an easy argument shows that a device with noise
level $\beta$ must achieve an expected score of at least $\mathbf{w}_G
- \beta$.

We now discuss completeness.
We will make use of a refined Azuma-Hoeffding inequality~\cite{AH}.
\begin{lemma}\label{lm:AZ} Suppose that $S_1, S_2, ..., S_N$ is a Martingale
with 
\[ |S_{i+1}-S_i|\le 1,\]
and 
\[\mathrm{Var}\left[S_{i+1}-S_i \ |\ S_1,..., S_i\right]\le w,\]
for all $i$, $1\le i\le N-1$. Then for any $\epsilon\in(0,1)$,
\begin{equation}\label{eqn:AZ}
\Prob\left[ S_N \ge \epsilon{w} N \right] \le \exp\left(-\epsilon^2\frac{{w}}{2}N\left(1-\frac{1-{w}}{3}\epsilon\right)\right).
\end{equation}
In particular if $\epsilon\le1$, we have
\begin{equation}\label{eqn:AZ_use}
\Prob\left[ S_N \ge \epsilon{w} N \right] \le \exp\left(-\epsilon^2\frac{{w}}{3}N\right).
\end{equation}

\end{lemma}

\begin{proposition}
\label{completenessprop}
Suppose that the device in Protocol $R$ has noise level
$\eta' < \eta$.  Then the probability of aborting is at most
$\exp ( - (\eta - \eta')^2 q N / 3 )$.
\end{proposition}

\begin{proof}
Let $I_1, \ldots, I_N$ and $Y_1, \ldots, Y_N$ be random
variables containing the inputs and outputs for Protocol $R$.
Let $Z_i$ be equal to $1$ if the game is won on the
$i$th round and $0$ otherwise.
Let
\begin{eqnarray}
z_i = \mathbf{E} [ Z_i \mid I_1, \ldots, I_{i-1}, Y_1, \ldots, Y_{i-1}].
\end{eqnarray}
By definition, Protocol $R$~(Fig.~\ref{protocolrtable}) aborts when
\begin{equation}\label{eqn:abort_event}
\sum_i g_i(1-Z_i) \ge (1-\mathbf{w}_G+\eta) qN.
\end{equation}
By assumption,
\begin{equation}
\sum_i  \left(  \mathbf{w}_G - z_i \right) \le \eta' N.
 \end{equation}
Let 
\begin{eqnarray}
R_i & = & \sum_{k=1}^i g_k(1-Z_k)  -q \sum_{k=1}^i (1-z_k).
\end{eqnarray}
Then $R_1, R_2, \ldots$ is a Martingale with
\begin{equation}
\mathrm{Var}\left[R_i-R_{i-1} \ |\ R_1, ..., R_{i-1}\right] = q (1-z_i) [1-q(1-z_i)]  \le q,
\end{equation}
thus (\ref{eqn:abort_event}) implies that
\begin{equation}
\sum_i g_i(1-Z_i)  -q \sum_i (1-z_i) \ge \eta qN - q\sum_i (\mathbf{w}_G-z_i) \ge (\eta-\eta') qN.
\end{equation}
Thus by Corollary~\ref{co:w}, the probability of aborting is $\le \exp(-(\eta-\eta')^2qN/3)$.
\end{proof}

\def\tr{{\rm tr}}
\section{Unbounded Expansion}\label{sec:unbounded}
In this section, we prove a general result, the Composition Lemma (Lemma~\ref{lm:composition} below),
that implies Corollary~\ref{co:unbounded} straightforwardly.
This general result implies that known untrusted-device randomness
expansion protocols,
including ours, can be composed sequentially with additive errors, 
even if only two  devices are used. The proof will be short, but it is important to define the 
error parameters appropriately. 
We us a high-level framework for rigorously reasoning about these protocols,
following~\cite{CSW14}.
We only sketch the necessary elements and refer interested readers to~\cite{CSW14}
for a more comprehensive description.

Since we are cross-feeding inputs and outputs between
devices, we will
use the same syntax for both input and output states.
A {\bf protocol state space} $\H$ is a three-part Hilbert space
\begin{equation}\label{eqn:protocol_space}
\H= C \otimes  D\otimes E,
\end{equation}
where $C$, $D$, $E$ are referred to as the {\bf classical}, {\bf device},
and {\bf adversary} subsystems, respectively.\footnote{If
$D$ is a quantum device, then by a small
abuse of notation, let us also use the letter $D$ to denote the quantum
system inside $D$.}
We also represent a protocol space by
the triple $(C, D, E)$.  

We call a {\em subnormalized}  classical-quantum-quantum state $\rho$ over $\H$
a {\bf protocol state}.
Those states are the accepting (or non-aborting) portion of the normalized states
in our protocols.  Denote by $\hat\rho=\rho/\Tr(\rho)$ the corresponding
normalized state. Correspondingly, we allow quantum operations to be trace-non-increasing
with the understanding that the missing trace (from the unit) corresponds to rejecting (or aborting).

We call a protocol state $\rho_{CDE}$ {\bf device-uniform}, {\bf adversary-uniform}, or {\bf global-uniform}
if in $\hat\rho$, $C$ is uniform with respect to $D$, or $E$, or $DE$, respectively. 
For any $\epsilon\in[0,2]$, $\rho$ is said to be {\bf $\epsilon$-device-uniform}
if there exists a subnormalized device-uniform state $\tilde\rho_D$ within $\epsilon$ trace-distance to
$\rho$.
Similarly define $\epsilon$-adversary-uniform and $\epsilon$-global-uniform.

A {\bf strong untrusted-device (UD) extractor} $\Pi$ 
is a 
procedure which takes as input a classical register $X$, a device $D$,
and a quantum system $E$, then performs $X$-controlled operations
on $D$ and $E$, produces a classical output register $Y$,
and then ``aborts'' or ``succeeds.''   For our discussion in this paper,
we allow the steps in the procedure to include both
classical interaction with the device, and arbitrary
device-adversary operations (i.e., quantum operations
on the composite system $DE$).
The protocol $\Pi$
maps protocol states over the input protocol space $(X, D, E)$
 to those over the output protocol space
$(Y, D, XE)$ (the success states).  An {\bf implementation}
of the extractor $\Pi$ is a specification of the initial state
of $(X, D, E)$, and the measurements performed by the device,
and the operations used in any device-adversary interactions.

If $\Pi$ is a strong UD extractor which involves no device-adversary
interactions, then an {\bf ideal implementation} for $\Pi$ is one in
which the device is
such that it has the same conditional input-output distribution on each use.
If an ideal implementation has been specified, we then use the term \textbf{noise level} in
the same sense as in Definition~\ref{noisemodeldef}: an implementation
for $\Pi$ has noise level $\eta$ if
at each use, the output distribution of the device $D$ on any transcript
and any input
is within statistical distance $\eta$ from that of the ideal implementation.
The extractor $\Pi$ has {\bf completeness error $\epsilon_c$} tolerating
noise level $\eta$ if for any 
implementation having noise level $\eta$, the success probability
of $\Pi$ is at least $1 - \epsilon_c$.  

Let $\Pi$ is a strong UD extractor which has no
device-adversary interactions.  We call $\epsilon_s$ a {\bf soundness error} of $\Pi$ on a set $\mathcal{S}$ of protocol states
if for any $\rho\in \mathcal{S}$ and any compatible implementation, $\Pi(\rho)$ is $\epsilon_s$-adversary uniform.  
We say that on $\mathcal{S}$, $\Pi$ has a {\bf adjustment completeness error } $\hat\epsilon_c$ tolerating a noise level $\eta$,
if there exists an ideal implementation, 
such that for all {\em normalized} states $\rho\in\mathcal{S}$ and all implementations with noise level $\eta$, 
the protocol's final state is within $\hat\epsilon_c$ trace distance to a {\em normalized} adversary-uniform state.
This change of completeness error is not substantial: $\hat\epsilon_c\le 2(\epsilon_c+\epsilon_s)$,
and $\epsilon_s\le \hat\epsilon_c$.

The Equivalence Lemma of Chung, Shi and Wu~\cite{CSW14} states the following.
\begin{theorem}[The Equivalence Lemma~\cite{CSW14}]\label{lm:EL}
Let $\Pi$ be a strong UD extractor that has no device-adversary interactions. Then on
the set of global-uniform inputs and the set of device-uniform inputs,
$\Pi$ has the same soundness error, adjustment completeness error, and noise tolerating level.
\end{theorem}

We now formally define the composition of protocols using two devices.
(The earliest mention of this approach that we know of was in \cite{pam:2010}.) \begin{definition} Let $D_0$ and $D_1$ be two untrusted quantum devices and $T\ge1$ be an integer. 
A {\bf cross-feeding protocol} $\Sigma$ using $D_0$ and $D_1$ consists of 
a sequence of strong UD extractors $\Sigma_i$, $i=0,1,..., T$,
such that $\Sigma_i$ uses $D_{i\mod 2}$ and
the output of $\Sigma_i$ is used as the input to $\Sigma_{i+1}$.
The input to $\Sigma$ is the input to $\Sigma_0$, and the output is that of $\Sigma_{T}$.

Furthermore, $\Pi$ is said to use {\bf restorable} devices if, for each $i$, $1\le i\le T$,
between $\Sigma_{i-1}$ and $\Sigma_{i+1}$, there is a device-adversary variable operation
$A_i$ on $X_iD_{i+1\mod 2}E$,
controlled by $X_i$.
\end{definition}

\begin{figure*}
\begin{center}
\begin{tikzpicture}[shorten >=1pt,node distance=2cm,thick,scale=1, every node/.style={scale=1}]
\node [draw,circle]  (global_state) {};
\node [draw,rectangle] (Sigma0) [above right of = global_state] {$\Sigma_0$};
\node [draw,rectangle] (A0) [ below right of = global_state] {$A_0$};
\node [draw,rectangle] (Sigma1) [right  of = Sigma0] {$\Sigma_1$};
\node [draw,rectangle] (A1) [right  of = A0] {$A_1$};
\node [draw,rectangle] (Sigma2) [right  of = Sigma1] {$\Sigma_2$};
\node [draw,rectangle] (A2) [right  of = A1] {$A_2$};
\node [] (Sigma3) [right  of = Sigma2] {$\cdots$};
\node [] (A3) [right  of = A2] {$\cdots$};
\node [draw,rectangle] (Sigma4) [right  of = Sigma3] {$\Sigma_{T-1}$};
\node [draw,rectangle] (A4) [right  of = A3] {$A_{T-1}$};
\node [draw,rectangle] (SigmaT) [right  of = Sigma4] {$\Sigma_T$};
\node [draw,rectangle] (AT) [right  of = A4] {$A_T$};
\node  (SigmaE) [right  of = SigmaT] {};
\node (AE) [right  of = AT] {};

\path[->]
(global_state) edge[out=90, in=180] node [left] {$XD_0$}(Sigma0)
(Sigma0) edge node [above] {$X_1$} (Sigma1)
(Sigma1) edge node [above] {$X_2$} (Sigma2)
(Sigma2) edge node [above] {$X_3$} (Sigma3)
(Sigma3) edge node [above] {$X_{T-1}$} (Sigma4)
(Sigma4) edge node [above] {$X_T$} (SigmaT)
(SigmaT) edge node [above] {$Y$} (SigmaE)
(global_state) edge[out=270, in=180] node [left] {$XD_1E$}(A0)
(A0) edge node [below] {$E^1$} (A1)
(A1) edge node [below] {$E^2$} (A2)
(A2) edge node [below] {$E^3$} (A3)
(A3) edge node [below] {$E^{T-1}$} (A4)
(A4) edge node [below] {$E^{T}$} (AT)
(AT) edge node [below] {$E^{T+1}$} (AE)
(Sigma0) edge node [right, pos=.10] {$X_1D_0^1$} (A1)
(A0) edge node [right, pos=.10] {$D_1^1$} (Sigma1)
(Sigma1) edge node [right, pos=.10] {$X_2D_1^2$} (A2)
(A1) edge node [right, pos=.10] {$D_0^2$} (Sigma2)
(Sigma2) edge node [right, pos=.10] {$X_3D_2^3$} (A3)
(A2) edge node [right, pos=.10] {$D_1^3$} (Sigma3)
(Sigma3) edge node [right, pos=.10] {} (A4)
(A3) edge node [right, pos=.10] {} (Sigma4)
(Sigma4) edge node [right, pos=.10] {$X_TD_{1-b}^T$} (AT)
(A4) edge node [right, pos=.10] {$D_{b}^T$} (SigmaT)
(AT) edge [out=90, in=180] node [above, pos=.75] {$D_{1-b}^{T+1}$} ([yshift=.75cm]AE.north)
(SigmaT) edge [out=270, in=180] node [above, pos=.75] {$D_{b}^{T+1}$} ([yshift=-.75cm]SigmaE.south)

;
\end{tikzpicture}
\end{center}
\caption{The tensor network representation of a cross-feeding protocol.
Superscripts indicate that a (device or adversary) subsystem may change
after an operation. Each $E^t$ includes a copy of $X_0,...,X_{t-1}$ 
to be consistent with the definition of strong UD extractors.
The subscript $b=T\mod 2$.}\label{fig:cross_feeding}
\end{figure*}
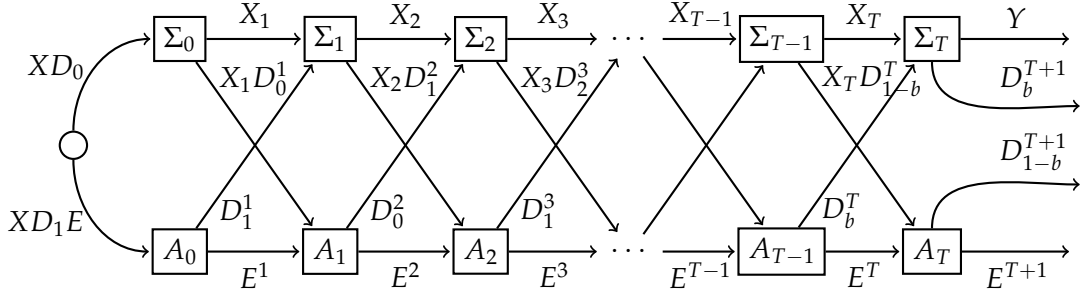

A cross-feeding protocol is illustrated in Fig.~(\ref{fig:cross_feeding}).
Allowing the device-adversary $A_i$ operations reduces the technical challenges for practical implementations,
since in an honest implementation, $A_i$ can be used to replenish the consumed entanglement.
While our wording of ``two'' devices was inspired by a possible implementation of using two physical devices,
the inactive device can certainly be replaced by a different physical device, since such replacement is one possible $A_i$ operation.
The essence of our two device protocol lies in how communication is restricted: besides forbidding communication
between an active device and its external world --- which is already required for a single-device protocol --- the additional
constraint  is that after an active device finishes its work, no information is allowed to travel from the device
to the inactive device before the latter becomes active. A single device would not satisfy this latter requirement.

\begin{lemma}[Composition Lemma] \label{lm:composition}Let $\Sigma$ be a cross-feeding protocol defined above.
Assume that for each $i$, $0\le i\le T$,
$\Sigma_i$ has a soundness error $\epsilon_{s,i}$, 
and an adjustment completeness error $\hat\epsilon_{c,i}$ tolerating a noise level $\eta$ (set to be the same for all $\Sigma_i$),
with respect to device-uniform inputs.
Let $\epsilon_s:=\sum_{i=0}^T\epsilon_{s,i}$ and $\hat\epsilon_c:=\sum_{i=0}^{T}\hat \epsilon_{c,i}$.
Then $\Sigma$ on states uniform to $D_0$
has a soundness error $\epsilon_s$,
and an adjustment completeness error $\hat\epsilon_c$ tolerating an $\eta$ level of noise.

In particular, the above statement holds if each $\Sigma_i$ has no device-adversary interaction
and the parameters $\epsilon_s, \hat\epsilon_c, \eta$ are valid on global-uniform inputs.
\end{lemma}

The soundness proof uses the following two facts, both of which follow directly from the corresponding definitions.
\begin{fact}\label{fact:switch}Let  $b\in\{0,1\}$ and $\rho=\rho_{YD_b (D_{1-b}E)}$ be adversary-uniform.
Then $\rho$ as $\rho_{YD_{1-b}(D_bE)}$ is device-uniform for $D_{1-b}$.
This remains true for $(I\otimes A_{D_bE})\rho$ for any operation $A_{D_bE}$ on $D_bE$.
\end{fact}

\begin{fact}\label{fact:add} If $\Pi$ is a strong UD extractor with an
$\epsilon$ soundness error, and  $\rho$ is $\delta$-device-uniform,
$\Pi(\rho)$ is $(\epsilon+\delta)$-adversary uniform.
\end{fact}

\begin{proof}
The proof for the device-uniform case follows from 
a straightforward  inductive proof on the following two statements.
Denote by $\epsilon_s^i:=\sum_{j=0}^{i}\epsilon_{s,j}$, and
$\hat\epsilon_c^i:=\sum_{j=0}^{i}\hat\epsilon_{c,j}$.

\begin{itemize}
\item(Soundness) On any implementation and any initial
input uniform to $D_0$, for each $i$, $0\le i\le T$, the output of
$\Sigma_i$ is $\epsilon_s^i$-adversary-uniform.
\item(Completeness) Fix an ideal implementation for each $\Sigma_i$ to achieve the adjustment completeness error.
For any $\eta$-deviated implementation, any normalized initial input uniform to $D_0$,
any $i$, $0\le i\le T$, the output of $\Sigma_i$ is $\hat\epsilon_c^i$-close
to a normalized adversary-uniform state.
\end{itemize}

More specifically, for the soundness argument, the base case holds by applying Fact~\ref{fact:add}
to $\Sigma_0$ and the assumption that the input is uniform to $D_0$. For the inductive step,
by the inductive hypothesis  (assuming it holds for $i$, $0\le i\le T-1$) and Fact~\ref{fact:switch},
the input to $\Sigma_{i+1}$ is $\epsilon_s^i$-device uniform. By Fact~\ref{fact:add}, the output 
is thus $\epsilon_s^{i+1}$-adversary uniform. The proof for the completeness follows from the definition of
completeness and triangle inequality.

The global-uniform case follows by applying the Equivalence Lemma~\ref{lm:EL}.
\end{proof}

Corollary~\ref{co:unbounded} follows by using our robust protocol in Corollary~\ref{co:oneshotRand}.

\section{Untrusted-device Quantum Key Distribution}\label{sec:qkd}
In this section, we shall first formally define what we mean by a key distribution protocol using
untrusted quantum devices. We then present Protocol $R_{\textrm{kd}}$, a natural adaptation of Protocol $R$ for 
untrusted-device quantum key distribution,
then we prove its correctness (Corollary~\ref{co:qkd}.) 

\subsection{Definitions}\label{subsec:qkdDefs}
  A {\bf min-entropy untrusted-device key distribution (ME-UD-KD) protocol} $\Pi_\textrm{kd}$ 
 is a communication protocol in the following form 
between two parties Alice and Bob who have access to distinct components
of an untrusted quantum device.
Before the protocol starts, they share a string that is uniformly random to the device.
They communicate through a public, but authenticated, channel. At each step,
both the message they send and the new input to their device components are
a deterministic function of the initial randomness, the messages received, and the previous
output of their device component. The protocol terminates with a public bit $S$, indicating
if the protocol succeeds or aborts, and Alice and Bob each have a private string: $A$ and $\tilde A$, respectively.

The protocol  is said to have a yield $M$ with a {\bf soundness error } $\epsilon_s$ 
if both the following conditions hold.
\begin{enumerate}
\item[(a)]\label{cond:qkd:min_ent}  the joint state $(S, A, E)$ is $\epsilon_s$-close to a mixture of
an aborting state and one where $A$ has $M$ extractable bits, and
\item[(b)] the joint state $(S, A, \tilde A)$ is $\epsilon_s$-close to a mixture of an aborting state and one where $A=\tilde A$.
\end{enumerate}
\YtalkingPoint{Above: changed the notation to be consistent with later writing; also clarified (b) to allow abort.}

The protocol is said to have a completeness error $\epsilon_c$ with respect to
a non-empty class of untrusted devices $\mathcal{U}_\textrm{honest}$, if for any device in this class,
the protocol aborts with probability $\le \epsilon_c$.

  If in the above definition, Condition (a) has ``$M$ extractable bits"
replaced by ``$M$ uniformly random bits'', then we call
the protocol simply an {\bf untrusted-device key distribution protocol}
with those parameters.

\subsection{The Protocol $R_{\textrm{kd}}$}
Protocol $R_{\textrm{kd}}$ is an adaptation of Protocol $R$ to the distributed setting and is described in Fig.~\ref{protocolrqkd}.
\begin{figure}
\begin{center}
\setlength{\fboxsep}{10pt}
\setlength{\fboxrule}{1pt}
\fbox{\parbox{\linewidth}{

\begin{wrapfigure}{r}{.5\textwidth}
\scalebox{0.8}{
\begin{tikzpicture}[shorten >=1pt,node distance=1.5cm,thick,scale=1, every node/.style={scale=1}]

\node [rectangle, minimum width=.5cm, minimum height = 1.5cm, draw,fill=gray] (brick) {};
\node [rectangle, minimum size=1.5cm,draw] (D1) [left of = brick] {$D_1$};
\node [rectangle, minimum size=1.5cm,draw] (D2) [right of = brick] {$D_2$};
\node [rectangle, minimum width=.5cm, minimum height = 1.5cm, draw,fill=gray] (brick2) [right of = D2] {};
\node [rectangle, minimum size=1.5cm,draw] (Dn) [right of = brick2] {$D_n$};
\node (Alice) [above of = D1,yshift=.5cm] {Alice};
\node (Bob) [above of = brick2,yshift=.5cm] {Bob};

\draw [rounded corners] ([xshift=-.25cm,yshift=.25cm] D1.north west) rectangle ([xshift=.25cm,yshift=-.25cm]D1.south east);
\draw [rounded corners] ([xshift=-.25cm,yshift=.25cm] D2.north west) rectangle ([xshift=.25cm,yshift=-.25cm]Dn.south east);

\draw [<->, thick] (Alice) -- (Bob);
\draw [<->, thick] (Alice) -- ([yshift=.25cm] D1.north);
\draw [<->, thick] (Bob) -- ([yshift=.25cm] brick2.north);


\end{tikzpicture}
}
\captionof*{figure}{A diagram of Protocol $R_{\textrm{kd}}$.}
\label{fig:Rqkd}
\end{wrapfigure}

\textit{Arguments:}
\begin{enumerate}[topsep=3pt]
\item[$G:$] An $n$-player nonlocal game that is a {\bf strong self-test} (Definition~\ref{def:strong}). Assume without loss of generality that $0^n$ is an input on which the winning probability is no less than the average in the optimal quantum strategy.
\item[$D:$] An untrusted device (with $n$ components) that can play $G$ repeatedly and cannot receive any additional information. 
Alice interacts with the first component while Bob interacts the rest of the device. No communication is allowed
among the components during Step 1-4 of the protocol. All random bits chosen by Alice and Bob together are assumed
to be perfectly random to $D$.
\item[$N:$] a positive integer (the {\bf output length}.)
\item[$\lambda:$]   A real $\in (0, {\mathbf w}_G-1/2)$. ($1/2-\lambda$ is the {\bf key error fraction}.)
\item[$\eta:$] A real $\in(0,\frac{1}{2})$. (The {\bf error tolerance}.)
\item[$q:$] A real $\in(0,1)$.  (The {\bf test probability}.)

\end{enumerate} 

\medskip

\textit{Protocol:}

\begin{enumerate}[topsep=3pt]
\item Repeat the following procedure for $N$ times. Alice and Bob will each produce a raw key,
stored as an $N$-bit binary string $A$ and $B$, respectively.
\begin{enumerate}
\item Alice and Bob choose a bit $g \in \{ 0, 1 \}$ 
according to a biased $(1 - q, q )$ distribution.


\item If $g = 1$ (``game round''), then Alice and Bob choose an input string at random from $\{ 0, 1 \}^n$ according
to the probability distribution specified by $G$. 
They give their part(s) of $D$ the corresponding input bit, exchange their output bits
and record a ``P'' (pass)  or an ``F'' (fail) 
according to the rules of the game $G$, and store this bit (``P'' as $1$ and ``F'' as $0$)
as their raw key bit for this round. 
\YtalkingPoint{Sorry that I revert this as it'd be more convenient for proving Cor. 1.9; otherwise I'll have to apply IR on a variable length.}

\item If $g = 0$ (``generation round''), then the input string
$00 \ldots 0$ is given to the device.
Alice sets the raw key bit in $A$ for this round to be her output bit.
Bob sets his raw key bit in $B$ to be the unique bit that when XOR'ed with the output bit(s) of his device component(s)
would constitute a win for the game. That is, their bits are the same if and only if they win the game.

\end{enumerate}


\item  If the total number of failures is more than $(1-{\mathbf w}_G+\eta) q N$,
the protocol \textbf{aborts}.  

\item If not yet aborted, they run an Efficient Information Reconciliation (such as Protocol~EIR in Fig.~\ref{fig:eir})
on $A$ and $B$, the parameters $\lambda$ 
and $\epsilon=\exp(-qN)$.  Alice's final output is $A$ (unchanged),
and Bob's final output $\tilde A$ is his output from the information reconciliation protocol.
\talkingPoint{I'm not
sure exactly what this means -- could it be made clearer? -Carl}
\YtalkingPoint{Change how the raw keys are stored for game rounds to be the same as the test rounds, and simplify the language
for the protocol. Also now EIR doesn't abort.}
\end{enumerate}
}}
\caption{Protocol $R_{\textrm{kd}}$}
\label{protocolrqkd}
\end{center}
\end{figure}
There are two main steps in the proof for Corollary~\ref{co:qkd}. The first is to show that 
for an appropriate range of the parameters, 
Protocol $R$ has a soundness and completeness error of $\exp(-\Omega(qN))$ with the ideal
state being that $A$ and $B$ differ in at most a $(1/2-\lambda)$ fraction, for a constant $\lambda$.
The second step is to construct the Efficient
Information Reconciliation Protocol that works on the ideal state
and for Bob to correct the differences with some small failure probability. 
We present those two steps in two separate subsections, which are followed by the proof for the Corollary.

\subsection{Error Rate}
The completeness error is straightforward, so our focus will be on the soundness error.

Our result applies to a broader class of games than the strong self-tests.
\begin{definition} \label{def:fselftest} Let $f:[0,1]\to [0,1]$ be a strictly increasing concave function
with $f(0)=0$.  
A game $G$ is said to be {\bf $f$-self-testing in probability} if there exists an input $x_0$ such that the following holds:
If for any $\theta\in(0,1)$ and any quantum strategy that wins with probability $(1-\theta) {\mathbf w}_G$, the game wins
on $x_0$ with probability $\ge (1-f(\theta)){\mathbf w}_G$.
\end{definition}

\commentout{  The following theorem uses the notion of \textit{second-order
robust self-test} (which in this paper we call \textit{strong self-test})
from the paper \cite{MillerS:self-testing:2013}.}
\begin{theorem}\label{thm:MS_selftest}
Let $G$ be a strong self-test.
Then there exists a constant $C>0$ such that
$G$ is $C \sqrt{\theta}$-self-testing in probability.
\end{theorem}

\begin{proof}
  Since $G$ is strongly self-testing, there is a unique quantum
strategy which achieves the optimal winning probability $\mathbf{w}_G$
(see section~\ref{gamessec}).  Let $x_0$ be an input string 
(which
occurs with nonzero probability in $G$)
such that, if the optimal strategy is applied
on input $x_0$, the winning 
probability is at least $\mathbf{w}_G$.

  If a given quantum strategy for $G$
achieves a score of $(1 - \theta ) \mathbf{w}_G$,
then by the strong self-testing property its output
distribution on input $x_0$ is $C_1  \sqrt{\theta} $-close
to that of the optimal strategy, for some constant
$C_1$.  The result follows.
\end{proof}

  Consequently all strong self-tests are $O(\sqrt{\theta})$-self-testing in probability.

We now fix a game $G$ that is $f$-self-testing in probability for some function $f$ on input $00\dots0$. \YtalkingPoint{need clarify that 00...0 is w.l.o.g.} 
Let $w_i$ , $1\le i\le N$, be the random variable denoting the chance of winning the $i$th round game under the full input 
distribution, right after the $(i-1)$th round is played. Similarly define $w_i^0$ by replacing the full input distribution with
the input $00\ldots0$. These random variables may be correlated as the behavior of the $i$th game may depend on
the history of the previous $i-1$ games. Let $W_i$ be a random variable
which is equal to $1$ if the game is won on the $i$th round,
and $0$ otherwise. Note that the expected value of $W_i$
is equal to $w_i$
if the $i$th round is a game round, and is equal to $w_i^0$ if the
$i$th round is a generation round. Another useful fact about $W_i$'s is that when $g_i=0$,
$W_i=1$ if and only if the $i$'th bits of $A$ and $B$ are equal. This follows from the construction of $B$.
Consequently,
\begin{equation}\label{eqn:ABW}
\sum_i W_i = N- |A+B|.
\end{equation} 

Intuitively,  if the devices are doing well on the game rounds, they should do well on the randomness generating rounds as well
because of self-testing. The following theorem is one way to express this intuition.
\begin{lemma}\label{lm:qkderr} 
Let $G$ be a strong self-test game that is $f$-self-testing in probability for some $f$ on input $00\ldots0$.
Consider Protocol R~(Fig.~\ref{protocolrtable}) using $G$ and an arbitrary $q\in(0,1)$.
For any $\lambda\in(0, {\mathbf w}_G-1/2)$ and all sufficiently small constant $\eta>0$,
there exist constants $\alpha,\beta>0$ such that for the events
\begin{align}
&P:=\sum_i g_i(1-W_i) \le (1-{\mathbf w}_G+\eta) qN,\label{eqn:EventPassing}\\
&M:=\sum_i (1-g_i)W_i \le (1/2+\lambda)(1-q)N, \label{eqn:EventBad}\quad\textrm{and,}\\
&E:=P\wedge M\label{eqn:EventE},
\end{align}
we have
\begin{align}
\Prob [E]
\le \exp\left(- \alpha qN \right) + \exp\left(-\beta N\right).\label{eqn:qkdeventBound}
\end{align}
\end{lemma}

To prove the above lemma, we first derive two concentration results. Consider
\begin{equation}
T_i : = \sum_{j=1}^i \left( g_i (1-W_i) - q (1-w_i) \right).
\end{equation}
Since $E[g_i(1-W_i) - q (1-w_i) \ |\ T_1, ..., T_{i-1}] = 0$,
and 
\begin{equation}
\mathrm{Var}\left[T_i-T_{i-1} \ |\ T_1, ..., T_{i-1}\right] = q (1-w_i) [1-q(1-w_i)]  \le q,
\end{equation}
applying Lemma~\ref{lm:AZ}, we have
\begin{corollary}\label{co:w} 
For any $\epsilon\in(0,1)$,
\begin{equation}
\Prob\left[ \sum_i g_i (1-W_i)  - q\sum_i (1-w_i) \le -\epsilon q N\right] \le 
\exp\left(-\epsilon^2\frac{q}{3}N\right).
\end{equation}
\end{corollary}

Consider now
\begin{equation}
S_i : = \sum_{j=1}^i\left( (1-g_i)W_i -  (1-q) w^0_i \right).
\end{equation}
Then $S_i$ is a Martingale and
\begin{equation}
\mathrm{Var}\left[S_i-S_{i-1} \ |\ S_1, ..., S_{i-1}\right] = (1-q) w_i^0 [1-(1-q)w_i^0]  \le 1-q.
\end{equation}

Thus the following Corollary follows from the standard Azuma-Hoeffding bound.
\begin{corollary}  \label{co:w0}
For any $\epsilon>0$,
\begin{equation}
\Prob\left[ \sum_i \left((1-g_0)W_i  -  (1-q) w^0_i  \right )\le - \epsilon (1-q)N\right] \le 
\exp\left(-\frac{\epsilon^2}{3}(1-q)N\right).
\end{equation}
\end{corollary}

\begin{proof}[of Lemma~\ref{lm:qkderr}]
Fix an arbitrary $\lambda\in(0, {\mathbf w}_G-1/2)$. Let $\eta_0=\eta_0(\lambda), \epsilon_1,\epsilon_2\in(0,1)$
be determined later. Fix an arbitrary $\eta\in(0, \eta_0)$. 
Define the following two events
\begin{eqnarray}
E_1&:=& \sum_i g_i(1-W_i) > q\sum_i (1-w_i) - \epsilon_1 qN,\label{eqn:event1}\\
E_2&:=& \sum_i (1-g_i)W_i  > (1-q)\sum_i w_i^0 -\epsilon_2(1-q) N.\label{eqn:event2}
\end{eqnarray}
Apply Corollaries~(\ref{co:w}) and (\ref{co:w0}) with $\epsilon=\epsilon_1$ and $\epsilon=\epsilon_2$, respectively, we have
\begin{eqnarray}
\Prob[\bar E_1] &\le& \exp\left(- \frac{\epsilon_1^2}{3}qN \right), \label{eqn:qkdevent1Bound}\\
\Prob[\bar E_2] & \le& \exp\left(- \frac{\epsilon_2^2}{3}(1-q)N\right).\label{eqn:qkdevent2Bound}
\end{eqnarray}
Then
\begin{align}
 \Prob[E] & \le \Prob[\bar E_1] + \Prob[\bar E_2] + \Prob[E\wedge E_1\wedge E_2]\\
 & \le \exp\left(- \frac{\epsilon_1^2}{3}qN \right) + \exp\left(- \frac{\epsilon_2^2}{3}(1-q)N\right) + \Prob[E\wedge E_1\wedge E_2],\label{eqn:addEvents}
\end{align}
where the bounds from (\ref{eqn:qkdevent1Bound}, \ref{eqn:qkdevent2Bound}) are used.

To bound $ \Prob[E\wedge E_1\wedge E_2]$, denote by 
\begin{equation}
 \hat w := \frac{1}{N}\sum_iw_i/{\mathbf w}_G,\quad\textrm{and,}\quad \hat w^0:=\frac{1}{N}\sum_iw^0_i/{\mathbf w}_G.
 \end{equation}
 Event $P$ and $E_1$ imply
 \begin{equation}\label{eqn:wUpperbound}
 1-\hat w < (\eta+\epsilon_1)/{\mathbf w}_G = \eta_0/{\mathbf w}_G.
 \end{equation}
 By the assumption that $G$ is $f$-self-testing in probability and the concavity of $f$,
 the above implies
 \begin{equation}\label{eqn:w0UB}
 1-\hat w^0 < f\left(\eta_0/{\mathbf w}_G\right).
 \end{equation}
 
 Meanwhile, Event $M$ and $E_2$ imply
 \begin{equation} \label{eqn:averagedUpperbound}
 1-\hat w^0 > 1- \frac{1/2+\lambda+\epsilon_2}{\wg}.
 \end{equation}
 The last two inequalities imply
 \begin{align}\label{eqn:false}
f\left((\eta+\epsilon_1)/\wg\right) > 1- (1/2+\lambda+\epsilon_2)/\wg.
\end{align}
Since $f(\theta)\to0$ when $\theta\to0^+$, the LHS of the above inequality $\to0$ when $\eta+\epsilon_1\to0$. Note that for any fixed $\lambda<\wg-1/2$,  RHS $>0$ when $\epsilon_2\to0^+$. Following this intuition, we define
\begin{align}
\eta_0:=\max \left\{ t\in [0, \wg]: f(t/\wg) \le 1- (1/2+\lambda)/\wg\right\}.
\end{align}

Now if one sets $\epsilon_1=\eta_0-\eta-\epsilon'$ and let  $\epsilon', \epsilon_2\to0^+$, 
\begin{align*}
\lim_{\epsilon'\to0^+} f\left((\eta+\epsilon_1)/\wg\right) = f(\eta_0/\wg) \le 1-(1/2+\lambda)/\wg = \lim_{\epsilon_2\to0^+}1- (1/2+\lambda+\epsilon_2)/\wg.
\end{align*}
Thus for some sufficiently small $\epsilon'_0>0$ and $\epsilon_2>0$, Eqn.~\ref{eqn:false} becomes false,
which means that Event $E\wedge E_1\wedge E_2$ does not occur. Setting 
\begin{align}
\alpha:=(\eta_0-\eta-\epsilon'_0)^2/3\quad\textrm{and}\quad
\beta:=\epsilon_2,
\end{align}
and by Eqn.~(\ref{eqn:addEvents}),
\begin{equation}
\Prob[E] \le \exp\left(- \alpha qN \right)+ \exp\left(-\beta N\right).
\end{equation}
Thus the theorem holds.
\end{proof}

\subsection{Efficient Information Reconciliation} 
We now arrive at the problem of resolving differences between Alice and Bob's keys.
This problem, called {\em information reconciliation} (IR), has been studied since the early days
of quantum cryptography (the earliest works include~\cite{Robert:thesis,bennett1988privacy}).
There are several variations of the problem, for examples, depending on how the differences are quantified
and if computationally efficient solutions are sought. 
The content of this subsection is a synthesis of known results; as such we do not claim
any credit of originality. We choose to include it here because of our goals may be different
from other sources. Also, efficient constructions of 
a component (list-decodable
codes) known to be useful for IR long ago only became known more recently. The IR protocol presented here
follows a well-known framework (e.g., as described in~\cite{Smith:optimal}), but
will use the latest tools, some known after~\cite{Smith:optimal}. Thus, to the best of our knowledge,
no other sources have put these known facts together.
 
We summarize our goals for IR. First, we would like to succeed {\em whenever}
the differences (referred to as errors) are bounded away from the above by $1/2$-fraction.
We hope that the solution is efficient, not just in term of computational complexity,
but also, most critically, the bits communicated, as well as the number of shared random bits used.
This is because any bit communicated in this stage
will be subtracted from the min-entropy guarantee, and that our goal is to achieve
secure quantum key distribution with a short seed. We note that in the literature, the issue
of computational efficiency and the amount of share randomness were often not considered, or 
were considered under a different set of assumptions (e.g., \cite{R05,TomamichelL}).

We define a quantity to describe the limit
of surviving fraction of min-entropy. 

\begin{definition}[Efficient Information Reconciliation] 
Let $\lambda\in(0,1/2)$, $\epsilon\in(0,1)$, $N$, $R$ and $M$ be integers, and $T$ be
a function on $N$, $\lambda$ and $\epsilon$. 
An {\bf information reconciliation protocol} with those parameters is a 
communication protocol between two parties Alice and Bob with the following property.
On any $N$-bit strings $A$ and $B$, known to Alice and Bob, respectively, 
they start the protocol with a shared $R$-bit string, communicate $M$ bits,
and finally Bob outputs an $N$-bit string $\tilde A$. If $A=\tilde A$, the protocol
{\em succeeds}; otherwise it fails. For all $A$ and $B$ of Hamming distance $|A\oplus B|\le (1/2-\lambda) N$, the probability of failure is
\begin{equation}
\Prob[A\ne \tilde A] \le\epsilon.
\end{equation}
\YtalkingPoint{changed the def so that the protocol will run on all input $A$ and $B$ but the performance guarantee applies to
when $A$ and $B$ are close.}
The computation complexity of the protocol is
$\le T$.

The protocol is said to be {\bf efficient} if for a constant $\lambda$,
$R=O(\log (N/\epsilon))$, $M\le(1-c)N + O(\log(1/\epsilon))$ for some constant $c=c(\lambda)$,
and $T=\textrm{poly}(N, \log(1/\epsilon))$.
\label{def:eic}
\end{definition}

The key ingredient in the protocol is to use binary linear error-correcting codes.
When the relative error is $<1/4$, one can use a uniquely decodable code, as shown
by~\cite{BBCS}. Otherwise, there is no binary code with  a constant rate, by the Plotkin bound. Thus
we will have to resort to list-decodable binary linear codes. A folklore approach for
pinning down the actual error from the decoded list is to use hashing. Here we use
approximate universal hashing. Explicit constructions of all these three tools
are known and are summarized below.

\begin{theorem}[Corollary of Theorem 5 in~\cite{GuruswamiIRCodes}]\label{thm:GuruswamiIndy}
For any $\lambda\in(0, 1/4)$, there exists
a family of binary linear codes with a relative error $1/4-\lambda$, a rate $\Omega(\lambda^3)$,
and linear time complexity for encoding and decoding.
\end{theorem}

\begin{theorem}[\cite{GR08Journal} (Theorem 5.3 and Remark 5.2)]\label{thm:GR}
For any $\lambda\in(0,1/2)$, and for an infinite number of integers $N>0$, there exists
a binary linear code of block length $N$, relative error $1/2-\lambda$, rate $\Omega(\lambda^3)$,
that can be list-decoded into a list of size $N^{\tilde O(\log 1/\lambda^3)}$ with $O(N^{O(1/\lambda^4)}))$
encoding and decoding time.  
\end{theorem}

\begin{definition}[Approximate Universal Hash Functions]
A set $H$ of functions $h: U\to V$
is a $\epsilon$-Universal Hash Function (-UHF) family if for all $u, u'\in U$, $u\ne u'$,
\begin{align}
\Prob_{h\in H}[h(u)=h(u')]\le \epsilon.
\end{align}
\end{definition}

It is well known that good approximate UHF exists. A standard construction is the following (see, e.g., \cite{BonehS:book}).
Let $\mathbb{F}_p$ be a finite field of size $q$, $U=\mathbb{F}_q^\ell$, $H=V=\mathbb{F}_p$, where each $k\in H$ is identified
with the function $h_k$
\begin{align}
h_k: (a_{\ell-1}, a_{\ell-2},\cdots, a_0) \mapsto k^{\ell} + a_{\ell-1}k^{\ell-1}+\cdots+a_1k+a_0.
\end{align}
Clearly if $(a'_{\ell-1},a'_{\ell-2},\cdots, a'_0)\ne (a_{\ell-1}, a_{\ell-2}, \cdots, a_0)$,
\begin{align}
\Prob_{k\in H} [ h_k(a'_{\ell-1},a'_{\ell-2},\cdots, a'_0)=h_k(a_{\ell-1}, a_{\ell-2}, \cdots, a_0)] \le \ell/p.
\end{align}
Thus $H$ is an $\ell/p$-UHF with $|U|=p^{\ell}$, $|V|=|H|=p$.

We will use an approximate UHF of the following parameters.
\begin{proposition}\label{prop:UHF}
For all sufficiently large integer $N$ and any $\epsilon \ge \frac{1}{4} 2^{-N}$, there exists an explicit $\epsilon$-UHF 
from $\{0, 1\}^N\to\{0, 1\}^n$ of size $2^n$, where $n=\left\lceil\log\left(\frac{N}{\epsilon}/\log\frac{N}{\epsilon}\right) \right\rceil +2$. 
\end{proposition}
\begin{proof}
In the construction described above, use the finite field of size $2^n$ and set $\ell = \lfloor \epsilon 2^n\rfloor$. We need only
to check that $n,\ell\ge1$ and $n\ell \ge N$, which is indeed the case.
\end{proof}

We are ready to present our protocol for  Efficient Information Reconciliation and prove its correctness.
\begin{figure}
\begin{center}
\setlength{\fboxsep}{20pt}
\setlength{\fboxrule}{1pt}
\fbox{\parbox{5.75in}{
\textit{Arguments:}
\begin{enumerate}[topsep=3pt]
\item[$\lambda:$] A real constant $\in (0,1/2]$.
\item[$X, Y:$] Binary strings of length $N$ such that $|X\oplus Y| \le (1/2-\lambda)N$.
\item[$\epsilon:$] A failure probability. Can be $0$ if $\lambda\in(1/4, 1/2]$.
\item[$A:$] The check matrix of an explicitly constructible (i.e. encoding and decoding in polynomial time) binary linear error-correcting code $C$ of length $N$, relative error $1/2-\lambda$, and
a linear rate $R=R(\lambda)$. The code $C$ is uniquely decodable if $\epsilon=0$.
Such code exists (e.g.,~\cite{AlonCodes}, \cite{GuruswamiIRCodes}) with $R(\lambda)=\Omega((\lambda-1/4)^3)$).
If $C$ is list-decodable code, the list size $L=L(N,\lambda)=N^{O(1)}$. Such a code exists (e.g, with $R(\lambda)=\Omega(\lambda^3)$ as in~\cite{GuruswamiR06}).
\item[$\mathcal{H}:$] If $C$ is list-decodable, let $\epsilon':=\epsilon/L$.  $\mathcal{H}$ is an explicit $\epsilon'$-UHF
from $\{0,1\}^N$ to $\{0,1\}^k$ of size $2^k$, where $k=\left\lceil \log\left(\frac{N}{\epsilon'}/\log\frac{N}{\epsilon'}\right)\right\rceil+2 = \log(1/\epsilon) + O(\log N)$. Such $\mathcal{H}$ exists according to Proposition~\ref{prop:UHF}.
\end{enumerate} 

\medskip

\textit{Protocol:}

\begin{enumerate}[topsep=3pt]
\item Alice sends Bob $AX \in\{0,1\}^{(1-R)N}$.
\item If $C$ is uniquely decodable, Bob computes the error syndrome $AY+AX=A(X+Y)$,
runs the decoding algorithm to obtain the unique $D$ with $|D|\le (1/2-\lambda)N$ and $AD=A(X+Y)$.
The protocol terminates with Bob outputting $Y+D$.
\item Otherwise ($C$ is list-decodable with list size $L$), Bob list-decodes from $A(X+Y)$ to obtain
a list $\{\Delta_1, \Delta_2,\cdots, \Delta_L\}$, where by the property of $C$, $X+Y=\Delta_i$, for some $i$, $1\le i\le L$.
\item Alice and Bob draw a random $h\in\mathcal{H}$, and Alice sends Bob $h(X)$. Bob checks
if there exists a unique $\Delta_i$ such that $h(Y+\Delta_i)=h(X)$. If yes,
Bob outputs $Y+\Delta_i$; otherwise he outputs $Y$.
\YtalkingPoint{change the protocol so that it won't abort. This simplifies the soundness analysis.}
\end{enumerate}
}}
\caption{Protocol EIR: an Efficient Information Reconciliation protocol}
\label{fig:eir}
\end{center}
\end{figure}

One may note that in the final step of Protocol EIR,  Bob could alternatively abort when there is no unique
$\Delta_i$ such that $h(Y+\Delta_i)=h(X)$. For technical convenience, our definition of Efficient Information Reconciliation does not allow abort. But it can be easily modified to allow aborting, and resulting performance parameters will be similar.
\YtalkingPoint{Added the above re. aborting in EIR.}

\begin{proposition}\label{prop:eir} The Protocol EIR in Fig.~\ref{fig:eir}
is an Efficient Information Reconciliation protocol~(Definition~\ref{def:eic}).
If $\lambda>1/4$ and a uniquely decodable code $C$ is used, no randomness is needed and the protocol
succeeds with certainty.
\end{proposition}

\begin{proof}
The length of Alice's message, the correctness of Bob's output, and the computational complexities follow from the properties of the error-correcting code~(Theorems~\ref{thm:GuruswamiIndy} and \ref{thm:GR}).
For the case of $\lambda\le 1/4$, the length of the shared randomness follows from the property of $\mathcal{H}$.
To analyze the failure probability, first observe that under the assumption that $|X+Y|\le (1/2-\lambda)N$,
$X+Y=D_i$ for some $i$.
Thus the chance of failure is precisely the existence of $i'\ne i$
such that $h(Y+D_i)=h(Y+D_{i'})$. This probability is no more than
$L\epsilon' =\epsilon$, as desired.
\end{proof}

We remark that for an Efficient Information Reconciliation protocol, 
there may be a tradeoff between the communication cost and the randomness used. 
For example, when the error rate $1/2-\lambda <1/4$, using a uniquely decodable code from Theorem~\ref{thm:GuruswamiIndy}
avoids the use of randomness but $c(\lambda)=O((\lambda-1/4)^3)$. If one uses the list-decodable code from Theorem~\ref{thm:GR},
the rate may be higher at the cost of some randomness.

\subsection{The Security of Protocol $R_{\textrm{kd}}$} 
We are now ready to prove our main result for untrusted-device QKD.
\begin{proofof}{Corollary~\ref{co:qkd}} 
We set $r_G$ to be the supremum of reals $R$ such that for some $\lambda<{\mathbf w}_G-1/2$, there exists an infinite family
of explicit list-decodable\footnote{We require that the size of the list is polynomial in the block length.}
binary linear codes of rate $R$ and relative error $1/2-\lambda$. 
By using the list-decodable code from Theorem~\ref{thm:GR}, $r_G=\Omega(({\mathbf w}_G-1/2)^3)>0$. 

We will show that any $r<r_G$ can be achieved. 
The proof for completeness is a standard application of concentration inequalities thus we leave the proof for the interested reader.
We shall focus on proving the soundness.

Let $\delta=r_G-r$.
Let $\lambda\in(0,{\mathbf w}_G-1/2)$ be such that there exists an Efficient Information Reconciliation protocol $P_{\textrm{EIR}}$ with $c(\lambda) \ge r_G- \delta/3$. Such $\lambda$ and $P_{\textrm{EIR}}$ exist by the definition
of $r_G$ and Proposition~\ref{prop:eir}.

Applying Theorem~\ref{mainthmprelim} with the $\delta$ parameter there set to be
$\delta/3$, we get the constants $K, b, q_0$ and $\eta_0$.
Let $\eta\le\eta_0$ and $q\le q_0$ so that Theorem~\ref{mainthmprelim} applies. Further
assume that $\eta$ is small enough so that Lemma~\ref{lm:qkderr}  also applies.

To prove Condition (a) of the soundness definition (subsection~\ref{subsec:qkdDefs}), note that 
by  Theorems~\ref{mainthmprelim}, the $SAE$-state (where $S$ is the aborting decision bit, and $E$ is the adversary's system) 
before information reconciliation has $(1- \delta/3)N$ extractable
bits with soundness error $\epsilon'_{s}:=K\exp(-bqN)$.
By definition, $P_{\mathrm{EIR}}$ communicates $\le(1-c(\lambda))N + O(\log 1/\epsilon)$ bits. Thus
the yield in $A$ after information reconciliation is  at least
\begin{flalign}
&\left[(1-\delta/3)-(1-c(\lambda))\right]\cdot N - O(\log1/\epsilon)\\
=& \left[c(\lambda)-\delta/3\right]\cdot N - O(qN)\\
\ge& (r_G-2\delta/3 -O(q))N.
\end{flalign}
If necessary, we lower the upper-bound for $q$ so that in the above, $O(q)\le \delta/3$. Thus
the yield in $A$ is at least $(r_G-\delta) N=rN$.
Since $P_{\textrm{EIR}}$ does not abort or change $AE$, the final state when restricted to $SAE$ remains unchanged,
thus is $\epsilon'_{s}$-close to a mixture of an aborting state and a state where $A$ has $rN$ extractable bits.

To satisfy soundness condition (b) (see subsection~\ref{subsec:qkdDefs}), we now bound the probability of the event 
$E_{\neq}$ that the protocol does not abort and Alice and Bob's keys ($A$ and $\tilde A$) disagree.
That is, with $P$ being the passing event (\ref{eqn:EventPassing}),
\begin{align}
E_{\neq}:=\left(P\wedge (A\ne \tilde A)\right).
\end{align}

Recall that $A$ and $B$ are the raw keys before $P_{\textrm{EIR}}$.
Denote by $\Delta$ the event that $|A+B| < (1/2-\lambda) N$.
Let $C_G$ be the event that the number of game rounds is $\le (1/2+\lambda) q N$. By the Chernoff bound,
with $\gamma:=\frac{(1/2-\lambda)^2}{2}$,
\begin{align}\label{eqn:C_G}
\Prob[C_G] \le \exp(-\gamma qN).
\end{align}

Let events $P$, $M$, $E$ be defined as in Lemma~\ref{lm:qkderr}, which we now apply
with the above $\lambda$.
Note that $(\bar M\wedge \bar C_G)$ implies $\Delta$,
because by construction, the raw key bits for game rounds always agree.

We now upper-bound $\Prob[E_{\neq}]$.
\begin{align}\label{eqn:sound_err_agreement}
\Prob[E_{\neq}] &\le \Prob[E] + \Prob[C_G] + \Prob[E_{\neq}\wedge \bar E\wedge \bar C_G]\\
&= \Prob[E] + \Prob[C_G] + \Prob[P \wedge (A\neq \tilde A) \wedge \bar C_G \wedge \bar M]\\
&\le \Prob[E] + \Prob[C_G] +\Prob[(A\neq\tilde A) \wedge \Delta]\\
&\le \Prob[E]+ \Prob[C_G] +\Prob[(A\neq\tilde A) |\Delta].
\end{align}

Applying Lemma~\ref{lm:qkderr}, equation (\ref{eqn:C_G}), and definition of Efficient Information Reconciliation,
the above is upper-bounded by
\begin{align}
\exp(-\alpha qN) + \exp(-\beta N) + \exp(-\gamma qN) + \exp(-qN).
\end{align}
Thus setting $b':=\min\{b, \alpha, \beta, \gamma, 1\}$ and $K'=\max\{K, 5\}$, we have 
that the soundness error is
\begin{align}
K'\exp(-b'qN)=\exp(-\Omega(qN)+O(1)),
\end{align}
thus proving the soundness result.

The number of random bits used in the expansion protocol is $O(Nh(q))$, and the number used
in $P_{\textrm{EIR}}$  is $O(\log N/\epsilon)$, where
$\epsilon=\exp(-qN)$. This gives a total
of $O(Nh(q) + \log N + qN)$ random bits, which is $O(Nh(q)+\log N)$ (or simply $O(Nh(q))$ when $qN=\Omega(1)$).

We leave the claims on the instantiation to the reader.
\end{proofof}

\section{Further Directions}

A natural goal at this point is to
improve the certified rate of Protocol R.
This is important for the practical realization of our protocols.
By the discussion in section~\ref{REUDsec}, this reduces to two simple questions.
First, what techniques are there for
computing the trust coefficient $\mathbf{v}_G$
of a binary XOR game?  Second,
is it possible to reprove
Theorem~\ref{uncertaintythm}
in such a way that the limiting function $\pi(x )$
becomes larger?
A related question is to improve the key rate of Protocol $R_{\textrm{kd}}$.
The ``hybrid'' technique of Vazirani and Vidick~\cite{Vazirani:fully} for mixing the CHSH game
with a trivial game with unit quantum winning strategy may extend to general
binary XOR games.

It would also be interesting to explore whether Theorem~\ref{mainthmprelim}
could be extended to nonlocal games outside the class of strong self-tests.
Such an extension will not only facilitate the realization of those protocols, but also
will further identify the essential feature of
quantum information enabling those protocols. As the characterization of
strong self-tests is critical for our proof, developing a theory of
robust self-testing beyond binary XOR games may be useful for 
our question. It is also conceivable that there exist fairly broad conditions
under which a classical security proof, which is typically much easier to establish,
automatically imply quantum security. We consider identifying such a wholesale security lifting
principle as a major open problem.

A different direction to extend our result is to prove security
based on physical principles more general than quantum mechanics, such as
non-signaling principle, or information causality~\cite{IC}.

Our protocols require some initial perfect randomness to start with. The Chung-Shi-Wu protocol~\cite{CSW14}
relaxes this requirement to an arbitrary min-entropy source and tolerates a universal constant level of
noise. However, those were achieved at a great cost on the number of non-communicating devices.
Another major open problem is whether our protocol can be modified to handle non-uniform input.

Randomness expansion can be thought of as a ``seeded'' extractions of randomness from untrusted quantum devices,
in the sense of Chung, Shi, and Wu~\cite{CSW14}.
Our one-shot and unbounded expansion results demonstrate a tradeoff between the seed length
and the output length different from that in classical extractors.
Recall that 
a classical extractor with output length $N$ and error parameter $\epsilon$ requires $\Omega(\log N/\epsilon)$
seed length, while our unbounded expansion protocol can have a fixed seed length (which determines the error parameter). 
What is the maximum amount of randomness one can extract from a device of a given amount of
entanglement (i.e. is the exponential rate optimal for one device)?
What can one say about the tradeoff between expansion rate and some proper quantity describing
the communication restrictions? Answers to those questions will reveal fundamental features
of untrusted quantum devices as a source for randomness extraction, and will hopefully
lead to an intuitive understanding of where the randomness comes from.

Yet another important direction forward is to prove security in more complicated composition scenarios
than the cross-feeding protocol. As pointed out by \cite{BCK:memory},
a device reused may store previous runs' information thus potentially may cause security problem
in sequentially composed QKD protocols. While such ``memory attack'' appears not to be 
a problem for sequential compositions of our randomness expansion protocol, it may for other more complicated compositions. 
Thus it is desirable to design untrusted-device protocols and prove their security under broader classes of compositions.

\section{Acknowledgments}
We are indebted to Anne Broadbent, Kai-Min Chung, Roger Colbeck, Brett Hemenway, Adrian Kent,
Christopher Portmann, 
Thomas Vidick, Ilya Volkovich, Xiaodi Wu, and Andrew Yao for useful discussions, and to Michael Ben-Or, Qi Cheng, Venkatesan Guruswami, and Adam Smith
for pointers to the literature on error-correcting codes
and information reconciliation.

\appendix

\section{Supplementary Material}

\subsection{The Canonical Form for Two Binary Measurements}

\begin{theorem}
\label{canintermed}
Let $V$ be a finite dimensional $\mathbb{C}$-vector space
and let $X_0, X_1$ be Hermitian operators on $V$ satisfying
$\left\| X_0 \right\| , \left\| X_1 \right\| \leq 1$.  Then,
there exists a unitary embedding $U \colon W \to \mathbb{C}^{2n}$,
$n \geq 1$, and operators $Y_0, Y_1$ of the form
\begin{eqnarray}
\label{canonicalrepeat}
Y_0   =   \left[ \begin{array}{ccccccc}
0 & 1  \\
1 & 0  \\
&& 0 & 1 \\
&& 1 & 0 \\
& & & & \ddots \\
&& & & & 0 & 1 \\
&& & & & 1 & 0 \\
\end{array} \right] & \hskip0.6in  &
Y_1  =   \left[ \begin{array}{ccccccc}
0 & \zeta_1  \\
\overline{\zeta_1}  & 0 \\
&&0 & \zeta_2  \\
&&\overline{\zeta_2}  & 0 \\
 &&& & \ddots \\
& &&& & 0 & \zeta_{m_j} \\
& & &&& \overline{\zeta_{m_j}} & 0 \\
\end{array} \right]
\end{eqnarray}
with $\left\| \zeta_k \right\| = 1$, such that
$X_k = U^* Y_k U$ for $k \in \{ 0, 1 \}$.
\end{theorem}

We prove this theorem by a series of lemmas.
Consider the class of all triples $(V, X_0, X_1)$ satisfying
the condition from the first sentence of Theorem~\ref{canintermed}.
Consider the following two conditions on such triples:
\begin{enumerate}
\item[(A)] The operators $X_k$ satisfy $X_k^2 = \mathbb{I}$.
\item[(B)] The vector space $V$ is equal to $\mathbb{C}^{m}$,
and $X_0, X_1$ have a uniform diagonal block form:
\begin{eqnarray*}
X_k   =   \left[ \begin{array}{cccccccccc}
B_k^1 \\
& B_k^2 \\
&& \ddots \\
&&& B_k^r \\
&&&& b_k^1 \\
&&&&& b_k^2 \\
&&&&&& \ddots \\
&&&&&&& b_k^s
\end{array} \right] 
\end{eqnarray*}
where $2r + s = m$, $b_k^j \in \{ -1 , +1 \}$ and each $B_k^j$
is a $2 \times 2$ Hermitian matrix with eigenvalues 
$+1$ and $-1$.
\end{enumerate}

\begin{lemma}
Any triple $(V, X_0, X_1)$ satisfying the conditions of Theorem~\ref{canintermed}
has a unitary embedding into a triple satisfying condition (A).
\end{lemma}

\begin{proof}
Let $U \colon V \to V \oplus V$ be given by $U ( v ) = v \oplus 0$, and let
$\{ X'_k \mid k = 0, 1 \}$ be the operators on $V \oplus V$ defined by
\begin{eqnarray}
X'_k & = & \left[ \begin{array}{c|c}
X_k & \sqrt{ \mathbb{I} - X_k^2 } \\
\hline
\sqrt{ \mathbb{I } - X_k^2 } & - X_k \end{array} \right].
\end{eqnarray}
It is easily checked that $(X'_k)^2 = \mathbb{I}$.
\end{proof}

\begin{lemma}
Any triple $(V, X_0, X_1)$ satisfying condition (A) has a unitary embedding
into a triple satisfying  condition (B).
\end{lemma}

\begin{proof}
We can choose an orthonormal basis $\{ v_1, \ldots, v_{\dim V} \}$ for $V$ such that $X_0$ has the form
\begin{eqnarray}
X_0 & = & \left[ \begin{array}{c|c}
\mathbb{I}_n & 0  \\
\hline
0 & - \mathbb{I}_m \end{array} \right].
\end{eqnarray}
where $\mathbb{I}_r$ denotes the $r \times r$ identity matrix.
By an appropriate unitary transformation of $V$ that respects
this block structure, we obtain another orthonormal basis
$\{ v'_1 , \ldots, v'_{\dim v} \}$ such that $X_0$ and $X_1$
have the form
\begin{eqnarray}
X_0 = \left[ \begin{array}{c|c}
\mathbb{I}_n & 0  \\
\hline
0 & - \mathbb{I}_m \end{array} \right] \hskip0.3in
\textnormal{ and } \hskip0.3in
X_1 = \left[ \begin{array}{c|c}
A & D  \\
\hline
D^* & C \end{array} \right], 
\end{eqnarray}
where $A$ and $C$ are diagonal matrices.  The
condition $X_1^2 = \mathbb{I}$ implies that
$A^2 + D D^* = \mathbb{I}$ and $D^* D + C^2
= \mathbb{I}$.  Since both $D D^*$ and $D^* D$
are diagonal, $D$ is diagonal.  Reordering
the bases yields the desired form.
\end{proof}

\begin{lemma}
Any triple $(V, X_0, X_1)$ satisfying condition (B) 
has a unitary embedding into a triple of the form
(\ref{canonicalrepeat}).
\end{lemma}

\begin{proof}
It suffices to prove the lemma for the case where 
$X_0, X_1$ are both scalars, and the case where $X_0, X_1$
are each $2 \times 2$ Hermitian matrices with eigenvalues $+1$ and $-1$.
The first case is easy and is left to the reader.  For the second
case, we can find an orthonormal basis $\{ v_1, v_2 \}$ for $\mathbb{C}^2$
under which $X_0 = \left[ \begin{array}{cc} 0 & 1 \\ 1 & 0 \end{array} \right]$,
and then find a basis of the form $\{ (\cos \theta) v_1 + i (\sin \theta) v_2, 
(\cos \theta) v_2 + i (\sin \theta) v_1 \}$ with
$z \in \mathbb{C}, |z | = 1$ under which $X_1$ is an antidiagonal matrix.
\end{proof}

This completes the proof of Theorem~\ref{canintermed}.

\subsection{Smooth Min-entropy and Renyi Divergence}

This subsection provides supporting proofs for section~\ref{renyisection}.

\begin{proposition}
\label{divergenceprop}
Let $\alpha \in (1, 2]$.
Let $\rho$ be a density operator on 
a finite-dimensional Hilbert space $V$, and let
$\sigma$ be a positive semidefinite operator on $V$ such that
$\textnormal{Supp } \sigma \supseteq
\textnormal{Supp } \rho$.  Then, there exists a positive semidefinite
operator $\rho'$ such that $\rho' \leq \sigma$ and
\begin{eqnarray}
\log \left\| \rho- \rho' \right\|_1 \leq \frac{\alpha - 1 }{2} \cdot D_\alpha
( \rho \| \sigma ) + \frac{1}{2}
\end{eqnarray}
\end{proposition}

\begin{proof}
Our proof is based on 
the proof of Lemma 19 in \cite{DupuisFS:2013} (which, in turn, is based on \cite{TomamichelCR:2009}).  For any Hermitian operator $H$, let
$P_H^+$ denote projection on the subspace spanned by the positive eigenvectors
of $H$, and let $\Tr^+ ( H ) = \Tr ( P_H^+ H P_H^+ )$.   Let
\begin{eqnarray}
\delta & = & \Tr^+ ( \rho - \sigma ).
\end{eqnarray}
Note that, by the construction from the proof of Lemma 15 in \cite{TomamichelCR:2009},
there must exist a subnormalized operator $\rho'$ such that $\rho' \leq \sigma$ 
and $\left\| \rho' - \rho \right\|_1 \leq \sqrt{2 \delta}$.  

Let $P = P^+_{\rho - \sigma}$, and let $P^\perp$ denote the complement
of $P$.  Note
that by applying the data processing inequality for $D_\alpha$ (see Theorem 
5 in \cite{MullerDSFT:2013}) to the quantum operation
$X \mapsto \left| 0 \right> \left< 0 \right| \otimes P X P
+ \left| 1 \right> \left< 1 \right| \otimes P^\perp X P^\perp$, we have
\begin{eqnarray}
D_\alpha ( \rho \| \sigma ) & \geq &
\frac{1}{\alpha - 1} \log \left( \Tr \left[ \left( ( P \sigma P )^{\frac{1 - \alpha}{2 \alpha}} (P \rho P ) ( P \sigma P )^{\frac{1 - \alpha}{2 \alpha}} \right)^\alpha \right. \right. \\
& & +
\left. \left.\left( ( P^\perp ( \sigma ) P^\perp )^{\frac{1 - \alpha}{2 \alpha}} (P^\perp \rho P^\perp ) ( P^\perp \sigma P^\perp )^{ \frac{1 - \alpha}{2 \alpha}} \right)^\alpha \right] \right) \\
& \geq &
\frac{1}{\alpha - 1} \log \left( \Tr \left[ \left( ( P \sigma P )^{\frac{1 - \alpha}{2 \alpha}} (P \rho P ) ( P \sigma P )^{\frac{1 - \alpha}{2 \alpha}} \right)^\alpha \right] \right) 
\end{eqnarray}
Let $\overline{\sigma} = P \sigma P$ and $\overline{\rho} = P \rho P$.   We have
the following.
\begin{eqnarray}
D_\alpha ( \rho \| \sigma ) & \geq & \frac{1}{\alpha - 1} \log \left( \Tr \left[
\left( \overline{\sigma}^{\frac{1 - \alpha}{2 \alpha}} \overline{\rho}
\overline{\sigma}^{\frac{1 - \alpha}{2 \alpha}} \right) \left( \overline{\sigma}^{\frac{1 - \alpha}{2 \alpha}} \overline{\rho}
\overline{\sigma}^{\frac{1 - \alpha}{2 \alpha}} \right)^{\alpha - 1} \right] \right)
\end{eqnarray}
Note that
$\overline{\rho} \geq \overline{\sigma}$ by construction, and 
$Z \mapsto Z^{\alpha - 1 }$ is a monotone function (see part (a) of Proposition~\ref{powmatrixprop}).  Therefore we have the following.
\begin{eqnarray}
D_\alpha ( \rho \| \sigma )
& \geq & \frac{1}{\alpha - 1} \log \left( \Tr \left[
\left( \overline{\sigma}^{\frac{1 - \alpha}{2 \alpha}} \overline{\rho}
\overline{\sigma}^{\frac{1 - \alpha}{2 \alpha}} \right) \left( \overline{\sigma}^{\frac{1 - \alpha}{2 \alpha}} \overline{\sigma}
\overline{\sigma}^{\frac{1 - \alpha}{2 \alpha}} \right)^{\alpha - 1} \right] \right) \\
& \geq & \frac{1}{\alpha - 1} \log \left( \Tr \left[
\left( \overline{\sigma}^{\frac{1 - \alpha}{2 \alpha}} \overline{\rho}
\overline{\sigma}^{\frac{1 - \alpha}{2 \alpha}} \right) \overline{\sigma}^\frac{\alpha-1}{\alpha} \right] \right) \\
& \geq & \frac{1}{\alpha - 1} \log \left( \Tr \left[
\overline{\rho} \right] \right) \\
& \geq & \frac{1}{\alpha - 1} \log \delta 
\end{eqnarray}
where in the last line we used the fact that  $\Tr ( \overline{\rho} ) \geq \Tr (\overline{\rho} - \overline{\sigma} )
= \delta$.
Let $\rho'$ be a positive semidefinite operator satisfying
$\rho' \leq \sigma$ 
and $\left\| \rho' - \rho \right\|_1 \leq \sqrt{2 \delta}$.  Then we have
\begin{eqnarray}
D_\alpha ( \rho \| \sigma )
& \geq & \frac{1}{\alpha - 1} \log \left( \left\| \rho' - \rho \right\|_1^2 /2 \right), 
\end{eqnarray}
which implies the desired result.
\end{proof}

\begin{proposition}
\label{altdivergenceprop}
Suppose that in Proposition~\ref{divergenceprop},
$V$ is the state space of a bipartite quantum system $AB$,
and $\rho, \sigma$ are classical-quantum operators.\footnote{That is,
$A$ is a classical register and
$\rho , \sigma$ have the form $\rho = \sum_i \left| a_i \right> \left<
a_i \right| \otimes \rho_i$ and $\sigma = \sum_i \left| a_i \right>
\left< a_i \right| \otimes \sigma_i$ where $\{ a_1, \ldots, a_n \}$
is a standard basis for ${A}$.}  Then,
there exists an operator $\rho'$ satisfying the conditions
of Proposition~\ref{divergenceprop} such
that $\rho'$ itself is a classical-quantum operator.
\end{proposition}

\begin{proof}
This is an easy consequence
of the construction for
$\rho'$ (from the proof of Lemma 15 in \cite{TomamichelCR:2009})
which was used in the proof of Proposition~\ref{divergenceprop}.
\end{proof}

\begin{proof}[of Proposition~\ref{maxrenyiprop}]
Let $\lambda$ be the quantity on the right side of inequality (\ref{adivbound}).
We have
\begin{eqnarray}
D_\alpha ( \rho \| 2^{-\lambda} \sigma ) & = & \frac{2 \log \epsilon - 1}{
 \alpha - 1 }.
\end{eqnarray}
By Proposition~\ref{divergenceprop}, we can find a positive semidefinite operator
$\rho' \leq 2^{-\lambda} \sigma$ such that
\begin{eqnarray}
\left\| \rho' - \rho \right\|_1
\leq \epsilon.
\end{eqnarray}
The result follows from the definition of $D_{max}^\epsilon$.
\end{proof}

\subsection{Variables and Functions Used in Section~\ref{partialtrustapp}}

\label{varfuncapp}

In this subsection we collect together the variables in functions 
that are used in the proof of security for Protocol $A'$.  We include
also the assertions about the limits of the functions.  (This is
intended just for the reader's convenience --- all these statements
are included in the body of the paper.)

\vskip0.2in

\textbf{Variables:}

\begin{center}
\begin{tabular}{| c c l | l | }
\hline
$N$ & $\in$ & $\mathbb{N}$ & number of rounds \\
\hline
$q$ & $\in$ & $( 0 , 1)$ & test probability \\
\hline
$t$ & $\in$ & $[ 0 , 1]$ & failure parameter \\
\hline
$v$ &$\in$ & $(0, 1]$
&  trust coefficient \\
\hline
$h$ & $\in$ & $[ 0, 1 - v]$ &
coin flip coefficient \\
\hline 
$\eta$ & $\in$ & $( 0, v/2)$ & error tolerance \\
\hline
$\kappa$ & $\in $ & $(0 , \infty)$ & failure penalty \\
\hline
$r$ & $\in$ & $( 0, 1/(q \kappa ) ]$ & multiplier for R{\'e}nyi coefficient \\
\hline
$\epsilon$ & $\in$ & $(0, \sqrt{2} ]$ & error parameter for smooth min-entropy \\
\hline
\end{tabular}
\end{center}

\textbf{Functions:}

\vskip0.2in

Note that the functions $\gamma ( q, \kappa, r )$ and $\mathbf{r} ( v, \eta,
q, \kappa )$ defined below are written simply as $\gamma$ and $\mathbf{r}$.

\vskip0.2in

\fbox{\parbox{5.75in}{

\begin{eqnarray*}
\gamma ( q, \kappa, r ) & := & q \kappa r
\end{eqnarray*}

\begin{eqnarray*}
\lim_{(q, \kappa) \to (0, 0)} \gamma ( q, \kappa, r ) & = & 0
\end{eqnarray*}
}}

\vskip0.2in

\fbox{\parbox{5.75in}{

\begin{eqnarray*}
\Pi ( \gamma, t ) & := & - \frac{1}{\gamma} \log \left\{ 2^{- \gamma} \left[ (1 - t)^\frac{1}{1 + 2 \gamma} +
t^{\frac{1}{1 + 2 \gamma}} \right]^{1 + 2 \gamma} \right\}
\end{eqnarray*}

\begin{eqnarray*}
\pi ( t ) & := 1 - 2 t \log \left( \frac{1}{t} \right)
- 2 (1-t) \log \left( \frac{1}{1-t } \right)
\end{eqnarray*}

\begin{eqnarray*}
\lim_{(q, \kappa, t) \to (0, 0, t_0 )} \Pi  ( \gamma, t ) & = & \pi ( t_0 )
\end{eqnarray*}

}}

\vskip0.2in

\fbox{\parbox{5.75in}{

\begin{eqnarray*}
\lambda ( v, h, q, \kappa, r, t ) & := & \left( (1 - q) 2^{- \gamma \Pi ( \gamma , t ) }
+ q \left\{ 1 - (1 - 2^{-\kappa} ) [(h/2)^{1+\gamma} + v^{1+\gamma} t ]
\right\} \right)^{1/\gamma}
\end{eqnarray*}

\begin{eqnarray*}
\Lambda ( v, h, q, \kappa, r, t ) & := & - \log ( \lambda ( v, h, q, \kappa, r, t ))
\end{eqnarray*}

\begin{eqnarray*}
\lim_{\substack{(q, \kappa, t ) \to  (0, 0, t_0)}} \Lambda ( v, h, q, \kappa, r, t )
& = & \pi ( t_0 ) + \frac{h /2 + v t_0}{r}
\end{eqnarray*}

}}

\vskip0.2in

\fbox{\parbox{5.75in}{

\begin{eqnarray*}
\Delta ( v, h, q, \kappa, r ) & := & \min_{s \in [0, 1 ]} \Lambda ( v, h, q, \kappa, r, s )
\end{eqnarray*}

\begin{eqnarray*}
\lim_{(q, \kappa ) \to (0, 0 )} \Delta ( v, h, q, \kappa, r)
& = & \min_{s \in [0,1]} \left( \pi ( s )+ \frac{ h/2 + vs}{r}  \right)
\end{eqnarray*}

}}

\vskip0.2in

\fbox{\parbox{5.75in}{

\begin{eqnarray*}
R ( v, h, \eta, q, \kappa, r ) &  := & -
\frac{h/2 + \eta}{r} + \Delta ( v, h, q, \kappa, r)
\end{eqnarray*}

\begin{eqnarray*}
\lim_{(q, \kappa) \to (0, 0)}
R ( v, h, \eta, q, \kappa, r )
& = & \min_{s \in [0, 1 ]} \left[ \pi ( s ) + \frac{vs - \eta}{r}
\right]
\end{eqnarray*}

}}

\vskip0.2in

\fbox{\parbox{5.75in}{

\begin{eqnarray*}
\mathbf{r}(v, \eta, q, \kappa) & := & \min \left\{ \frac{ v }{- \pi'(\eta / v )} , 
\frac{1}{q \kappa } \right\}
\end{eqnarray*}

\begin{eqnarray*}
\lim_{(q,\kappa) \to (0, 0)} \mathbf{r}(v, \eta, q, \kappa) & = &
\frac{ v }{- \pi'(\eta / v )}
\end{eqnarray*}

}}

\vskip0.2in

\fbox{\parbox{5.75in}{

\begin{eqnarray*}
T ( v, h, \eta,  q, \kappa) & := &
R ( v, h, \eta,  q, \kappa, \mathbf{r} )
\end{eqnarray*}

\begin{eqnarray*}
\lim_{(q, \kappa) \to ( 0, 0)}
T ( v, h, \eta, q, \kappa)
& = & \pi ( \eta / v )
\end{eqnarray*}

}}

\vskip0.2in

\fbox{\parbox{5.75in}{

\begin{eqnarray*}
F ( v, h, \eta, q, \kappa ) & := & \frac{2}{\mathbf{r}}
\end{eqnarray*}

\begin{eqnarray*}
\lim_{(q, \kappa) \to (0, 0)}
F ( v, h, \eta , q , \kappa ) & = &
\frac{- 2 \pi' ( \eta / v ) }{v}
\end{eqnarray*}
}}

\subsection{Mathematical Results}

\begin{proposition}
\label{etothecprop}
Let $U \subseteq \mathbb{R}^n$, let $\mathbf{z} \in \mathbb{R}^n$ be an element
in the closure of $U$, and let $f, g$ be continuous functions from $U$ to $\mathbb{R}$.
Suppose that
\begin{eqnarray}
\lim_{\mathbf{x} \to \mathbf{z} } f(\mathbf{x} ) & = & 0
\end{eqnarray}
and
\begin{eqnarray}
\lim_{\mathbf{x} \to \mathbf{z} } \frac{f ( \mathbf{x} ) }{g ( \mathbf{x} )  }
& = & c.
\end{eqnarray}
Then,
\begin{eqnarray}
\label{etothec}
\lim_{\mathbf{x} \to \mathbf{z} } ( 1 + f ( \mathbf{x} ) )^{1/g ( \mathbf{x} ) }
& = & e^c.
\end{eqnarray}
\end{proposition}

\begin{proof}
This can be proved easily by taking the natural logarithm of both sides of
(\ref{etothec}).
\end{proof}

\begin{proposition}
\label{compactprop}
Let $U \subseteq \mathbb{R}^n$ and $V \subseteq \mathbb{R}^m$,
and assume that $V$ is \textbf{compact}.  Let $f \colon U \times V \to 
\mathbb{R}$ be a continuous function.  Let $\mathbf{z} \in \mathbb{R}^n$ be an
element in the closure of $U$, and assume that $\lim_{(\mathbf{x}, 
\mathbf{y} ) \to (\mathbf{z} , \mathbf{y}_0 )}
f ( \mathbf{x} , \mathbf{y} )$ exists for every $\mathbf{y}_0 \in V$.
Then,
\begin{eqnarray}
\label{limitequiv}
\lim_{\mathbf{x} \to \mathbf{z}} \min_{\mathbf{y} \in V } f ( \mathbf{x} ,
\mathbf{y} ) & = & \min_{\mathbf{y} \in V }  \lim_{\mathbf{x} \to \mathbf{z}} f ( \mathbf{x} ,
\mathbf{y} ).
\end{eqnarray}
\end{proposition}

\begin{proof}
By assumption, there exists a continuous extension of $f$ to $(U \cup \{ \mathbf{z}
\} ) \times V$.  Denote this extension by $\overline{f}$.  Let $h ( \mathbf{x} , \mathbf{y} ) =
\overline{f} ( \mathbf{x} , \mathbf{y} ) - \overline{f} ( \mathbf{z} ,
\mathbf{y} )$.

Let $\delta > 0$.
For any $\mathbf{y} \in V$,
since $h ( \mathbf{z} , \mathbf{y} ) = 0$ and $h$ is continuous
at $(\mathbf{z} , \mathbf{y} )$,
we can find an $\epsilon_\mathbf{y} > 0$ such that the values
of $h$ on the cylinder
\begin{eqnarray}
\left\{ ( \mathbf{x} , \mathbf{y}' ) \mid
| \mathbf{x} - \mathbf{z} | < \epsilon_\mathbf{y} ,
| \mathbf{y}' - \mathbf{y} | < \epsilon_\mathbf{y} \right\}
\end{eqnarray}
are confined to $[- \delta , \delta ]$.  Since $V$ is compact,
we can choose a finite set $S \subseteq V$ such that the
the $\epsilon_y$-cylinders for $y \in S$ cover $V$.  Letting
$\epsilon = \min_{y \in S} \epsilon_y$, we find that the values
of $h$ on the $\epsilon$-neighborhood of $V$ are
confined to $[-\delta, \delta]$.  Therefore,
the minimum of $f ( \mathbf{x} , \mathbf{y} )$
on the $\epsilon$-neighborhood of $V$ is within $\delta$
of $\min_{\mathbf{y} \in V } \overline{f} ( \mathbf{z} ,
\mathbf{y} )$.  The desired equality (\ref{limitequiv})
follows.
\end{proof}

\label{renyismapp}

\bibliographystyle{apalike}
\bibliography{quantumsec}

\end{document}